\documentclass[12pt,oneside,reqno,titlepage, hyperfootnotes=false]{amsart}
\usepackage{lmodern}
\usepackage[T1]{fontenc}
\usepackage[latin9]{inputenc}
\usepackage[a4paper]{geometry}
\usepackage{color}
\usepackage{amstext}
\usepackage{amsthm}
\usepackage{amssymb}
\usepackage{graphicx}
\usepackage{setspace}
\usepackage[authoryear]{natbib}
\usepackage[unicode=true,pdfusetitle,
 bookmarks=true,bookmarksnumbered=false,bookmarksopen=false,
 breaklinks=false,pdfborder={0 0 0},pdfborderstyle={},backref=false,colorlinks=true]
 {hyperref}

\makeatletter

\providecommand{\tabularnewline}{\\}

\numberwithin{equation}{section}
\numberwithin{figure}{section}

\usepackage{amsfonts}
\usepackage{amstext}
\usepackage{amsthm}

\usepackage{hyperref}
\hypersetup{
    colorlinks = true,
    citecolor = black
}

\usepackage{threeparttable}

\usepackage[
singlelinecheck=false % <-- important
]{caption}

\DeclareMathOperator*{\argmin}{arg\,min}

\DeclareMathOperator*{\argmax}{arg\,max}

\usepackage{float}
\newtheorem{thm}{Theorem}

\newtheorem{lem}{Lemma}

\newtheorem*{asm1}{Assumption 1}
\newtheorem*{asm2}{Assumption 2}
\newtheorem*{asm3}{Assumption 3}
\newtheorem*{asm4}{Assumption 4}
\newtheorem*{asm5}{Assumption 5}

\theoremstyle{definition}

\makeatother

\begin{document}
\title{Temporal-Difference estimation of dynamic discrete choice models }
\author{Karun Adusumilli\textsuperscript{+} and Dita Eckardt\textsuperscript{*}}
\begin{abstract}
We study the use of Temporal-Difference learning for estimating the
structural parameters in dynamic discrete choice models. Our algorithms
are based on the conditional choice probability approach but use functional
approximations to estimate various terms in the pseudo-likelihood
function. We suggest two approaches: The first -linear semi-gradient-
provides approximations to the recursive terms using basis functions.
The second -Approximate Value Iteration- builds a sequence of approximations
to the recursive terms by solving non-parametric estimation problems.
Our approaches are fast and naturally allow for continuous and/or
high-dimensional state spaces. Furthermore, they do not require specification
of transition densities. In dynamic games, they avoid integrating
over other players\textquoteright{} actions, further heightening the
computational advantage. Our proposals can be paired with popular
existing methods such as pseudo-maximum-likelihood, and we propose
locally robust corrections for the latter to achieve parametric rates
of convergence. Monte Carlo simulations confirm the properties of
our algorithms in practice.
\end{abstract}

\thanks{\textit{This version}: \today{}. \\
\textsuperscript{+}Department of Economics, University of Pennsylvania;
akarun@sas.upenn.edu.\\
\textsuperscript{*}Department of Economics, University of Warwick;
dita.eckardt@warwick.ac.uk.\\
 We would like to thank Xiaohong Chen, Frank Diebold, Aviv Nevo, Whitney
Newey and Frank Schorfheide for helpful discussions.}
\maketitle

\section{Introduction \label{sec:Introduction}}

Recent years have seen a number of important developments in the field
of Reinforcement Learning (RL) for computation of value functions.
The goal of this paper is to study the use of a popular RL technique,
Temporal-Difference (TD) learning, for estimation and inference in
Dynamic Discrete Choice (DDC) models.

DDC models are frequently used to describe the inter-temporal choices
of \linebreak forward-looking individuals in a variety of contexts.
In these models, agents maximize their expected future payoff through
repeated choice amongst a set of discrete alternatives. Based on a
revealed preference argument, structural estimation proceeds by using
microdata on choices and outcomes to recover the underlying model
parameters.\footnote{See \citet{AguirregabiriaMira2010} for a detailed survey of the literature
on the estimation of DDC models.} A key challenge in this literature is the complexity of estimation.
Uncovering the structural parameters originally required an explicit
solution to a dynamic programming problem in addition to the optimization
of an estimation criterion. A key advance has been \citeauthor{HotzMiller1993}'s
\citeyearpar{HotzMiller1993} Conditional Choice Probability (CCP)
algorithm which avoids the repeated solution of the inter-temporal
optimization problem by taking advantage of a mapping between value
function differences and conditional choice probabilities. 

Unfortunately, the standard CCP algorithm is computationally infeasible
when the underlying states are continuous and/or the state space is
high-dimensional. Such state spaces are common in applications. One
frequently employed approach to tackle continuous state spaces is
through state space discretization, e.g., \citet{kalouptsidi2014time}
and \citet{almagro2019location} use aggregation and clustering methods
to do this. However, it is not always clear how to perform such a
discretization in practice, and moreover, it introduces bias into
the parameter estimates. An alternative is to employ functional approximations
for the value functions. For instance, \citet{barwick2015costs} and
\citet{kalouptsidi2018detection} use estimated transition densities
and numerical/analytical integration to approximate the value functions
using linear sieves and LASSO, respectively. However, the theoretical
properties of these methods when using machine learning methods (such
as LASSO) are as yet unknown, and they still require estimation of
transition densities, which is not straightforward, along with numerical
integration, which can be computationally expensive.\footnote{Yet another alternative is to use forward Monte Carlo simulations
(\citealp{BBL2007}, \citealp{HMSS1994}), but this again becomes
very involved as the number of continuous state variables or players
increases. The use of a finite number of Monte Carlo simulations also
adds bias to the estimates.} 

The aim of this paper is to develop tractable algorithms for CCP estimation
when the state variables are continuous and/or the state space is
large. Such algorithms should possess three properties: First, they
should be fast to compute even under high-dimensional state spaces.
Second, they should avoid state space discretization, and instead
rely on functional approximation of value functions. Third, they should
avoid estimation of transition densities which are difficult to parameterize
and estimate under continuous states.

In this paper, we suggest two methods, based on the idea of TD learning,
that satisfy all the above properties. The methods involve two different
techniques for estimating various recursive terms (which are akin
to value functions) that arise in CCP estimation. The first approach,
the linear semi-gradient method, provides functional approximations
to the recursive terms using basis functions. This approach simply
involves inverting a low-dimensional matrix, where the dimension is
the number of basis functions being used. Thus the computational cost
is trivial in most settings. Furthermore, it only requires the observed
sequences of current and future state-action pairs as input and estimation
of transition densities is not needed. The second approach, Approximate
Value Iteration (AVI), builds a sequence of approximations to the
value terms by solving a non-parametric estimation problem in each
step. Almost any machine learning (ML) method devised for prediction
can be used for approximation under this method, including (but not
limited to) LASSO, random forests and neural networks. To the best
of our knowledge, the AVI estimator is the first estimator for DDC
models that can be applied with any ML method. Hence, it naturally
allows for very high-dimensional state spaces. Again, no estimation
of transition densities is required. We derive the non-parametric
rates of convergence for estimation of the value terms under both
these methods. With the estimates of these functions at hand, estimation
of the structural parameters can proceed with standard methods such
as pseudo-maximum-likelihood estimation (PMLE, \citet{AguirregabiriaMira2002})
or minimum distance estimation. 

The focus of this paper is on the estimation of structural parameters.
To this end, our procedures allow one to avoid modeling state transitions.
Performing counterfactual analysis may still require estimating the
transition density, but we argue that our techniques remain useful,
even for this purpose, for two reasons: First, counterfactuals often
involve new transition densities which are different from the ones
that enter the estimation of the structural parameters, see \citet{kalouptsidi2018detection}
for an example. Our methods therefore avoid estimation of the original
transition densities that may not be needed for counterfactual analysis.
Second, with continuous states, decoupling the estimation of structural
parameters and transition densities may be beneficial for robustness
and efficiency reasons. For instance, it is common to employ AR (e.g.,
\citealp{AguirregabiriaMira2007}; \citealp{kalouptsidi2014time})
or VAR (\citealp{barwick2015costs}) specifications for transition
densities. However, even within these specifications there are a number
of choices to be made (e.g., dimension of VARs, distribution of error
terms etc) and the estimates of the structural parameters may not
be robust to these choices. More importantly, even when non-parametric
estimates of transition densities are available, simply plugging them
into the second-stage PMLE criterion would seriously degrade the rate
of convergence of structural parameters. One would need to adjust
the PMLE to account for the non-parametric first stage, but it is
not known what form this adjustment takes. By contrast, our proposals
employ non-parametric estimates of value functions, and as described
below, we derive the necessary adjustments to account for this. To
carry out counterfactual analysis, we therefore suggest plugging in
our estimates of the structural parameters, which do not rely on non-parametric
estimates of transition densities and are robust to mis-specification,
together with non-parametric estimates of the transition densities. 

The above discussion highlights that in continuous state spaces, estimation
of structural parameters is inherently a problem of semi-parametric
estimation. In fact, even under discrete states, \citet{AguirregabiriaMira2002}
show that estimation of transition densities affects the variance
of the structural parameter estimates. If the state variables are
continuous, existing two-step CCP methods, such as the PMLE will no
longer converge at parametric rates. We therefore derive a locally
robust estimator by adding a correction term to the PMLE criterion
function that accounts for the non-parametric estimation of value
function terms using either of our TD methods. We emphasize that this
construction is novel and does not directly follow from existing results,
e.g., in \citet{CEINR2018}. Its computation is particularly straightforward
under the linear semi-gradient approach. The resulting estimator converges
at parametric rates under continuous states and unrestricted transition
densities. 

Our TD estimators are thus consistent, computationally very cheap,
and they converge at parametric rates. Importantly, they provide a
feasible estimation method when the states are continuous and/or the
state space is large. This is particularly important for the estimation
of dynamic discrete games. Even with discrete states, existing methods
for the estimation of dynamic games (\citealp{BBL2007}; \citealp{AguirregabiriaMira2007};
\citealp{PesendorferSchmidt-Dengler2008}) require integrating out
other players' actions. With many players, or under continuous states,
this can get quite cumbersome. By contrast, our procedure works directly
with the joint empirical distribution of the states and their sample
successors. Thus the `integrating out' is done implicitly within the
sample expectations. 

Finally, we also incorporate permanent unobserved heterogeneity into
our methods by combining the TD estimation with an Expectation-Maximization
(EM) algorithm \citep{DLR1977}. 

A range of Monte Carlo studies confirm the workings of our algorithms.
First, we present simulations based on the dynamic firm entry problem
described in \citet{AgMag2018}. The model has seven structural parameters
of interest and five continuous state variables. Existing methods
usually struggle at this dimensionality; certainly, discretization
of the state space would not work too well. We show simulations using
both the linear semi-gradient method and the AVI method where we estimate
the value function terms using random forests, and derive results
both with and without locally robust correction. Our findings suggest
that while the linear semi-gradient has very little bias even without
the locally robust correction, the AVI method has slightly higher
small sample bias which is substantially reduced when using our locally
robust estimator. Overall, our estimators perform very well in this
setting and they outperform CCP estimators that employ discretization,
leading to a 4 to 11-fold reduction in average mean squared error
across the seven structural parameters of interest. 

Second, we test our algorithms for dynamic discrete games based on
a firm entry game similar to that outlined in \citet{AguirregabiriaMira2007}.
We use the linear semi-gradient method here and, as before, our estimates
are closely centered around the true parameters. Since this approach
requires the selection of a set of basis functions for the functional
approximations, we present results for different sets of basis functions
(a second, third and fourth order polynomial) in this model. Our findings
suggest that the choice of basis functions has only a small effect
on the performance of the estimator. Moreover, we show that a simple
cross-validation procedure may be used to select the preferred set
of functions. 

\subsection{Related literature}

The paper contributes to the literature on estimation of DDC models.
\citet{Rust1987} is the seminal work in this literature. Motivated
by computational considerations, \citeauthor{HotzMiller1993} \citeyearpar{HotzMiller1993}
propose the CCP algorithm. The CCP idea has subsequently been refined
by \citet{HMSS1994} who suggest a simulation-based method, and \citet{AguirregabiriaMira2002}
who develop a pseudo-likelihood estimator. \citet{ArcidiaconoMiller2011}
introduce and exploit the property of finite dependence to speed up
CCP estimation. Despite these advances, the estimation of DDC models
remains constrained by its computational complexity, particularly
in the large class of models where finite dependence does not hold.
Estimation of dynamic discrete games is particularly affected by these
issues as the strategic interaction of agents means that the state
space increases exponentially with the number of players. It is also
uncommon for finite dependence to hold in dynamic games. 

The standard CCP algorithm is a two-step method, and is known to suffer
from severe bias in finite samples. Aguirregabiria and Mira \citeyearpar{AguirregabiriaMira2002,AguirregabiriaMira2007}
address this issue by presenting a recursive CCP algorithm, the nested
pseudo-likelihood (NPL) algorithm. Under discrete states, the first-step
estimation does not affect the rate of convergence, but shows up in
form of higher-order bias for the structural parameters. The NPL algorithm
can ameliorate this, see \citet{KasaharaShimotsu2008}. With continuous
states, however, estimation of the transition density introduces bias
that is the dominant term in determining the rate of convergence.
This motivates the construction of our locally robust estimator which
gets rid of this bias and restores parametric rates of convergence. 

\citet{ackerberg2014asymptotic} and \citet{CEINR2018} consider semi-parametric
estimation using ML methods when either finite dependence or a ``terminal
action'' property holds (\citealp{HotzMiller1993}). \citet{CEINR2018}
also derive locally robust corrections for this setting. In both cases,
the PMLE criterion can be written as a function of choice probabilities
only (transition densities are not required); the authors employ non-parametric
estimates for choice probabilities and correct for this estimation
in the second stage. Computation and estimation is thus relatively
simpler under finite dependence. By contrast, our methods are applicable
to the more general and difficult setting where finite dependence
may not apply. Nevertheless, the computational speed of our linear
semi-gradient procedure is comparable to methods that exploit finite
dependence. For dynamic games, \citet{semenova2018machine} allows
for high-dimensional state spaces, but the parameters are only partially
identified. 

In making use of TD learning, our methods relate to the literature
on RL, particularly batch RL. Batch RL describes learning about how
to map states into actions so as to maximize an expected payoff, using
a fixed set of data (also called a batch); see \citet{lange2012batch}
for a survey.\footnote{See \citet{SuttonBarto2018} for a detailed treatment of RL in genreal.}
It is closely related to the idea of `experience replay', an important
ingredient of RL algorithms that achieve human level play in Atari
games (\citealp{mnih2015human}). A key step in RL, including batch
RL, is the estimation of value functions. TD learning methods, first
formulated by \citet{sutton1988learning}, are the most commonly used
set of algorithms for this purpose. We study non-parametric estimation
of value functions using two TD methods: semi-gradients and AVI. Our
analysis builds on the techniques developed by \citet{TsitsiklisRoy1997}
for linear semi-gradients, and \citet{munos2008finite} for AVI. While
\citet{TsitsiklisRoy1997} focus on online learning (i.e., where collection
of data and estimation of value functions is conducted simultaneously),
we translate their methods to batch learning. With regards to \citet{munos2008finite},
we differ in employing assumptions that are more common to econometrics
and our characterization of the rates is also different (compare Theorem
2 in their paper with our Theorem \ref{Theorem 3}).

TD methods are distinct from other value function approximation methods
developed in economics, e.g., parametric policy iteration (\citealp{benitez2000comparison}),
simulation and interpolation (\citealt{keane1994solution}), and sieve
value function iteration (\citealp{arcidiacono2013approximating}).
The last of these is similar in spirit to AVI with linear functional
approximations. However, our semi-gradient method provides a linear
approximation in a single step without any need for iterations, and
we analyze AVI under generic machine learning methods. Our approximation
results, and the technical arguments leading to them, are thus different
from \citet{arcidiacono2013approximating}; in fact, their setting
is different too as the authors focus on estimating the `optimal'
value function, while the recursive terms in our setting are more
akin to a value function under a fixed policy. 

The remainder of this paper is organized as follows. Section \ref{sec:Setup}
outlines the setup of the DDC model and fixes notation. Section \ref{sec:Temporal-difference-learning-alg}
describes our TD estimation method for the functional approximations
of the value functions. Section \ref{sec:Theoretical-Properties-of}
proves its theoretical properties and describes the second-step estimation
of the structural parameters under discrete and continuous state variables.
Section \ref{sec:Estimation-of-dynamic} discusses the estimation
of dynamic discrete games. Section \ref{sec:Simulations} provides
Monte Carlo simulations for our algorithm. Section \ref{sec:Conclusions}
concludes. The Online Appendix discusses extensions to permanent unobserved
heterogeneity.

\section{Setup \label{sec:Setup}}

We start with a single agent DDC model. In particular, we consider
an infinite horizon model in discrete time with $i=1,\dots,n$ agents.
We assume that the individuals are homogeneous, relegating extensions
for unobserved heterogeneity to Online Appendix \ref{sec:Incorporating-permanent-unobserv}.
In each period, an agent chooses among $A$ mutually exclusive actions,
each of which is denoted by $a$. The payoff from the action depends
on the current state $x$. Choosing action $a$ when the state is
$x$ gives the agent an instantaneous utility of $z(a,x)^{\intercal}\theta+e$,
where $z(a,x)$ is some known vector valued function of $a,x$ and
$e$ is an idiosyncratic error term. We denote the realization of
the state of an individual $i$ at time $t$ by $x_{it}$, and her
corresponding action and error terms by $a_{it}$ and $e_{it}$. We
assume that $e_{it}$ is an iid draw from some known distribution
$g_{e}(\cdot)$. Let $(a^{\prime},x^{\prime})$ denote the one-period
ahead random variables immediately following the actions and states
$(a,x)$, where $x^{\prime}\sim K(\cdot\vert a,x)$, with $K(\cdot\vert a,x)$
denoting the transition density given $a,x$ (more precisely, it is
the Markov kernel). We do not make any parametric assumptions about
$K(\cdot|a,x)$. The utility from future periods is discounted by
$\beta$. 

Agent $i$ chooses actions ${\bf a}_{i}=(a_{i1},a_{i2},\dots)$ to
sequentially maximize the discounted sum of payoffs 
\[
E\left[\sum_{t=1}^{\infty}\beta^{t}\left\{ z(x_{it},a_{it})^{\intercal}\theta^{*}+e_{it}\right\} \right].
\]
The econometrician observes a panel consisting of state-action pairs
for all individuals, $({\bf x}_{i},{\bf a}_{i})=\{(x_{i1},a_{i1}),\dots,(x_{iT},a_{iT})\}$,
for $T$ periods (note, however, that the agent maximizes an infinite
horizon objective, not a fixed $T$ one). Typically $T\ll n$ in applications,
so we work within an asymptotic regime where $n\to\infty$ but $T$
is fixed. Using this data, the econometrician aims to recover the
structural parameters $\theta^{*}$. 

In this paper, we study the CCP approach for estimating $\theta^{*}$
(\citealp{HotzMiller1993}). CCP methods are based on the conditional
choice probabilities of choosing action $a$ given state $x$. We
denote these by $P_{t}(a\vert x)$ for a given period $t$ but henceforth
drop the subscript $t$ with the idea that it can be made a part of
the state variable $x$, if needed (we should also add that some of
our theoretical results are based on assuming stationarity, i.e $P_{t}(a\vert x)$
is independent of $t$). Denote $e(a,x)$ as expected value of the
idiosyncratic error term $e$ given that action $a$ was chosen. \citet{HotzMiller1993}
show that if the distribution of $e$ follows a Generalized Extreme
Value (GEV) distribution, it is possible to express $e(a,x)$ as a
function of the choice probabilities $P(a\vert x)$, i.e $e(a,x)=\mathcal{G}(P(a\vert x))$.
We assume that $e$ follows a Type I Extreme Value distribution, which
is perhaps the most common choice in applications. In this case $e(a,x)=\gamma-\ln P(a\vert x)$,
where $\gamma$ is the Euler constant. 

The standard procedure in the CCP approach is as follows: Under the
given distributional assumptions, the parameters are obtained as the
maximizers of the pseudo-likelihood function
\[
Q(\theta)=\sum_{i=1}^{n}\sum_{t=1}^{T-1}\log\frac{\exp\left\{ h(a_{it},x_{it})^{\intercal}\theta+g(a_{it},x_{it})\right\} }{\sum_{a}\exp\left\{ h(a,x_{it})^{\intercal}\theta+g(a,x_{it})\right\} },
\]
where for any $(a,x)$, $h(.)$ and $g(.)$ solve the following recursive
expressions:
\begin{align}
h(a,x) & =z(a,x)+\beta\mathbb{E}\left[h(a^{\prime},x^{\prime})\vert a,x\right],\label{eq:Recursion}\\
g(a,x) & =\beta\mathbb{E}\left[e(a^{\prime},x^{\prime})+g(a^{\prime},x^{\prime})\vert a,x\right].\nonumber 
\end{align}
Here, $\mathbb{E}[.]$ denotes the expectation over the distribution
of $(a^{\prime},x^{\prime})$ conditional on $(a,x)$. Note that $\mathbb{E}\left[h(a^{\prime},x^{\prime})\vert a,x\right]$
is a function of $K(\cdot|a,x),P(.\vert x)$. Both $h(a,x)$ and $g(a,x)$
have a `value-function' form, which turns out to be useful for our
approach.

Clearly, $h(.)$ and $g(.)$ are functions of $K(\cdot|\cdot)$ and
$P(.\vert.)$. Since the latter are unknown, the current literature
proceeds by first estimating these as $(\hat{K},\hat{P})$. Typically,
$\hat{K}$ is obtained by MLE based on a parametric form of $K(x^{\prime}\vert a,x;\theta_{f})$,
while $\hat{P}$ is estimated non-parametrically using either a blocking
scheme or kernel regression. Then, given $(\hat{K},\hat{P})$, the
values of $h(.)$ and $g(.)$ can be estimated by solving the recursive
equations (\ref{eq:Recursion}). In the next section, we propose an
alternative algorithm for maximizing $Q(\theta)$ that directly estimates
$h(\cdot)$ and $g(\cdot)$ in a single step without requiring any
knowledge about or estimation of $K(\cdot\vert\cdot)$. 

\subsubsection*{Notation}

We assume that the distribution of $(a_{it},x_{it})$ is time stationary.
This greatly simplifies our notation. It is not necessary for our
results on the approximation properties of our TD methods, see Appendix
\ref{sec:Appendix:A}, but we do require it for deriving a locally
robust estimator. Let $\mathbb{P}$ denote the stationary population
(i.e, in the limit as $n\to\infty$) distribution of $(a,x,a^{\prime},x^{\prime})$,
and $\mathbb{E}[\cdot]$ the corresponding expectation over $\mathbb{P}$.
Define $\mathbb{E}_{n}[\cdot]$ as the expectation over the empirical
distribution, $\mathbb{P}_{n}$, of $(a,x,a^{\prime},x^{\prime})$.
In particular, $\mathbb{E}_{n}[f(a,x,a^{\prime},x^{\prime})]:=(n(T-1))^{-1}\sum_{i=1}^{n}\sum_{t=1}^{T-1}f(a_{it},x_{it},a_{it+1},x_{it+1})$,
i.e we always drop the last time period in the summation index even
if $f(\cdot)$ does not depend on $a^{\prime},x^{\prime}$.

Let $\mathcal{H}$ denote the space of all square integrable functions
over the domain $\mathcal{A}\times\mathcal{X}$ of $(a,x)$. Define
the pseudo-norm $\left\Vert \cdot\right\Vert _{2}$ over $\mathcal{\mathcal{H}}$
as $\left\Vert f\right\Vert _{2}:=\mathbb{E}[\vert f(a,x)\vert^{2}]^{1/2}$
for all $f\in\mathcal{\mathcal{H}}$. We use $\left|\cdot\right|$
to denote the usual Euclidean norm on a Euclidean space.

\section{Temporal-difference estimation \label{sec:Temporal-difference-learning-alg}}

This section presents our TD method for estimating $h(.)$ and $g(.)$.
Note that $h(\cdot)$ is a vector, of the same dimension as $\theta$.
Our methods provide functional approximations separately for each
component $h^{(j)}$ of $h$. To simplify notation, we drop the superscript
$j$ indexing the elements of $h(.)$ and proceed as if the latter,
and therefore $\theta^{*}$, is a scalar. However, it should be taken
as implicit that all our results hold for general $h(.)$, as long
as each of its elements are treated separately.

For any candidate function, $f(a,x)$, for $h(a,x)$, denote the TD
error by 
\[
\delta_{z}(a,x,a^{\prime},x^{\prime};f):=z(a,x)+\beta f(a^{\prime},x^{\prime})-f(a,x),
\]
and the dynamic programming operator by
\[
\Gamma_{z}[f](a,x):=z(a,x)+\beta\mathbb{E}[f(a^{\prime},x^{\prime})\vert a,x].
\]
Clearly, $h(a,x)$ is the unique fixed point of $\Gamma_{z}[\cdot]$.
TD estimation involves approximating $h(a,x)$ using a functional
class $\mathcal{F}$, where each element $h(\cdot;\omega)$ of $\mathcal{F}$
is indexed by a finite-dimensional vector $\omega$. The aim in TD
estimation is to ostensibly minimize the mean-squared TD error 
\[
\textrm{TDE}(\omega):=\mathbb{E}\left[\left\Vert z(a,x)+\beta h(a^{\prime},x^{\prime};\omega)-h(a,x;\omega)\right\Vert ^{2}\right].
\]
However, this minimization problem is neither computationally feasible
nor is it proven to converge when $h\notin\mathcal{F}$. Instead,
two approaches are commonly used. 

The first approach, the semi-gradient method, involves updating $\omega$
using the semi-gradients
\begin{equation}
\omega_{j+1}=\omega_{j}+\gamma\mathbb{E}\left[\left\{ z(a,x)+\beta h(a^{\prime},x^{\prime};\omega_{j})-h(a,x;\omega_{j})\right\} \nabla_{\omega}h(a,x;\omega_{j})\right]\label{eq:Semi-gradients}
\end{equation}
for some small value of $\gamma$. As the name suggests, the above
is not a complete gradient as the derivative does not take into account
how $\omega$ affects the `target', i.e., the future value $h(a^{\prime},x^{\prime};\omega)$.
Nevertheless, for linear functional classes $\mathcal{F}$, it is
possible to explicitly characterize the limit point of the updates,
$\omega^{*}$, and compute it directly. Section \ref{subsec:Semi-gradient-methods}
describes this in greater detail. In the RL literature, it is popular
to employ semi-gradients with neural networks as the functional class
$\mathcal{F}$, but it appears difficult to extend our theoretical
analysis to this setting. 

The second approach, Approximate Value Iteration (AVI; \citealp{munos2008finite}),
employs the idea of `target networks'. In this approach, the parameters
in the future value of $h$, i.e $h(a^{\prime},x^{\prime};\omega)$,
are fixed at the current $\omega$, and the functional parameters
iteratively updated as 
\begin{equation}
\omega_{j+1}=\argmin_{\omega}\mathbb{E}\left[\left\Vert z(a,x)+\beta h(a^{\prime},x^{\prime};\omega_{j})-h(a,x;\omega)\right\Vert ^{2}\right].\label{eq:AVI}
\end{equation}
Clearly, the semi-gradient method and AVI are closely related: if
one were to solve the minimization problem in (\ref{eq:AVI}) using
gradient descent, the updates (within each iteration) would look similar
to (\ref{eq:Semi-gradients}) except for fixing the value of $\omega$
in $h(a^{\prime},x^{\prime};\omega)$ at the past values. After the
gradient updates converge, i.e at the end of the iteration, $h(a^{\prime},x^{\prime};\omega)$
is revised with the new $\omega$. The semi-gradient approach can
thus be considered a one-step variant of AVI. Section \ref{subsec:Approximate-Value-Iteration}
describes AVI in greater detail. We characterize the theoretical properties
of AVI under general functional classes $\mathcal{F}$ including neural
networks, random forests, LASSO etc. 

The approximation to $g$ follows similarly after replacing $\delta_{z}(\cdot;f),\Gamma_{z}[\cdot]$
by
\begin{align*}
\delta_{e}(a,x,a^{\prime},x^{\prime};f) & :=\beta e(a^{\prime},x^{\prime})+\beta f(a^{\prime},x^{\prime})-f(a,x),\\
\Gamma_{e}[f](a,x) & :=\beta\mathbb{E}[e(a^{\prime},x^{\prime})+f(a^{\prime},x^{\prime})\vert a,x].
\end{align*}

\subsection{Semi-gradients\label{subsec:Semi-gradient-methods}}

The semi-gradient approach with linear $\mathcal{F}$ is particularly
well suited for computation. Let $\phi(a,x)$ consist of a set of
basis functions over the domain $(a,x)$. Then the linear approximation
class is $\mathcal{F}\equiv\{\phi(a,x)^{\intercal}\omega:\omega\in\mathbb{R}^{k}\}$.
Denote the projection operator onto $\mathcal{F}$ by $P_{\phi}$:
\[
P_{\phi}[f](a,x):=\phi(a,x)^{\intercal}\mathbb{E}[\phi(a,x)\phi(a,x)^{\intercal}]^{-1}\mathbb{E}[\phi(a,x)f(a,x)].
\]

For linear basis functions, it can be shown, e.g \citet{TsitsiklisRoy1997},
that the sequence of functional approximations $h(a,x;\omega_{j}):=\phi(a,x)^{\intercal}\omega_{j}$
converges to $h^{*}:=\phi(a,x)^{\intercal}\omega^{*}$, defined as
the fixed point of the projected dynamic programming operator $P_{\phi}\Gamma_{z}[\cdot]$
(i.e., $P_{\phi}\Gamma_{z}[h^{*}]=h^{*}$). Based on this characterization,
we show in Lemma \ref{Lemma 1} in Appendix \ref{sec:Appendix:A}
that $h^{*}(a,x)=\phi(a,x)^{\intercal}\omega^{*}$, where
\begin{equation}
\omega^{*}=\mathbb{E}\left[\phi(a,x)\left(\phi(a,x)-\beta\phi(a^{\prime},x^{\prime})\right)^{\intercal}\right]^{-1}\mathbb{E}\left[\phi(a,x)z(a,x)\right].\label{eq:inversion for TD fixed point}
\end{equation}
Lemma \ref{Lemma 2} in Appendix \ref{sec:Appendix:A} assures that
$\mathbb{E}\left[\phi(a,x)\left(\phi(a,x)-\beta\phi(a^{\prime},x^{\prime})\right)^{\intercal}\right]$
is indeed non-singular as long as $\beta<1$ and $\mathbb{E}[\phi(a,x)\phi(a,x)^{\intercal}]$
is non-singular. As defined above, $\omega^{*}$ cannot be computed
directly, since it is a function of the true expectation $\mathbb{E}[\cdot]$.
We can however obtain an estimator, $\hat{\omega}$, after replacing
$\mathbb{E}[\cdot]$ with the sample expectation $\mathbb{E}_{n}[\cdot]$:
\begin{equation}
\hat{\omega}=\mathbb{E}_{n}\left[\phi(a,x)\left(\phi(a,x)-\beta\phi(a^{\prime},x^{\prime})\right)^{\intercal}\right]^{-1}\mathbb{E}_{n}\left[\phi(a,x)z(a,x)\right].\label{eq:empirical TD fixed point}
\end{equation}
Using $\hat{\omega}$, we obtain an estimate of $h(\cdot)$ as $\hat{h}(a,x)=\phi(a,x)^{\intercal}\hat{\omega}$. 

We now turn to the estimation of $g(\cdot)$. As with $h(\cdot)$,
we approximate $g(\cdot)$ using basis functions $r(a,x)$, which
may generally be different from $\phi(a,x)$. Let $P_{r}$ denote
the projection operator onto the space $\{r(a,x)^{\intercal}\xi:\xi\in\mathbb{R}^{k}\}$.
The limit of the semi-gradient iterations is $g^{*}(a,x):=r(a,x)^{\intercal}\xi^{*},$
defined as the fixed point of $P_{r}\Gamma_{e}[\cdot]$. Assuming
$e(a,x)$ is known, we obtain the following characterization of $\xi^{*}$
in analogy with (\ref{eq:inversion for TD fixed point}): 
\begin{equation}
\xi^{*}=\mathbb{E}\left[r(a,x)\left(r(a,x)-\beta r(a^{\prime},x^{\prime})\right)^{\intercal}\right]^{-1}\mathbb{E}\left[\beta r(a,x)e(a^{\prime},x^{\prime})\right].\label{eq:xi-estimate-method1}
\end{equation}
In the above, $e(a,x)=\gamma-\ln P(a\vert x)$ is a function of choice
probabilities, which are unknown. Denote $\eta(a,x):=P(a\vert x)$.
Suppose that we have access to a non-parametric estimator $\hat{\eta}$
of $\eta$. This can be obtained, e.g., through series or kernel regression.
We can then plug in this estimate to obtain $e(a,x;\hat{\eta}):=\gamma-\ln\hat{\eta}(a,x)$.
This in turn enables us to estimate $\xi^{*}$ using $\hat{\xi}$,
computed as
\begin{align}
\hat{\xi} & =\mathbb{E}_{n}\left[r(a,x)\left(r(a,x)-\beta r(a^{\prime},x^{\prime})\right)^{\intercal}\right]^{-1}\mathbb{E}_{n}\left[\beta r(a,x)e(a^{\prime},x^{\prime};\hat{\eta})\right].\label{eq:empirical TD for xi}
\end{align}
Using the above, we obtain an estimate of $g(\cdot)$ as $\hat{g}(a,x)=r(a,x)^{\intercal}\hat{\xi}$. 

Interestingly, estimation of $\xi^{*}$ is unaffected to a first order
by the estimation of $\hat{\eta}$, even though the latter converges
to the true $\eta$ at non-parametric rates (see Section \ref{sec:Theoretical-Properties-of}
for a formal statement). This is because an orthogonality property
holds for the estimation of $\xi$, in that 
\begin{equation}
\partial_{\eta}\mathbb{E}\left[\beta r(a,x)e(a^{\prime},x^{\prime};\eta)\right]=0,\label{eq:orthogonality for =00005Ceta}
\end{equation}
where $\partial_{\eta}\cdot$ denotes the Fr\'echet derivative with
respect to $\eta$. To show (\ref{eq:orthogonality for =00005Ceta}),
expand 
\begin{align}
\mathbb{E}\left[\beta r(a,x)e(a^{\prime},x^{\prime};\eta)\right] & =\mathbb{E}\left[\beta r(a,x)\mathbb{E}\left[\left.e(a^{\prime},x^{\prime};\eta)\right|x^{\prime}\right]\right]\nonumber \\
 & =\mathbb{E}\left[\beta r(a,x)\mathbb{E}\left[\left.\gamma-\ln\eta(a^{\prime},x^{\prime})\right|x^{\prime}\right]\right],\label{eq:expanding =00005Cxiterm}
\end{align}
where the first equality follows from the Markov property. Consider
the functional $M(\tilde{\eta}):=\mathbb{E}\left[\left.\ln\tilde{\eta}(a^{\prime},x^{\prime})\right|x^{\prime}\right]$
at different candidate values $\tilde{\eta}(\cdot,\cdot)$. At the
true conditional choice probability, $\eta$, $M(\tilde{\eta})$ becomes
the conditional entropy of $P(a\vert x^{\prime}$) and attains its
maximum. Hence, $\partial_{\eta}\mathbb{E}\left[\left.\ln\eta(a^{\prime},x^{\prime})\right|x^{\prime}\right]=0$
and (\ref{eq:orthogonality for =00005Ceta}) follows from (\ref{eq:expanding =00005Cxiterm}).
Consequently, $\hat{\xi}$ is a locally robust estimator for $\xi$.

Note that computation of $\hat{\omega}$ and $\hat{\xi}$ only involves
solving linear equations of dimension $\textrm{dim}(\phi)$ and $\textrm{dim}(r)$,
respectively. This is computationally very cheap. Using $\hat{h}(a,x)$
and $\hat{g}(a,x)$, we can in turn estimate $\theta^{*}$ in many
different ways. For instance, we can plug them into the PMLE estimator
\begin{equation}
\hat{\theta}:=\argmax_{\theta}\hat{Q}(\theta):=\sum_{i=1}^{n}\sum_{t=1}^{T-1}\log\frac{\exp\left\{ \hat{h}(a_{it},x_{it})\theta+\hat{g}(a_{it},x_{it})\right\} }{\sum_{a}\exp\left\{ \hat{h}(a,x_{it})\theta+\hat{g}(a,x_{it})\right\} }.\label{eq:estimate of theta}
\end{equation}
However, such plug-in estimates are sub-optimal. In Section \ref{subsec:Continuous-states-and},
we suggest a locally robust version of (\ref{eq:estimate of theta}). 

Suppose that the underlying states and actions are discrete, and that
our algorithm uses basis functions comprised of the set of all discrete
elements of $x,a$. Then, we show in Online Appendix \ref{subsec:Discrete-states}
that the resulting estimate of $h(a,x)$ is exactly the same as that
obtained from the standard CCP estimators, if both the choice and
transition probabilities were estimated using cell values. 

\subsection{Approximate Value Iteration (AVI)\label{subsec:Approximate-Value-Iteration}}

For a feasible estimation procedure using AVI, we can replace $\mathbb{E}[\cdot]$
by $\mathbb{E}_{n}[\cdot]$ in (\ref{eq:AVI}). The procedure builds
a sequence of approximations $\{\hat{h}_{j}:=h(a,x;\hat{\omega}_{j}),j=1,\dots,J\}$
for $h$, where 
\begin{equation}
\hat{h}_{j+1}=\argmin_{f\in\mathcal{F}}\mathbb{E}_{n}\left[\left\Vert z(a,x)+\beta\hat{h}_{j}(a^{\prime},x^{\prime})-f(a,x)\right\Vert ^{2}\right].\label{eq:sample AVI}
\end{equation}
The process can be started with an arbitrary initialization, e.g.,
$\hat{h}_{1}(a,x)=z(a,x)$. The maximum number of iterations, $J$,
is only limited by computational feasibility.

The minimization problem in (\ref{eq:sample AVI}) is equivalent to
a prediction problem using the functional class $\mathcal{F}$, where
the outcomes are $z(a,x)+\beta\hat{h}_{j}(a^{\prime},x^{\prime})$
evaluated at the various sample draws of $(a,x,a^{\prime},x^{\prime})$.
Hence, the estimation target for $\hat{h}_{j+1}$ is the conditional
expectation $\mathbb{E}[z(a,x)+\beta\hat{h}_{j}(a^{\prime},x^{\prime})\vert a,x]\equiv\Gamma_{z}[\hat{h}_{j}](a,x)$.
In other words, each $\hat{h}_{j+1}$ is a non-parametric approximation
to $\Gamma_{z}[\hat{h}_{j}]$, and in this manner AVI builds a series
of approximate value function iterations (as its name suggests).

The interpretation of (\ref{eq:sample AVI}) as a prediction problem
also enables us to employ any machine learning method devised for
prediction, including (but not limited to) LASSO, random forests and
neural networks. Our theoretical results show that it is possible
to estimate $h$ at suitably fast rates under very weak assumptions
on the non-parametric estimation rates of machine learning methods.

The estimation procedure for $g(\cdot)$ is similar: we construct
a sequence of approximations $\{\hat{g}_{j},j=1,\dots,J\}$ for $g$
as 
\[
\hat{g}_{j+1}=\argmin_{f\in\mathcal{F}}\mathbb{E}_{n}\left[\left\Vert \beta e(a^{\prime},x^{\prime};\hat{\eta})+\beta\hat{g}_{j}(a^{\prime},x^{\prime})-f(a,x)\right\Vert ^{2}\right].
\]
As in Section \ref{eq:Semi-gradients}, it will be shown that the
estimation error of $\eta$ is first-order ignorable for the estimation
of $g$. Using $\hat{h}(a,x)$ and $\hat{g}(a,x)$, we can, as before,
estimate $\theta^{*}$ in many different ways, including the PMLE
estimator \ref{eq:estimate of theta}.

Compared to the semi-gradient approach, AVI is computationally more
expensive as it requires solving $J$ prediction problems (in Section
\ref{subsec:Estimation-of-non-parametric} we show that in the worst
case $J\approx\ln n$, but this can be substantially reduced through
good initializations). However, semi-gradient methods require differentiable
classes of functions (e.g., random forests are not allowed) and it
appears difficult to characterize their theoretical properties in
the case of nonlinear basis functions.

\subsection{Tuning parameters\label{subsec:Tuning-parameters}}

Both the semi-gradient and AVI methods will require choosing tuning
parameters. For AVI this is straightforward: as each iteration is
a non-parametric estimation problem, the tuning parameters can be
chosen in the usual manner, e.g., using cross-validation. In the case
of linear semi-gradient methods, the tuning parameters are the dimensions
$k_{\phi}=\textrm{dim}(\phi)$ and $k_{r}=\textrm{dim}(r)$ of the
basis functions. In analogy with AVI, we propose selecting both through
a procedure akin to cross-validation. The value of $\omega$ is estimated
using a training sample and its performance evaluated on a hold-out
or test sample, where the performance is measured in terms of the
empirical mean-squared TD error $\mathbb{E}_{n,\textrm{test}}[\delta_{h}^{2}(a,x,a^{\prime},x^{\prime};\hat{h})]$
on the test dataset. The values of $k_{\phi},k_{r}$ are chosen to
minimize the mean squared TD error.

\subsection{Nonlinear utility functions\label{subsec:Nonlinear-utility-functions}}

We take the utility function to be linear in $\theta$ for simplicity
- and because it is the most common choice in practice - but this
is easily relaxed. Denote the nonlinear utility by $z(a,x;\theta)$.
Then, $\theta^{*}$ is the maximizer of the pseudo-likelihood criterion
\[
Q(\theta)=\sum_{i=1}^{n}\sum_{t=1}^{T-1}\log\frac{\exp\left\{ h(a_{it},x_{it};\theta)+g(a_{it},x_{it})\right\} }{\sum_{a}\exp\left\{ h(a,x_{it};\theta)+g(a,x_{it})\right\} },
\]
where, for each $\theta$, $h(.;\theta)$ is defined analogously to
(\ref{eq:Recursion}) with $z(a,x;\theta)$ replacing $z(a,x)$. 

We can use our TD procedures to obtain a functional approximation
$\hat{h}(a,x;\theta)$ for $h(a,x;\theta)$ at each $\theta$ (estimation
of $g(\cdot)$ remains unchanged as it does not depend on $\theta$).
As argued earlier, computation of $\hat{h}(\cdot;\theta)$ can be
very fast. Appealingly, the matrix $\mathbb{E}_{n}\left[\phi(a,x)\left(\phi(a,x)-\beta\phi(a^{\prime},x^{\prime})\right)^{\intercal}\right]$
employed in the linear semi-gradient estimate (\ref{eq:empirical TD fixed point})
does not feature $z(a,x;\theta)$, and therefore only has to be inverted
once. We can then plug in the values of $\hat{h}(.;\theta)$ and $\hat{g}(\cdot)$
to estimate $\theta^{*}$ as 
\[
\hat{\theta}=\arg\max_{\theta\in\Theta}\hat{Q}(\theta);\quad\hat{Q}(\theta):=\sum_{i=1}^{n}\sum_{t=1}^{T-1}\log\frac{\exp\left\{ \hat{h}(a_{it},x_{it};\theta)+\hat{g}(a_{it},x_{it})\right\} }{\sum_{a}\exp\left\{ \hat{h}(a,x_{it};\theta)+\hat{g}(a,x_{it})\right\} }.
\]
A locally robust counterpart to $\hat{Q}(\theta)$ can be derived
in the same manner as in Section \ref{subsec:Continuous-states-and}.
However, computation of $\hat{\theta}$ is more involved as it is
no longer convex optimization. 

\subsection{Unobserved heterogeneity}

In Online Appendix \ref{sec:Incorporating-permanent-unobserv} we
incorporate permanent unobserved heterogeneity into our models by
pairing our TD methods with the sequential Expectation-Maximization
(EM) algorithm (\citealp{arcidiacono2003finite}). The resulting algorithm
can handle discrete heterogeneity in both individual utilities and
transition densities. Online Appendix \ref{subsec:Bus-Engine-Replacement}
provides Monte-Carlo evidence suggesting that the algorithm works
well in practice. 

\section{Theoretical Properties of TD estimators\label{sec:Theoretical-Properties-of}}

\subsection{Estimation of non-parametric terms\label{subsec:Estimation-of-non-parametric}}

We characterize rates of convergence for estimation of $h(\cdot)$
and $g(\cdot)$ under both semi-gradients and AVI.

\subsubsection{Linear semi-gradients}

We impose the following assumptions for estimation of $h(\cdot)$. 

\begin{asm1} (i) The basis vector $\phi(a,x)$ is linearly independent
(i.e. $\phi(a,x)^{\intercal}\omega\linebreak=0$ for all $(a,x)$
if and only if $\omega=0$). Additionally, the eigenvalues of $\mathbb{E}[\phi(a,x)\linebreak\phi(a,x)^{\intercal}]$
are uniformly bounded away from zero for all $k_{\phi}$. \label{1(i)}

(ii) $\vert\phi(a,x)\vert_{\infty}\le M$ for some $M<\infty$. \label{1(ii) }

(iii) There exists $C<\infty$ and $\alpha>0$ such that $\left\Vert h-P_{\phi}[h]\right\Vert _{2}\le Ck_{\phi}^{-\alpha}$.\label{1(iii)}

(iv) The domain of $(a,x)$ is a compact set, and $\vert z(a,x)\vert_{\infty}\le L$
for some $L<\infty$. \label{1(iv)}

(v) $k_{\phi}\to\infty$ and $k_{\phi}^{2}/n\to0$ as $n\to\infty$.
\label{1(v)}

\end{asm1}

Assumption 1(i) rules out multi-collinearity in the basis functions.
This is easily satisfied. Assumption 1(ii)  ensures that the basis
functions are bounded. This is again a mild requirement and is easily
satisfied if either the domain of $(a,x)$ is compact, or the basis
functions are chosen appropriately (e.g., a Fourier basis). Assumption
1(iii) is a standard condition on the rate of approximation of $h(a,x)$
using a basis approximation. The value of $\alpha$ is related to
the smoothness of $h(\cdot)$. \citet{newey1997convergence} shows
that for splines and power series, $\alpha=r/d$, where $r$ is the
number of continuous derivatives of $h(a,\cdot)$ and $d$ is the
dimension of $x$. Similar results can also be derived for other approximating
functions such as Fourier series, wavelets and Bernstein polynomials.
The smoothness properties of $h(a,\cdot)$ are discussed in Online
Appendix \ref{subsec:Smoothness-properties-of}, where we provide
primitive conditions on $z(a,x),K(x^{\prime}\vert a,x)$ that ensure
existence of $r$ continuous derivatives of $h(a,\cdot)$ for each
$a\in\mathcal{A}$. Assumption 1(iv) requires $z(a,x)$ to be bounded.
Finally, Assumption 1(v) specifies the rate at which the dimension
of the basis functions are allowed to grow. The rate requirements
are mild, and are the same as those employed for standard series estimation.
For the theoretical properties, the exact rate of $k_{\phi}$ is not
relevant up to a first order since we propose estimators of $\theta^{*}$
that are locally robust to estimation of $h(\cdot)$. 

We then have the following theorem on the estimation of $h(a,x)$:

\begin{thm} \label{Theorem 1}Under Assumptions 1(i) - 1(v), the
following hold: 

(i) Both $\omega^{*}$ and $\hat{\omega}$ exist, the latter with
probability approaching one.

(ii) $\left\Vert h(a,x)-\phi(a,x)^{\intercal}\omega^{*}\right\Vert _{2}\le(1-\beta)^{-1}\left\Vert h-P_{\phi}[h]\right\Vert _{2}\le C(1-\beta)^{-1}k_{\phi}^{-\alpha}$.

(iii) The $L^{2}$ error for the difference between $h(a,x)$ and
$\phi(a,x)^{\intercal}\hat{\omega}$ is bounded as 
\[
\left\Vert h(a,x)-\phi(a,x)^{\intercal}\hat{\omega}\right\Vert _{2}=O_{p}\left((1-\beta)^{-1}\left\{ \frac{k_{\phi}}{\sqrt{n}}+k_{\phi}^{-\alpha}\right\} \right).
\]
\end{thm}

We prove Theorem \ref{Theorem 1} in Appendix \ref{subsec:Proof-of-Theorem1}
by adapting the results of \citet{TsitsiklisRoy1997}. The first part
of Theorem \ref{Theorem 1} assures that both population and empirical
TD fixed points exist. The second and third parts of Theorem \ref{Theorem 1}
imply that the approximation bias and MSE of linear semi-gradients
are analogous to those of standard series estimation apart from a
$(1-\beta)^{-1}$ factor. 

For the estimation of $\hat{\xi}$ we make use of cross-fitting as
a technical device to obtain easy-to-verify assumptions for the estimation
of $\hat{\eta}$. This entails the following: we randomly partition
the data into two folds. We estimate $\hat{\xi}$ separately for each
fold using $\hat{\eta}$ estimated from the opposite fold. The final
estimate of $\xi^{*}$ is the weighted average of $\hat{\xi}$ from
both the folds. We think of cross-fitting in this context as a convenient
assumption for the proofs, and do not believe it is necessary in practice. 

We impose the following assumptions for the estimation of $g(a,x)$.
Let $k_{r}$ denote the dimension of $r(a,x)$.

\begin{asm2} (i) The basis vector $r(a,x)$ is linearly independent,
and the eigenvalues of $\mathbb{E}[r(a,x)r(a,x)^{\intercal}]$ are
uniformly bounded away from zero for all $k_{r}$. \label{2(i)}

(ii) $\vert r(a,x)\vert_{\infty}\le M$ for some $M<\infty$. \label{2(ii) }

(iii) There exists $C<\infty$ and $\alpha>0$ such that $\left\Vert g(a,x)-P_{r}[g(a,x)]\right\Vert _{2}\le Ck_{r}^{-\alpha}$.\label{2(iii)}

(iv) The domain of $(a,x)$ is a compact set, and $\vert e(a,x)\vert_{\infty}\le L<\infty$.
\label{2(iv)}

(v) $k_{r}\to\infty$ and $k_{r}^{2}/n\to0$ as $n\to\infty$. \label{2(v)}

(vi) $\hat{\xi}$ is estimated from a cross-fitting procedure described
above. The conditional choice probability function satisfies $\eta(a,x)\ge\delta>0$,
where $\delta$ is independent of $a,x$. Additionally, $\left\Vert \eta-\hat{\eta}\right\Vert _{\infty}=o_{p}(1)$
and $\left\Vert \eta-\hat{\eta}\right\Vert _{2}^{2}=o_{p}(n^{-1/2})$.\label{2(vi)}

\end{asm2}

Assumption 2 is a direct analogue of Assumption 1, except for the
last part which provides regularity conditions when $\eta(\cdot)$
is estimated. These conditions are typical for locally robust estimates
and only require the non-parametric function $\eta(a,x)$ to be estimable
at faster than $n^{-1/4}$ rates. This is easily verified for most
non-parametric estimation methods such as kernel or series regression.
Under these assumptions, we have the following analogue of Theorem
\ref{Theorem 1}, which we prove in Appendix \ref{subsec:Proof-of-Theorem2}
.

\begin{thm} \label{Theorem 2}Under Assumptions 2(i) to 2(vi),
the following hold: 

(i) Both $\xi^{*}$ and $\hat{\xi}$ exist, the latter with probability
approaching one.

(ii) $\left\Vert g(a,x)-r(a,x)^{\intercal}\xi^{*}\right\Vert _{2}\le(1-\beta)^{-1}\left\Vert g(a,x)-P_{r}g(a,x)\right\Vert _{2}\le C(1-\beta)^{-1}k_{r}^{-\alpha}$.

(iii) The $L^{2}$ error for the difference between $g(a,x)$ and
$r(a,x)^{\intercal}\hat{\xi}$ is bounded as 
\[
\left\Vert g(a,x)-r(a,x)^{\intercal}\hat{\xi}\right\Vert _{2}=O_{p}\left((1-\beta)^{-1}\left\{ \frac{k_{r}}{\sqrt{n}}+k_{r}^{-\alpha}\right\} \right).
\]
\end{thm}

\subsubsection{Approximate Value Iteration}

We can expand the estimation error $\left\Vert \hat{h}_{J}-h\right\Vert _{2}$
in terms of the non-parametric estimation errors $\left\Vert \Gamma_{z}[\hat{h}_{j-1}]-h_{j}\right\Vert _{2}$
for $j=1,\dots,J$. In particular, since $\Gamma_{z}[h]=h$ and $\Gamma_{z}[\cdot]$
is a $\beta$-contraction, we have 
\begin{align*}
\left\Vert h-\hat{h}_{j}\right\Vert _{2} & \le\left\Vert \Gamma_{z}[h]-\Gamma_{z}[\hat{h}_{j-1}]\right\Vert _{2}+\left\Vert \Gamma_{z}[\hat{h}_{j-1}]-h_{j}\right\Vert _{2}\\
 & \le\beta\left\Vert h-\hat{h}_{j-1}\right\Vert _{2}+\left\Vert \Gamma_{z}[\hat{h}_{j-1}]-h_{j}\right\Vert _{2}.
\end{align*}
Iterating the above gives
\begin{equation}
\left\Vert h-\hat{h}_{J}\right\Vert _{2}=\beta^{J-1}\left\Vert h-\hat{h}_{1}\right\Vert _{2}+\sum_{j=2}^{J}\beta^{J-j}\left\Vert \Gamma_{z}[\hat{h}_{j-1}]-\hat{h}_{j}\right\Vert _{2}.\label{eq:error propagation}
\end{equation}
Equation (\ref{eq:error propagation}) can be considered a special
case of error propagation (\citealp{munos2008finite}). 

Recall that $\hat{h}_{1}$ is an arbitrary initialization. It is thus
straightforward to provide conditions under which $\left\Vert h-\hat{h}_{1}\right\Vert _{2}$
is bounded by some constant $M_{1}$. As for the second term in (\ref{eq:error propagation}),
recall from the discussion in Section \ref{subsec:Approximate-Value-Iteration}
that the minimization problem (\ref{eq:sample AVI}) corresponds to
non-parametric estimation of $\Gamma_{z}[\hat{h}_{j-1}]$ using the
functional class $\mathcal{F}$. Most machine learning methods come
with guarantees on the non-parametric estimation rate $\left\Vert \Gamma_{z}[\hat{h}_{j-1}]-\hat{h}_{j}\right\Vert _{2}$. 

We now describe our assumptions for AVI. Let $\mathcal{X}$ denote
the $d$-dimensional space of $x$, and define $\mathcal{W}^{\gamma,\infty}(\mathcal{X})$
as the H\"{o}lder ball with smoothness parameter $\gamma$: 
\[
\mathcal{W}^{\gamma,\infty}(\mathcal{X}):=\left\{ f:\max_{0<\vert p\vert\le\gamma}\sup_{x\in\mathcal{X}}\left|D^{p}f\right|\le M\right\} .
\]

\begin{asm3} (i) The domain, $\mathcal{X}$, of $x$ is compact and
there exist $M_{0},M<\infty$ such that $\left|h\right|_{\infty}\le M_{0}$
and $h(\cdot,a)\in\mathcal{W}^{\gamma,\infty}(\mathcal{X})$ for each
$a$.\label{3(i)}

(ii) $\left|\hat{h}_{1}\right|_{\infty}\le M_{0}$ and $\left\Vert h-\hat{h}_{1}\right\Vert _{2}\le M_{1}$
for some $M_{1}<\infty$. \label{3(ii)}

(iii) $\left|\Gamma_{z}[f]\right|_{\infty}\le M_{0}$ and $\Gamma_{z}[f](\cdot,a)\in\mathcal{W}^{\gamma,\infty}(\mathcal{X})$
for all $a\in\mathcal{A}$ and $\left\{ f:\left|f\right|_{\infty}\le M_{0}\right\} $.\label{3(iii)}

(iv) The candidate class of functions $\mathcal{F}$ is such that
$\left|f\right|_{\infty}\le M_{0}$ for all $f\in\mathcal{F}.$ Additionally,
consider the non-parametric estimation problem\\
 $\hat{f}=\argmin_{\tilde{f}\in\mathcal{F}}n^{-1}\sum_{i=1}^{n}(y_{i}-\tilde{f}(a_{i},x_{i}))^{2}$,
where $y_{i}$ is compactly supported and $\mathbb{E}[y_{i}\vert a_{i},x_{i}]=f(a_{i},x_{i})$
for some $f\in\mathcal{W}^{\gamma,\infty}(\mathcal{X})$. Then, uniformly
over all $f\in\mathcal{W}^{\gamma,\infty}(\mathcal{X})$, $\mathbb{E}\left[\left\Vert f-\hat{f}\right\Vert _{2}\right]\le Cn^{-c}$
for constants $C<\infty$, $c>0$ independent of $n$, but which may
depend on $M,M_{0},\gamma$. \label{3(iv)}

\end{asm3}

Assumption 3(i) is a standard requirement in non-parametric estimation.
The assumption of $\gamma$-H\"{o}lder continuity is taken from \citet{farrell2021deep}.
Assumption 3(ii) is a mild condition on the initialization $\hat{h}_{1}$.
Assumption 3(iii), which is novel to this paper, is a crucial smoothness
condition requiring the operator $\Gamma_{z}[\cdot](\cdot,a)$ to
map all bounded $f$ onto $\mathcal{W}^{\gamma,\infty}(\mathcal{X})$.
In Online Appendix \ref{subsec:Smoothness-properties-of}, we show
that this is satisfied if $z(a,\cdot)$ and $K(x^{\prime}\vert a,\cdot)$
are $\gamma$-H\"{o}lder continuous.

Assumption 3(iv) is a high-level condition on the machine learning
(ML) method $\mathcal{F}$. The requirement of bounded $f$ implies
that the ML method cannot diverge in the $l_{\infty}$ sense, see
\citet{farrell2021deep} for a discussion of this in the context of
multi-layer perceptrons (MLPs). The second part of Assumption 3(iv)
implies that the ML method is able to non-parametrically approximate
all functions in $\mathcal{W}^{\gamma,\infty}(\mathcal{X})$ at the
rate of at least $n^{-c}$. Most ML methods are proven to satisfy
this. Consider, for instance, the class $\mathcal{F}$ of MLPs of
width $W$ and depth $L$; MLPs and, more generally, Neural Networks
are widely used in RL. The results of \citet{farrell2021deep} imply
that for $W\asymp n^{\frac{d}{2(\gamma+d)}}\ln^{2}n$ and $L\asymp\ln n$,
\[
\sup_{f\in\mathcal{W}^{\gamma,\infty}(\mathcal{X})}\mathbb{E}\left[\left\Vert f-\hat{f}\right\Vert _{2}\right]\le C\left\{ n^{-\frac{\gamma}{2(\gamma+d)}}\ln^{4}n+\sqrt{\frac{\ln\ln n}{n}}\right\} .
\]
Thus, Assumption 3(iv) is satisfied for MLPs. See \citet{biau2012analysis}
for related results on random forests. 

Assumptions 3(iii) and 3(iv) imply that one can estimate $\Gamma[f]$
for any $\left|f\right|_{\infty}\le M_{0}$ at the $n^{-c}$ rate,
i.e., $\sup_{j}\mathbb{E}\left[\left\Vert \Gamma_{z}[\hat{h}_{j-1}]-\hat{h}_{j}\right\Vert _{2}\right]\le Cn^{-c}.$
Combined with (\ref{eq:error propagation}), this proves: 

\begin{thm} \label{Theorem 3} Suppose Assumptions 3(i) to 3(iv)
hold. Then, for all $n$ large enough,
\[
\mathbb{E}\left[\left\Vert h-\hat{h}_{J}\right\Vert _{2}\right]\le\frac{C(1-\beta^{J-1})}{1-\beta}n^{-c}+M_{1}\beta^{J-1}.
\]
\end{thm}

The first term in the expression for $\mathbb{E}\left[\left\Vert h-\hat{h}_{J}\right\Vert _{2}\right]$
from Theorem \ref{Theorem 3} is the statistical rate of estimation
of $h$. The second term is the numerical error, which is seen to
decline exponentially with the number of iterations $J$. Setting
$J\asymp\ln n$ will ensure the numerical error is smaller than the
statistical rate of convergence. The number of iterations can be further
reduced using a good initialization that makes $M_{1}$ small. For
instance, $\hat{h}_{1}$ can be the linear semi-gradient estimator.
This is fast to compute and ensures $M_{1}=o_{p}(1)$. Incidentally,
Theorem \ref{Theorem 3} justifies the use of neural networks for
batch RL; to the best of our knowledge this appears to be new even
in the RL literature. 

Turning to estimation of $\hat{g}$, we again assume cross-fitting
is employed as in Theorem \ref{Theorem 2}, i.e., $\hat{\eta}$ is
computed from one half of the data, and $\hat{g}$ is computed using
AVI on the other half, taking $\hat{\eta}$ as given.

\begin{asm4} (i) Assumptions 3(i) - 3(iv) hold after replacing
$(h,\hat{h}_{1},\Gamma_{z}[\cdot])$ with $(g,\hat{g}_{1},\Gamma_{e}[\cdot])$.\label{4(i)}

(ii) $\hat{g}$ is estimated from a cross-fitting procedure. The conditional
choice probability function satisfies $\eta(a,\cdot)\in\mathcal{W}^{\gamma,\infty}(\mathcal{X})$
for all $a$, and $\eta(a,x)\ge\delta>0$, where $\delta$ is independent
of $a,x$. Additionally, $\left\Vert \eta-\hat{\eta}\right\Vert _{2}^{2}=o_{p}(n^{-1/2})$
and with probability approaching one, $\hat{\eta}(a,\cdot)\in\mathcal{W}^{\gamma,\infty}(\mathcal{X})$
for all $a$. \label{4(ii)}\end{asm4}

Assumption 4(ii) is similar to Assumption 2(vi) with the additional
requirement that $\eta$ and $\hat{\eta}$ be $\alpha$-H\"{o}lder
continuous, the latter with probability approaching one. Note that
$\hat{\eta}$ would be $\alpha$-H\"{o}lder continuous if the non-parametric
estimator consistently estimates not only $\eta$ but also its first
$\alpha$ derivatives. 

Let $\tilde{g}$ denote the fixed point of $\tilde{\Gamma}_{e}[f](a,x):=\beta\mathbb{E}[e(a^{\prime},x^{\prime};\hat{\eta})+f(a^{\prime},x^{\prime})\vert a,x]$.
Decompose 
\[
\left\Vert g-\hat{g}_{J}\right\Vert _{2}\le\left\Vert \tilde{g}-\hat{g}_{J}\right\Vert _{2}+\left\Vert g-\tilde{g}\right\Vert _{2}.
\]
Since $\hat{\eta}(a,\cdot)\in\mathcal{W}^{\gamma,\infty}(\mathcal{X})$,
the first term, $\left\Vert \tilde{g}-\hat{g}_{J}\right\Vert _{2}$,
can be bounded by a similar rate as in Theorem \ref{Theorem 3} (recall
that we take $\hat{\eta}$ as given under cross-fitting). As for the
second term, $\left\Vert g-\tilde{g}\right\Vert _{2}$, observe that
under Assumption 4,
\begin{align*}
\left\Vert g-\tilde{g}\right\Vert _{2} & =\beta\left\Vert \mathbb{E}[\ln\eta(a,x)-\ln\hat{\eta}(a^{\prime},x^{\prime})\vert a,x]\right\Vert _{2}\lesssim\left\Vert \eta-\hat{\eta}\right\Vert _{2}^{2},
\end{align*}
where the `$\lesssim$' holds with probability approaching one, and
uses $\partial_{\eta}\mathbb{E}\left[\left.\ln\eta(a^{\prime},x^{\prime})\right|x^{\prime}\right]=0$,
see the discussion following (\ref{eq:orthogonality for =00005Ceta}).
Since $\left\Vert \eta-\hat{\eta}\right\Vert _{2}^{2}=o_{p}(n^{-1/2})$,
the above proves:

\begin{thm} \label{Theorem 4} Suppose Assumptions 4(i) to 4(ii)
hold. Then, with probability approaching one, 
\[
\left\Vert g-\hat{g}_{J}\right\Vert _{2}\le\frac{C(1-\beta^{J-1})}{1-\beta}n^{-c}+M_{1}\beta^{J-1}+o(n^{-1/2}).
\]
\end{thm}

\subsection{Estimation of structural parameters\label{subsec:Continuous-states-and}}

Estimation of $h(a,x)$ and $g(a,x)$ is inherently non-parametric.
This is because $h(a,x)$ and $g(a,x)$ are functions of two non-parametric
terms: the choice probabilities $\eta(a,x)$, and the transition densities
$K(x^{\prime}\vert a,x)$. The TD estimators implicitly take both
into account. Under discrete states, the first-step estimation error
does not affect the rates of convergence of structural parameters
(see Online Appendix \ref{subsec:Discrete-states}). With continuous
states, however, the first-step non-parametric estimation error does
affect the estimation of $\theta^{*}$ to a first order when using
the PMLE criterion (\ref{eq:estimate of theta}). This is because
the estimates for $K(x^{\prime}\vert a,x)$ and $\theta^{*}$ are
not orthogonal under a PMLE, which extends to the lack of orthogonality
between the estimates for $h(a,x),g(a,x)$ and $\theta^{*}$. Consequently,
the PMLE estimator with plug-in values of $\hat{h}$ and $\hat{g}$
will converge at slower than parametric rates.

We can recover $\sqrt{n}$-consistent estimation by adjusting the
PMLE criterion to account for the first-stage estimation of $h$ and
$g$. Denote $(\tilde{{\bf a}},\tilde{{\bf x}}):=(a,x,a^{\prime},x^{\prime})$
and $m(a,x;\theta,h,g):=\partial_{\theta}\ln\pi(a,x;\theta,h,g)$,
where
\[
\pi(a,x;\theta,h,g):=\frac{\exp\left\{ h(a,x)\theta+g(a,x)\right\} }{\sum_{\breve{a}}\exp\left\{ h(\breve{a},x)\theta+g(\breve{a},x))\right\} }.
\]
The PMLE estimator with plug-in estimates solves $\mathbb{E}_{n}[m(a,x;\theta,\hat{h},\hat{g})]=0$,
but this is not robust to estimation of $h,g$. Let
\[
V(a,x;\theta,h,g):=h(a,x)\theta+g(a,x)
\]
denote the continuation value given $(a,x)$. Also, define $\lambda(a,x;\theta)$
as the fixed point of the `backward' dynamic programming operator
\begin{equation}
\Gamma^{\dagger}[f](a,x):=-\psi(a,x;\theta,h,g)+\beta\mathbb{E}\left[f(a^{-\prime},x^{-\prime})\vert a,x\right],\label{eq:adjustment term 1}
\end{equation}
where $(a^{-\prime},x^{-\prime})$ denotes the past actions and states
preceding $(a,x)$, and 
\begin{equation}
\psi(a,x;\theta,h,g)=-\sum_{\tilde{a}\in\mathcal{A}}\left\{ \mathbb{I}(a=\tilde{a})-\pi(\tilde{a},x;\theta,h,g)\right\} h(\tilde{a},x).\label{eq:derivate wrt h}
\end{equation}
In Online Appendix \ref{subsec:Verification-of-the}, we show that
the locally robust moment corresponding to \linebreak $m(a,x;\theta,h,g)$
is given by 
\begin{align}
\zeta(\tilde{{\bf a}},\tilde{{\bf x}};\theta,h,g) & :=m(a,x;\theta,h,g)-\lambda(a,x;\theta)\left\{ z(a,x)\theta+\beta e(a^{\prime},x^{\prime};\eta)\right.\nonumber \\
 & +\left.\beta V(a^{\prime},x^{\prime};\theta,h,g)-V(a,x;\theta,h,g)\right\} \label{eq:doubly robust moment - non parametric}
\end{align}

The construction of the locally robust moment (\ref{eq:doubly robust moment - non parametric})
is new. But it is infeasible since $\lambda(\cdot),h(\cdot)$ and
$g(\cdot)$ are unknown. However, we can replace these quantities
with consistent estimates. We have already described how to estimate
$h(\cdot),g(\cdot)$. Let $\tilde{\theta}$ denote the plug-in estimator
of $\theta^{*}$ using (\ref{eq:estimate of theta}); note that $\tilde{\theta}$
consistently estimates $\theta^{*}$ but is not efficient. An estimator,
$\hat{\lambda}(\cdot)$, of $\lambda(\cdot)$ can then be obtained
by applying either of our TD estimation methods on (\ref{eq:adjustment term 1}),
with $\tilde{\theta},\hat{h},\hat{g}$ plugged in in place of $\theta,h,g$.
For instance, using AVI, we could obtain iterative approximations
$\{\hat{\lambda}^{(j)},j=1,\dots,J\}$ for $\lambda(\cdot)$ using
\[
\hat{\lambda}^{(j+1)}=\argmin_{f\in\mathcal{F}}\mathbb{E}_{n}\left[\left\Vert -\psi(a,x;\tilde{\theta},\hat{h},\hat{g})+\beta\hat{\lambda}^{(j)}(a^{-\prime},x^{-\prime})-f(a,x)\right\Vert ^{2}\right].
\]
Plugging in $\hat{\lambda}(\cdot),\hat{h},\hat{g}$ into (\ref{eq:doubly robust moment - non parametric}),
we obtain the feasible locally robust moment
\begin{align}
\zeta_{n}(\tilde{{\bf a}},\tilde{{\bf x}};\theta,\hat{h},\hat{g}) & :=m(a,x;\theta,\hat{h},\hat{g})-\hat{\lambda}(a,x;\tilde{\theta})\left\{ z(a,x)\tilde{\theta}+\beta e(a^{\prime},x^{\prime};\hat{\eta})\right.\nonumber \\
 & \hfill\quad+\left.\beta V(a^{\prime},x^{\prime};\tilde{\theta},\hat{h},\hat{g})-V(a,x;\tilde{\theta},\hat{h},\hat{g})\right\} .\label{eq:feasible locally robust moment}
\end{align}
Based on the above, we can obtain a locally robust estimator, $\hat{\theta}$,
as the solution to $\mathbb{E}_{n}[\zeta_{n}(\tilde{{\bf a}},\tilde{{\bf x}};\theta,\hat{h},\hat{g})]=0$.
We recommend obtaining this estimate using cross-fitting, see Section
\ref{subsec:Non-parametric-analysis} for more details. Compared to
the plug-in estimate (\ref{eq:estimate of theta}), our locally robust
estimator requires computation of $\lambda(\cdot)$, but even this
can be avoided when linear semi-gradients are used for estimating
$h,g$ (see Online Appendix \ref{subsec:Construction-of-the}). Solving
$\mathbb{E}_{n}[\zeta_{n}(\tilde{{\bf a}},\tilde{{\bf x}};\theta,\hat{h},\hat{g})]=0$
is also computationally easy; the correction term is a constant, and
$\partial_{\theta}\zeta_{n}(\tilde{{\bf a}},\tilde{{\bf x}};\theta,\hat{h},\hat{g})=\partial_{\theta}m(a,x;\theta,\hat{h},\hat{g})$
is negative definite (as the PMLE criterion is concave), so solving
this is no harder than solving the original moment condition without
a correction term.

\subsubsection{$\sqrt{n}$-consistent estimation\label{subsec:Non-parametric-analysis}}

We focus on the general construction of the locally robust estimator,
$\hat{\theta}$, using (\ref{eq:feasible locally robust moment}).
As mentioned in the previous sub-section, we advocate cross-fitting
to obtain this estimator. This entails computing $\tilde{\theta},\hat{\lambda},\hat{h},\hat{g}$
using one half of the sample, say $\mathcal{N}_{2}$, and plugging
them into the locally robust moment to compute $\hat{\theta}$ as
the solution to $\mathbb{E}_{n}^{(1)}[\zeta_{n}(\tilde{{\bf a}},\tilde{{\bf x}};\theta,\hat{h},\hat{g})]=0$,
where $\mathbb{E}_{n}^{(1)}[\cdot]$ is the empirical expectation
using only observations from the other half of the sample $\mathcal{N}_{1}$.
Following the analysis of \citet{CEINR2018}, it can be shown that
this estimator has the same limiting distribution as the one based
on (\ref{eq:doubly robust moment - non parametric}). In particular,
it achieves parametric rates of convergence. We state the regularity
conditions below (for the remainder of this section we allow $\theta^{*}$
to be vector valued):

\begin{asm5} (i) $\theta^{*}\in\Theta$, a compact set, and $\mathbb{E}\left[\zeta(\tilde{{\bf a}},\tilde{{\bf x}};\theta,h,g)\right]=0\iff\theta=\theta^{*}$.

(ii) There exists a neighborhood, $\mathcal{N}$, of $\theta^{*}$
such that uniformly over $\theta\in\mathcal{N}$ and for $\parallel\!\tilde{h}-h\!\parallel$,
$\parallel\!\tilde{g}-g\!\parallel$ sufficiently small, $\parallel\negmedspace\partial_{\theta}\zeta(\tilde{{\bf a}},\tilde{{\bf x}};\theta,\tilde{h},\tilde{g})-\partial_{\theta}\zeta(\tilde{{\bf a}},\tilde{{\bf x}};\theta^{*},\tilde{h},\tilde{g})\!\parallel\le d(\tilde{{\bf a}},\tilde{{\bf x}})\left\Vert \theta-\theta^{*}\right\Vert $,
where $\mathbb{E}[d(\tilde{{\bf a}},\tilde{{\bf x}})]<\infty$. Furthermore,
$G:=\mathbb{E}\left[\partial_{\theta}\zeta(\tilde{{\bf a}},\tilde{{\bf x}};\theta^{*},h,g)\right]$
is invertible.

(iii) $\forall\ \theta\in\mathcal{N}$, $\left\Vert \lambda(\cdot,\cdot;\theta)-\lambda(\cdot,\cdot;\theta^{*})\right\Vert _{2}\le d(a,x)\left\Vert \theta-\theta^{*}\right\Vert $,
where $\mathbb{E}[d(a,x)]<\infty$.

(iv) $\left\Vert \hat{h}-h\right\Vert _{2}=o_{p}(n^{-1/4})$ and $\left\Vert \hat{g}-g\right\Vert _{2}=o_{p}(n^{-1/4})$.

(v) $\sup_{\theta}\left\Vert \hat{\lambda}(\cdot,\cdot;\theta)-\lambda(\cdot,\cdot;\theta)\right\Vert _{2}=o_{p}(n^{-1/4})$.
\end{asm5}

Assumption 5(i) implies $\theta^{*}$ is identified. Assumption 5(ii)
is a mild regularity condition that is similar to Assumption 5 in
\citet{CEINR2018}. Assumption 5(iii) is satisfied if $\psi(\cdot)$
is uniformly Lipschitz continuous in $\theta$. Assumption 5(iv) follows
from Theorems \ref{Theorem 1}-\ref{Theorem 4} under suitable conditions
on the degree of smoothness of $h,g$. For instance, it is satisfied
for AVI with neural networks if $\gamma\ge d$. Assumption 5(v) requires
$\lambda$ to be estimable at faster than $n^{-1/4}$ rates as well.
If $h,g$ are known, it is straightforward to derive $n^{-1/4}$ rates
as in Theorems \ref{Theorem 1}-\ref{Theorem 4}. For plug-in estimation,
we would need additional assumptions. For instance, we could employ
three-way sample splitting (the first third of the sample is used
to compute $\hat{h},\hat{g}$, which are then plugged into the second
third of the sample to estimate $\lambda$), in which case Assumption
5(v) holds under the regularity condition
\[
\mathbb{E}\left[\sup_{\theta}\left\Vert \psi(a,x;\theta,\hat{h},\hat{g})-\psi(a,x;\theta,h,g)\right\Vert \right]\le\left\Vert \hat{h}-h\right\Vert _{2}+\left\Vert \hat{g}-g\right\Vert _{2}.
\]
We refer to \citet{chernozhukov2018learning} for more details on
three-way splitting. 

We are now ready to state the main result of this section.

\begin{thm} \label{Theorem 5}Suppose that either Assumptions 1,
2 \& 5 or 3-5 hold. Then the estimator, $\hat{\theta}$ of $\theta^{*}$,
based on (\ref{eq:feasible locally robust moment}) is $\sqrt{n}$-consistent,
and satisfies
\[
\sqrt{n}(\hat{\theta}-\theta^{*})\Longrightarrow N(0,V),
\]
where $V=\left(G^{\intercal}\Omega^{-1}G\right)^{-1}$, with $\Omega:=\mathbb{E}\left[\zeta(\tilde{{\bf a}},\tilde{{\bf x}};\theta^{*},h,g)\zeta(\tilde{{\bf a}},\tilde{{\bf x}};\theta^{*},h,g)^{\intercal}\right]$.
\end{thm}

The proof of the above theorem follows by verifying the regularity
conditions of \citet[Theorem 9]{CEINR2018}. Since these are more
or less straightforward to verify given our previous results, we omit
the details. For inference on $\hat{\theta}$, the covariance matrix
$V$ can be estimated as $\hat{V}=\left(\hat{G}^{\intercal}\hat{\Omega}^{-1}\hat{G}\right)^{-1},$
where
\begin{align*}
\hat{G} & =\frac{1}{n(T-1)}\sum_{i=1}^{n}\sum_{t=1}^{T-1}\frac{\partial\zeta_{n}(a_{it},x_{it},a_{it+1},x_{it+1};\hat{\theta},\hat{h},\hat{g})}{\partial\theta^{\intercal}},\quad\textrm{and}\\
\hat{\Omega} & =\frac{1}{n(T-1)}\sum_{i=1}^{n}\sum_{t=1}^{T-1}\zeta_{n}(a_{it},x_{it},a_{it+1},x_{it+1};\hat{\theta},\hat{h},\hat{g})\zeta_{n}(a_{it},x_{it},a_{it+1},x_{it+1};\hat{\theta},\hat{h},\hat{g})^{\intercal}.
\end{align*}
\citet{CEINR2018} provide conditions under which $\hat{V}$ is consistent
for $V$; these are straightforward to translate to our context but
we omit them for brevity. Alternatively, one could employ the bootstrap. 

\section{Estimation of dynamic discrete games \label{sec:Estimation-of-dynamic}}

So far we have considered applications of our algorithm to single-agent
models, where we have argued that there are substantial computational
and statistical gains from using our procedure. These gains are magnified
when extended to the estimation of dynamic discrete games.

Our setup is based on \citet{AguirregabiriaMira2010}. We assume a
single Markov-Perfect-Equilibrium setup where multiple players $i=1,2,\dots,N$
play against each other in $M$ different markets. Each player chooses
among $A$ mutually exclusive actions to maximize an infinite horizon
objective. We observe the state of play for $T$ time periods, where
both $T$ and the number of players $N$ are assumed fixed while $M\to\infty$.
Utility of the players in any time period is affected by the actions
of all the others, and a set of states $x$ that are observed by all
players. The per-period utility is denoted by $z_{i}(a_{i},a_{-i},x)^{\intercal}\theta^{*}+e_{i}$
for each player $i$, for some finite-dimensional parameter $\theta^{*}$,
where $a_{i}$ denotes player $i$'s action, $a_{-i}$ denotes the
actions of all other players and $e_{i}$ is an idiosyncratic error
term. As in Section \ref{sec:Temporal-difference-learning-alg}, we
take $\theta^{*}$ to be scalar to simplify the notation; all our
results continue to hold for vector valued $\theta$, as long as each
dimension is treated separately. Evolution of the states in the next
period is determined by the transition density $K(x^{\prime}\vert a,x)$
where ${\bf a}:=(a_{1},\dots,a_{N})$ denotes the actions of all the
players. We denote by $x_{tm}$ the state at market $m$ in time period
$t,$ by ${\bf a}_{tm}$ the vector of actions by all players at time
$t$ in market $m$, and by $a_{itm}$ the action of player $i$ at
time $t$ in market $m$. We also let $P_{i}(a_{i}\vert x_{t})$ denote
the choice probability of player $i$ taking action $a_{i}$ when
the state is $x_{t}$ , and define $e(a_{i},x):=\gamma-\ln P_{i}(a_{i}\vert x)$. 

As in the single agent case, the parameters $\theta^{*}$ can be obtained
as solutions to the pseudo-likelihood function: 
\begin{equation}
Q(\theta)=\sum_{i=1}^{N}\sum_{m=1}^{M}\sum_{t=1}^{T-1}\log\frac{\exp\left\{ h_{i}(a_{itm},x_{tm})\theta+g_{i}(a_{itm},x_{tm})\right\} }{\sum_{a}\exp\left\{ h_{i}(a,x_{tm})\theta+g_{i}(a,x_{tm})\right\} },\label{eq:Game optimization problem for =00005Ctheta}
\end{equation}
where $h_{i}(.)$ and $g_{i}(.)$ are now player-specific, and given
by
\begin{align}
h_{i}(a_{i},x) & =\mathbb{E}[z_{i}(a,x)\vert a_{i},x]+\beta\mathbb{E}\left[h(a^{\prime},x^{\prime})\vert a_{i},x\right],\label{eq:Game-Recursion-1}\\
g_{i}(a_{i},x) & =\mathbb{E}\left[e(a^{\prime},x^{\prime})+\beta g(a^{\prime},x^{\prime})\vert a_{i},x\right].\nonumber 
\end{align}
In contrast to (\ref{eq:Recursion}) in the single agent case, the
expectation now averages over the actions of the other players as
well.

Previous literature estimates $\theta^{*}$ using a two-step procedure:
In the first step, the conditional choice probabilities $P_{i}(a_{i}\vert x_{t})$
are calculated non-parametrically. These, along with estimates of
$K(.)$ are then used to recursively solve for $h_{i}(.)$ and $g_{i}(.)$
using equation (\ref{eq:Game-Recursion-1}). This step requires integrating
over the actions of all the other players. Finally, given the estimated
values of $h_{i}(.)$ and $g_{i}(.)$, the parameter $\theta$ is
estimated through either pseudo-likelihood (\citealp{AguirregabiriaMira2007})
or minimum distance estimation (\citealp{PesendorferSchmidt-Dengler2008}). 

By contrast, our algorithm is a straightforward extension of those
suggested in earlier sections for single-agent models. Let $\hat{\eta}_{i}(a_{i},x)$
denote a non-parametric estimate of the choice probabilities for player
$i$ and denote $e(a_{i},x;\hat{\eta}_{i})=\gamma-\ln\hat{\eta}_{i}(a_{i},x)$.
We apply our TD methods on the recursion (\ref{eq:Game-Recursion-1}),
separately for each player. The linear semi-gradient estimates are
given by $\hat{h}_{i}(a_{i},x)=\phi(a_{i},x)^{\intercal}\hat{\omega}_{i}$
and $\hat{g}_{i}(a_{i},x)=r(a_{i},x)^{\intercal}\hat{\xi}$, where
\begin{align}
\hat{\omega}_{i} & =\mathbb{E}_{n}\left[\phi(a_{i},x)\left(\phi(a_{i},x)-\beta\phi(a_{i}^{\prime},x^{\prime})\right)^{\intercal}\right]^{-1}\mathbb{E}_{n}\left[\phi(a_{i},x)z_{i}(a_{i},a_{-i},x)\right],\nonumber \\
\hat{\xi}_{i} & =\mathbb{E}_{n}\left[r(a_{i},x)\left(r(a_{i},x)-\beta r(a_{i}^{\prime},x^{\prime})\right)^{\intercal}\right]^{-1}\mathbb{E}_{n}\left[\beta r(a_{i},x)e(a_{i}^{\prime},x^{\prime};\hat{\eta}_{i})\right],\label{eq:Game -value fn updates-1}
\end{align}
and for any function $f(\cdot)$, we define 
\begin{equation}
\mathbb{E}_{n}[f({\bf a},x,{\bf a}^{\prime},x^{\prime})]:=\frac{1}{M(T-1)}\sum_{m=1}^{M}\sum_{t=1}^{T-1}f({\bf a}_{tm},x_{tm},{\bf a}_{t+1m},x_{t+1m}).\label{eq:game empirical expectation}
\end{equation}
Similarly, the AVI iterations for $h_{i}(\cdot),g_{i}(\cdot)$ are
given by 
\begin{align}
\hat{h}_{i}^{(j+1)} & =\argmin_{f\in\mathcal{F}}\mathbb{E}_{n}\left[\left\Vert z_{i}(a_{i},a_{-i},x)+\beta\hat{h}_{i}^{(j)}(a_{i}^{\prime},x^{\prime})-f(a_{i},x)\right\Vert ^{2}\right],\nonumber \\
\hat{g}_{i}^{(j+1)} & =\argmin_{f\in\mathcal{F}}\mathbb{E}_{n}\left[\left\Vert \beta e(a_{i}^{\prime},x^{\prime};\hat{\eta}_{i})+\beta\hat{g}_{i}^{(j)}(a_{i}^{\prime},x^{\prime})-f(a_{i},x)\right\Vert ^{2}\right].\label{eq: Game -value fn updates-2}
\end{align}
If the players are symmetric ($z_{i}(a_{i},a_{-i},x)$ does not depend
on player $i$) we can obtain computationally faster and more precise
estimates by pooling across players.

Importantly, neither of the estimation strategies (\ref{eq:Game -value fn updates-1})
nor (\ref{eq: Game -value fn updates-2}) require partialling out
other players' actions, leading to a tremendous reduction of computation.
Intuitively, the procedures take expectations over other players'
actions `internally' using the empirical distribution. The non-parametric
estimates $\hat{h}_{i},\hat{g}_{i}$ can be plugged into the PMLE
criterion (\ref{eq:Game optimization problem for =00005Ctheta}) to
obtain an estimate for $\theta$ as 
\begin{align*}
\hat{\theta} & =\argmax_{\theta}\sum_{i}Q_{i}(\theta,\hat{h}_{i},\hat{g}_{i}),\ \textrm{where}\\
Q_{i}(\theta,\hat{h}_{i},\hat{g}_{i}) & :=\sum_{m=1}^{M}\sum_{t=1}^{T-1}\log\frac{\exp\left\{ \hat{h}_{i}(a_{itm},x_{tm})\theta+\hat{g}_{i}(a_{itm},x_{tm})\right\} }{\sum_{a}\exp\left\{ \hat{h}_{i}(a,x_{tm})\theta+\hat{g}_{i}(a,x_{tm})\right\} }.
\end{align*}
It is straightforward to construct a locally robust estimator for
$\theta$ in analogy with that for single-agent models. We describe
this in Online Appendix \ref{sec:Locally-robust-estimators-dynamic-games}.

By the same reasoning as in Online Appendix \ref{subsec:Discrete-states},
it possible to show that with discrete states, $h_{i}(.)$ and $g_{i}(.)$
are numerically identical to the estimates obtained by plugging in
cell estimates $\hat{P}_{j}(\cdot\vert x)$ and $\hat{K}(.)$ in (\ref{eq:Game-Recursion-1}).
This implies the psuedo-likelihood with plug-in estimates for $h(.)$
and $g(.)$ is not efficient even with discrete states, as discussed
by \citet{AguirregabiriaMira2007}. However the values of $h(.)$
and $g(.)$ can be plugged into other, more efficient objectives,
such as our locally robust estimator or the minimum distance estimator
of \citet{PesendorferSchmidt-Dengler2008}. With continuous states,
one would need to employ locally robust corrections even for the minimum
distance estimator to recover parametric rates of convergence for
$\theta$. The locally robust correction term can be constructed in
a similar way as that for the PMLE criterion. 

\section{Simulations\label{sec:Simulations}}

We run Monte Carlo simulations to test our estimation methods. We
start by presenting results for a DDC model using both the linear
semi-gradient approach, and the AVI approach with random forests as
the prediction method. For each of these methods, we compare findings
from the locally robust version of our estimators to those obtained
without correction. Our simulations for this part are based on the
firm entry problem described in \citet{AgMag2018}. 

In a second set of Monte Carlo simulations, we test our estimation
method for dynamic discrete games. For this part, we present results
from the linear semi-gradient approach using different sets of basis
functions to approximate the value function terms, and employ the
cross-validation method described in Section \ref{subsec:Tuning-parameters}
to select the preferred set of basis functions. Our simulations are
based on the dynamic firm entry game used in \citet{AguirregabiriaMira2007}. 

Finally, we run additional simulations based on the famous \citet{Rust1987}
bus engine replacement problem (see Online Appendix \ref{subsec:Bus-Engine-Replacement}).
Using this model, we also provide simulation results for a case with
permanent unobserved heterogeneity.

\subsection{Firm entry problem\label{subsec:Firm-Entry-Problem}}

Consider the following dynamic firm entry problem described in \citet{AgMag2018}.
A firm decides whether to enter ($a_{t}=1)$ or not enter ($a_{t}=0)$
in a market for $t=1,...,T$ time periods. The payoff when entering
is given by $\Pi_{t}=VP_{t}-FC_{t}-EC_{t}+\varepsilon_{t},$ where
$VP_{t}$, $FC_{t}$ and $EC_{t}$ denote the firm's variable profit,
fixed cost and entry cost, and $\varepsilon_{t}$ is a transitory
shock that follows a logistic distribution. Variable profit is given
by $VP_{t}=(\theta_{0}^{VP}+\theta_{1}^{VP}z_{1t}+\theta_{2}^{VP}z_{2t})\exp(\omega_{t})$,
where $\omega_{t}$ denotes the firm's productivity shock, and $z_{1t}$,
$z_{2t}$ are exogenous state variables affecting the price-cost margin
in the market. The fixed cost is given by $FC_{t}=\theta_{0}^{FC}+\theta_{1}^{FC}z_{3t}$,
and the entry cost is given by $EC_{t}=(\theta_{0}^{EC}+\theta_{1}^{EC}z_{4t})(1-a_{t-1})$,
where $z_{3t}$, $z_{4t}$ are further exogenous state variables,
and $a_{t-1}$ denotes the entry decision in period $t-1$ which is
an endogenous state variable. The payoff of not entering is normalized
to zero. The parameters $\theta^{*}\equiv\left\{ \theta_{0}^{VP},\theta_{1}^{VP},\theta_{2}^{VP},\theta_{0}^{FC},\theta_{1}^{FC},\theta_{0}^{EC},\theta_{1}^{EC}\right\} $
are the structural parameters of interest. The exogenous state variables
$z_{jt}$ and $\omega_{t}$ are continuous and follow AR(1) processes,
where $z_{jt}=\gamma_{0}^{j}+\gamma_{1}^{j}z_{jt-1}+e_{jt}$, and
$\omega_{t}=\gamma_{0}^{\omega}+\gamma_{1}^{\omega}\omega_{t-1}+e_{\omega t}$.
The error terms $e_{jt},e_{\omega t}$ follow normal $N(0,1)$ distributions.
The discount factor $\beta$ is $0.95$. 

To carry out the simulations, we choose values for the structural
parameters $\theta^{*}$ ($\theta_{0}^{VP}=0.5$, $\theta_{1}^{VP}=1.0$,
$\theta_{2}^{VP}=-1.0$, $\theta_{0}^{FC}=1.5$, $\theta_{1}^{FC}=1.0$,
$\theta_{0}^{EC}=1.0$, $\theta_{1}^{EC}=1.0$) and for the autoregressive
processes of $z_{jt}$ and $\omega_{t}$ ($\gamma_{0}^{j}=0.0$, $\gamma_{1}^{j}=0.6$,
$\gamma_{0}^{\omega}=0.2$, $\gamma_{1}^{\omega}=0.6$), and discretize
the exogenous state variables to obtain a transition matrix with a
6-point support following \citet{Tauchen1986}. The resulting dimension
of the state space is $2\times6^{5}=15,552.$ The discretization of
the support is for simulations only; our estimation algorithms treat
these variables as continuous and do not require any prior knowledge
of how they evolve (the knowledge of AR(1) dynamics is also not used).
We iterate on the value function to obtain the vector of choice probabilities
for each combination of the states, and use these to derive the steady-state
distribution of the state variables. Using the steady-state distributions,
we generate data for $3,000$ firms, with $T=2$ time periods. 

\subsubsection{Simulation results - firm entry problem\label{subsec:Simulation-Results-DDC}}

Panel A of Table \ref{tab:table 1-1-1} shows results for 1000 simulations
using the linear semi-gradient method. Panel B of Table \ref{tab:table 1-1-1}
presents results of 250 simulations using the AVI method. Each round
of the simulations begins by generating new data, where the first-period
state variables are drawn from the steady-state distribution. For
the results in Panel A, we parameterize $h(a,x)$ and $g(a,x)$ using
a second order polynomial in the state variables.\footnote{For the $\omega$'s relating to parameters $\theta_{0}^{VP},\theta_{1}^{VP},\theta_{2}^{VP},\theta_{0}^{FC},\theta_{1}^{FC}$,
and for $\xi$, the terms include a constant, the exogenous state
variables and their interactions up to a second order, the player's
binary choice $a_{t}$ and the interactions of $a_{t}$ with all terms
in the exogenous states. Given the set-up of the model, we treat the
interactions $z_{1t}\exp(\omega_{t})$ and $z_{2t}\exp(\omega_{t})$
instead of $z_{1t}$ and $z_{2t}$ themselves as state variables.
In addition to the terms included above, the $\omega$'s relating
to parameters $\theta_{0}^{EC}$ and $\theta_{1}^{EC}$ also contain
the terms $(1-a_{t-1})$ and $(1-a_{t-1})z_{4t}$, respectively. The
total number of terms included is 42 (43 for $\theta_{0}^{EC}$ and
$\theta_{1}^{EC}$).} For the results in Panel B, we approximate $h(a,x)$ and $g(a,x)$
using a random forest, where we iterate the AVI procedure 70 times
for each round of the simulations. For both the linear semi-gradient
and the AVI methods, we estimate the choice probabilities $\eta$
that enter $e(a',x';\eta)$ using a logit model where the explanatory
variables are the same as those used as basis functions in Panel A. 

We present results generated with and without the locally robust correction.
For the results without correction, we obtain estimates for $\theta^{*}$
using (\ref{eq:estimate of theta}). To generate the locally robust
estimates, we use moment equation (\ref{eq:locally robust sample moment})
for the linear semi-gradient method, and moment equation (\ref{eq:feasible locally robust moment})
for the AVI method where we employ a random forest to derive an estimate
for the $\lambda(a,x,\tilde{\theta})$ term contained in the locally
robust moment. As before, the AVI method for estimation of $\lambda(\cdot)$
iterated 70 times. We also use the sample splitting method described
in Section \ref{subsec:Non-parametric-analysis} for the locally robust
estimators, and we obtain the final $\hat{\theta}$ as weighted average
of the $\theta^{*}$ estimates from the two samples.\footnote{We run our simulations on a MacBook Pro with an M1 chip and 16 GB
of RAM. The computation time for one estimation round is about 4 seconds
for the linear semi-gradient method without locally robust correction,
and 14 seconds with locally robust correction. For the AVI method,
the computation time is about 90 seconds without locally robust correction,
and 315 seconds with locally robust correction.} 

Panel A of Table \ref{tab:table 1-1-1} shows that our linear semi-gradient
results are closely centered around the true values. While the locally
robust estimator should in theory be preferable, we find that it produces
results which are similar and if anything have slightly higher mean
squared error than the non-robust version of our algorithm. In fact,
there is very little bias, and the distribution of the estimates under
the non-robust version is already very close to normal, see Online
Appendix \ref{sec:Appendix:Additional-simulations} for the plots
of the finite sample distributions. On the other hand, locally robust
methods are associated with higher variance due to sample splitting.
So there appears to be no gain from the locally robust method here. 

Panel B of Table \ref{tab:table 1-1-1} shows that the AVI method
produces estimates with slightly higher small sample bias than the
linear semi-gradient. In this case, the locally robust method is more
useful as it decreases the bias significantly, especially for estimation
of the entry cost parameters.

\begin{table}
\begin{threeparttable}

\caption{Simulations: Firm entry problem}
\label{tab:table 1-1-1}

\begin{tabular}{cccccccc}
\hline 
 &  &  & \multicolumn{1}{c}{} &  &  & \multicolumn{1}{c}{} & \tabularnewline
 &  &  & \multicolumn{2}{c}{not locally robust} &  & \multicolumn{2}{c}{locally robust}\tabularnewline
\cline{4-5} \cline{5-5} \cline{7-8} \cline{8-8} 
 &  &  &  &  &  &  & \tabularnewline
 & DGP &  & TDL & MSE &  & TDL & MSE\tabularnewline
 & (1) &  & (2) & (3) &  & (4) & (5)\tabularnewline
\hline 
\emph{A. Linear semi-gradient} &  &  &  &  &  &  & \tabularnewline
 &  &  &  &  &  &  & \tabularnewline
$\theta_{0}^{VP}$  & 0.5 &  & 0.4898 & 0.0062 &  & 0.5376 & 0.0111\tabularnewline
 &  &  & (0.0781) &  &  & (0.0983) & \tabularnewline
$\theta_{1}^{VP}$  & 1.0 &  & 0.9883 & 0.0064 &  & 1.0650 & 0.0154\tabularnewline
 &  &  & (0.0793) &  &  & (0.1060) & \tabularnewline
$\theta_{2}^{VP}$  & -1.0 &  & -0.9908 & 0.0070 &  & -1.0675 & 0.0168\tabularnewline
 &  &  & (0.0831) &  &  & (0.1107) & \tabularnewline
$\theta_{0}^{FC}$  & 1.5 &  & 1.4905 & 0.0232 &  & 1.5745 & 0.0402\tabularnewline
 &  &  & (0.1521) &  &  & (0.1862) & \tabularnewline
$\theta_{1}^{FC}$  & 1.0 &  & 0.9877 & 0.0183 &  & 1.0446 & 0.0309\tabularnewline
 &  &  & (0.1348) &  &  & (0.1703) & \tabularnewline
$\theta_{0}^{EC}$ & 1.0 &  & 0.9949 & 0.0101 &  & 1.0253 & 0.0137\tabularnewline
 &  &  & (0.1002) &  &  & (0.1141) & \tabularnewline
$\theta_{1}^{EC}$ & 1.0 &  & 0.9978 & 0.0273 &  & 1.0529 & 0.0395\tabularnewline
 &  &  & (0.1654) &  &  & (0.1917) & \tabularnewline
\hline 
\emph{B. AVI} &  &  &  &  &  &  & \tabularnewline
 &  &  &  &  &  &  & \tabularnewline
$\theta_{0}^{VP}$  & 0.5 &  & 0.4183 & 0.0134 &  & 0.3903 & 0.0200\tabularnewline
 &  &  & (0.0823) &  &  & (0.0896) & \tabularnewline
$\theta_{1}^{VP}$  & 1.0 &  & 1.1043 & 0.0174 &  & 1.0608 & 0.0120\tabularnewline
 &  &  & (0.0810) &  &  & (0.0914) & \tabularnewline
$\theta_{2}^{VP}$  & -1.0 &  & -1.1080 & 0.0189 &  & -1.0663 & 0.0145\tabularnewline
 &  &  & (0.0856) &  &  & (0.1006) & \tabularnewline
$\theta_{0}^{FC}$  & 1.5 &  & 1.5453 & 0.0280 &  & 1.4471 & 0.0334\tabularnewline
 &  &  & (0.1614) &  &  & (0.1751) & \tabularnewline
$\theta_{1}^{FC}$  & 1.0 &  & 1.1209 & 0.0337 &  & 1.0707 & 0.0319\tabularnewline
 &  &  & (0.1384) &  &  & (0.1645) & \tabularnewline
$\theta_{0}^{EC}$ & 1.0 &  & 1.1388 & 0.0314 &  & 1.0430 & 0.0241\tabularnewline
 &  &  & (0.1104) &  &  & (0.1495) & \tabularnewline
$\theta_{1}^{EC}$ & 1.0 &  & 1.2700 & 0.1145 &  & 1.1466 & 0.0847\tabularnewline
 &  &  & (0.2043) &  &  & (0.2519) & \tabularnewline
\hline 
\end{tabular}

\footnotesize

Notes: The table reports results with 3000 firms. Panel A is based
on 1000 simulations, Panel B on 250 simulations. Column (1) shows
the true parameter values in the model. Columns (2) and (4) report
the empirical mean and standard deviation for the estimated parameters.
Columns (2)-(3) are based on the estimation method without correction
function, columns (4)-(5) report results using the locally robust
estimator. The mean squared errors are reported in columns (3) and
(5), respectively. 

\end{threeparttable}
\end{table}

\subsubsection{Comparison with existing methods}

We compare our findings from Section \ref{subsec:Firm-Entry-Problem}
to those that would be obtained with a standard CCP estimator where
the state variables are discretized and the transition and choice
probabilities are estimated using cell values. 

We start our simulations by generating data as outlined in Section
\ref{subsec:Firm-Entry-Problem}. We then discretize the state space
by creating dummy variables for each state variable $z_{1t},z_{2t},z_{3t},z_{4t}$
and $\exp(\omega_{t})$. To make the estimation feasible, the state
space needs to be restricted further. A common approach is to use
K-means clustering here, but this is not appropriate in the given
simulation setting where the state variables are independent by construction.
We therefore restrict the state space grid by combining variables
$z_{1t}$ and $z_{2t}$ into a binary variable taking value one whenever
both individual dummies take value one. The resulting state space
consists of four binary variables, implying 16 cells in the state
space grid. We try alternative feasible ways of discretizing the state
space, but find that these do not lead to improvements over the chosen
method. We run $1000$ simulations, and the results are shown in Table
\ref{Table D1}. 

It can be seen that, compared to the results from Table \ref{tab:table 1-1-1},
the CCP estimator leads to substantially larger bias in some of the
estimated parameters. Column (4) shows that the corresponding mean
squared errors are large and generally exceed those obtained using
our estimators. This is particularly true for parameters $\theta_{1}^{VP}$and
$\theta_{2}^{VP}$, highlighting the challenges inherent in the discretization
of continuous state spaces. Overall, the average mean squared error
increases from $0.01-0.04$ across all parameters in Table \ref{tab:table 1-1-1}
to $0.16$ in Table \ref{Table D1}. These results show that our estimation
methods lead to important improvements over existing methods in models
with as few as five continuous state variables.

\begin{table}
\begin{threeparttable}

\caption{Simulations: Firm entry problem - Comparison with standard CCP}

\label{Table D1}

\begin{tabular}{cccccc}
\hline 
 &  &  & \multicolumn{2}{c}{} & \tabularnewline
 & DGP &  & TDL & bias & MSE\tabularnewline
 & (1) &  & (2) & (3) & (4)\tabularnewline
\hline 
\emph{CCP with} &  &  &  &  & \tabularnewline
\emph{discretized state variables} &  &  &  &  & \tabularnewline
 &  &  &  &  & \tabularnewline
$\theta_{0}^{VP}$ & 0.5 &  & 0.1518 & -0.3482 & 0.1742\tabularnewline
 &  &  & (0.2303) &  & \tabularnewline
$\theta_{1}^{VP}$ & 1.0 &  & 0.7642 & -0.2358 & 0.3704\tabularnewline
 &  &  & (0.5613) &  & \tabularnewline
$\theta_{2}^{VP}$ & -1.0 &  & -0.3826 & 0.6174 & 0.4401\tabularnewline
 &  &  & (0.2428) &  & \tabularnewline
$\theta_{0}^{FC}$ & 1.5 &  & 1.4139 & -0.0861 & 0.0261\tabularnewline
 &  &  & (0.1366) &  & \tabularnewline
$\theta_{1}^{FC}$ & 1.0 &  & 0.8587 & -0.1413 & 0.0404\tabularnewline
 &  &  & (0.1431) &  & \tabularnewline
$\theta_{0}^{EC}$ & 1.0 &  & 0.7788 & -0.2212 & 0.0569\tabularnewline
 &  &  & (0.0894) &  & \tabularnewline
$\theta_{1}^{EC}$ & 1.0 &  & 0.9148 & -0.0852 & 0.0416\tabularnewline
 &  &  & (0.1854) &  & \tabularnewline
\hline 
\end{tabular}

\footnotesize

Notes: The table reports results of 1000 simulations with 3000 firms.
Column (1) shows the true parameter values in the model. Column (2)
reports the empirical mean and standard deviation for the estimated
parameters. Column (3) reports the average bias in the estimated parameters.
The mean squared errors are reported in column (4).

\end{threeparttable}
\end{table}

\subsection{Firm entry game\label{subsec:Firm-Entry-Game}}

Consider the following firm market entry game, which is similar to
that described in \citet{AguirregabiriaMira2007}. There are $i=1,...,5$
firms (players), and we observe their decision to enter ($a_{itm}=1)$
or not enter ($a_{itm}=0)$ in $m=1,...,M$ different markets for
$t=1,...,T$ time periods. Denote a firm's action by $j\in\left\{ 1,0\right\} $.
The payoff of each firm $i$ is affected by the decision of all the
other firms whether to enter, as well as firm $i$'s previous-period
entry decision. Current-period profits when entering are given by
\[
\Pi_{itm}=\theta_{RS}\ln(S_{tm})-\theta_{RN}\ln(1+\sum_{j\neq i}a_{jtm})-\theta_{FC}-\theta_{EC}(1-a_{i(t-1)m})+\varepsilon_{itm},
\]
where $\ln(S_{tm})$ is a measure of consumer market size of market
$m$ in period $t$, and $\varepsilon_{itm}$ is a transitory shock
that follows a logistic distribution. We assume that $\ln(S_{tm})$
is continuous and follows an AR(1) process, where the parameters are
the same across markets: $\ln(S_{tm})=\alpha+\lambda\ln(S_{(t-1)m})+u_{tm}.$
The error term $u_{tm}$ is assumed to follow a normal $N(0,1)$ distribution.
The profit of not entering is normalized to zero, and the discount
factor $\beta$ is $0.95$. The parameters $\theta^{*}\equiv\left\{ \theta_{RS},\theta_{RN},\theta_{FC},\theta_{EC}\right\} $
are the structural parameters of interest. The state variables in
this setting are given by the current market demand variable $S_{tm}$,
as well as the vector of all firms' previous entry decisions $a_{(t-1)m}=\left\{ a_{i(t-1)m}:i=1,...,5\right\} $. 

To carry out the simulations, we choose values for the structural
parameters $\theta^{*}$ ($\theta_{RS}=1,\theta_{RN}=1,\theta_{FC}=1.7,\theta_{EC}=1$),
and for the autoregressive process for log market size ($\alpha=1.5$,
$\lambda=0.5$). We discretize $\ln(S_{tm})$ and obtain a transition
matrix for the discretized variable with a 10-point support following
the method by \citet{Tauchen1986}. As in the Monte Carlo experiments
for the firm entry problem in Section \ref{subsec:Firm-Entry-Problem},
the discretization is for simulations of the data only and we treat
the state variables as continuous in our estimations. We then solve
for the Markov-Perfect-Equilibrium of the game. This is done by finding
the firms' conditional value functions $\nu_{j}(S_{tm},a_{(t-1)m})$
for each of the $2^{5}\times10=320$ possible combinations of the
state variables through repeated iteration, and using these to derive
the equilibrium choice probabilities $p(S_{tm},a_{(t-1)m})$. Based
on the equilibrium probabilities, we compute the equilibrium distribution
of state variables. Using the equilibrium distributions, we generate
data for $1000$ and for $3000$ markets, with $T=2$ time periods.

\subsubsection{Simulation results - firm entry game\label{subsec:Simulation-Results--}}

We present the results of $1000$ simulations based on the linear
semi-gradient method, without employing the locally robust correction.
Each round of the simulations begins by generating new data, where
the first-period state variables are drawn from the steady-state distribution.
In order to assess the sensitivity of our algorithm to different specifications
for the basis functions, we parameterize $h(a,x)$ and $g(a,x)$ using
different sets of polynomials in the state variables. In particular,
we show results where $h(a,x)$ and $g(a,x)$ are approximated using
a second, third or fourth order polynomial.\footnote{For the $\omega'$s relating to parameters $\theta_{RS},\theta_{RN},\theta_{FC}$
and for $\xi$, the terms include a constant, terms up to the second/third/fourth
order in the state variables $\ln(S_{tm})$ and $\ln(1+\sum_{j\neq i}a_{j(t-1)m})$,
the player's binary choice $a_{itm}$ and the interactions of $a_{itm}$
with all terms in the state variables. The total number of terms is
$12/20/30$. In addition to these terms, the $\omega'$s relating
to parameter $\theta_{EC}$ also contain the term $(1-a_{i(t-1)m})a_{itm}$. } For all simulations, the choice probabilities $\eta$ that enter
$e(a',x';\eta)$ are estimated using individual logit models for each
firm, where we use a third order polynomial in the state variables
as explanatory variables. We then estimate the parameters $\omega$
and $\xi$ using equation \ref{eq:Game -value fn updates-1}.\footnote{Given the symmetric set-up of the game, we pool the data across players
in this application.} Finally, we obtain estimates for the $\theta^{*}$ parameters as
the solutions to the pseudo-likelihood function \ref{eq:Game optimization problem for =00005Ctheta}.\footnote{We run our simulations on a MacBook Pro with an M1 chip and 16 GB
of RAM. The approximate computation times for one estimation round
are: 1.5 seconds (second order polynomial, 1000 markets); 3.7 seconds
(second order polynomial, 3000 markets); 2.5 seconds (third order
polynomial, 1000 markets); 6.7 seconds (third order polynomial, 3000
markets); 4 seconds (fourth order polynomial, 1000 markets); 12 seconds
(fourth order polynomial, 3000 markets).}

The results are shown in Table \ref{tab: table 3}. Panels A, B and
C present simulations for the same dataset using different basis functions
to parameterize the value function terms $h(a,x)$ and $g(a,x)$.
Column (2) shows that even with 1000 markets our algorithm produces
parameter estimates that are closely centered around the true values.
The results are generally similar across Panels A to C, although the
bias and mean squared error tends to be lowest for the second order
polynomial, and highest for the fourth order polynomial. This is especially
the case for the parameter on the number of market entrants, $\theta_{RN}$.
To assess these differences formally, we use the cross-validation
procedure described in Section \ref{subsec:Tuning-parameters}. The
procedure is applied to ten random samples of market size 1000, and
we find that the TD error criterion consistently selects the second
order polynomial as the optimal set of basis functions. This suggests
that the proposed cross-validation method provides useful guidance
when choosing a set of basis functions in practice.

In a similar version of the firm entry game, \citet{AguirregabiriaMira2007}
use the NPL algorithm and derive results comparable to ours. Our simulations
rely on slightly larger sample sizes, but note that these differences
are likely due to the fact that our algorithm implicitly estimates
the transition densities for market size and therefore does not rely
on the true transition density to be used in the estimation. For a
direct comparison of our results with those obtained using the NPL
algorithm, one would need to obtain a non-parametric estimate of the
transition density which is not trivial in practice.

As expected, columns (4) and (5) show that increasing the market size
generally reduces the small sample bias in the estimated parameters,
and leads to a fall in the empirical standard deviations. In addition
to being smaller, the mean squared error across Panels A to C is also
more similar across the three sets of basis functions. As before,
we employ the cross-validation method described above to compare these
specifications more formally. In line with the estimation results,
we find that all three polynomials now produce very similar sets of
mean squared TD error, even though the second order polynomial continues
to be the one that is selected by the criterion.\footnote{As before, we compute the TD criterion for ten random samples.}
While we view this as further evidence that the proposed cross-validation
method can provide useful guidance to choose a suitable set of basis
functions, more importantly, the small differences in the results
across panels A to C also suggest that our methods prove fairly robust
to this choice in practice.

\begin{table}
\begin{threeparttable}

\caption{Simulations: Firm entry game - Linear semi-gradient}

\label{tab: table 3}

\begin{tabular}{cccccccc}
\hline 
 &  &  &  &  &  &  & \tabularnewline
 & DGP &  & TDL & MSE &  & TDL & MSE\tabularnewline
 & (1) &  & (2) & (3) &  & (4) & (5)\tabularnewline
\hline 
\emph{A. 2nd order polynomial} &  &  & \multicolumn{2}{c}{\emph{1000 markets}} &  & \multicolumn{2}{c}{\emph{3000 markets}}\tabularnewline
 &  &  &  &  &  &  & \tabularnewline
$\theta_{RS}$ (market size) & 1.0 &  & 0.9718 & 0.0263 &  & 0.9847 & 0.0082\tabularnewline
 &  &  & (0.1598) &  &  & (0.0895) & \tabularnewline
$\theta_{RN}$ (n. of entrants) & 1.0 &  & 0.8962 & 0.2915 &  & 0.9586 & 0.0857\tabularnewline
 &  &  & (0.5301) &  &  & (0.2899) & \tabularnewline
$\theta_{FC}$ (fixed cost) & 1.7 &  & 1.7225 & 0.0870 &  & 1.6920 & 0.0262\tabularnewline
 &  &  & (0.2942) &  &  & (0.1617) & \tabularnewline
$\theta_{EC}$ (entry cost) & 1.0 &  & 1.0191 & 0.0042 &  & 1.0226 & 0.0018\tabularnewline
 &  &  & (0.0621) &  &  & (0.0353) & \tabularnewline
\hline 
\emph{B. 3rd order polynomial} &  &  & \multicolumn{2}{c}{\emph{1000 markets}} &  & \multicolumn{2}{c}{\emph{3000 markets}}\tabularnewline
 &  &  &  &  &  &  & \tabularnewline
$\theta_{RS}$ (market size) & 1.0 &  & 0.9149 & 0.0287 &  & 0.9648 & 0.0088\tabularnewline
 &  &  & (0.1466) &  &  & (0.0867) & \tabularnewline
$\theta_{RN}$ (n. of entrants) & 1.0 &  & 0.6905 & 0.3251 &  & 0.8862 & 0.0909\tabularnewline
 &  &  & (0.4792) &  &  & (0.2794) & \tabularnewline
$\theta_{FC}$ (fixed cost) & 1.7 &  & 1.7815 & 0.0806 &  & 1.7135 & 0.0249\tabularnewline
 &  &  & (0.2721) &  &  & (0.1573) & \tabularnewline
$\theta_{EC}$ (entry cost) & 1.0 &  & 1.0173 & 0.0042 &  & 1.0219 & 0.0017\tabularnewline
 &  &  & (0.0622) &  &  & (0.0353) & \tabularnewline
\hline 
\emph{C. 4th order polynomial} &  &  & \multicolumn{2}{c}{\emph{1000 markets}} &  & \multicolumn{2}{c}{\emph{3000 markets}}\tabularnewline
 &  &  &  &  &  &  & \tabularnewline
$\theta_{RS}$ (market size) & 1.0 &  & 0.8642 & 0.0358 &  & 0.9457 & 0.0101\tabularnewline
 &  &  & (0.1318) &  &  & (0.0845) & \tabularnewline
$\theta_{RN}$ (n. of entrants) & 1.0 &  & 0.5075 & 0.4206 &  & 0.8166 & 0.1067\tabularnewline
 &  &  & (0.4222) &  &  & (0.2705) & \tabularnewline
$\theta_{FC}$ (fixed cost) & 1.7 &  & 1.8338 & 0.0810 &  & 1.7337 & 0.0245\tabularnewline
 &  &  & (0.2513) &  &  & (0.1530) & \tabularnewline
$\theta_{EC}$ (entry cost) & 1.0 &  & 1.0159 & 0.0041 &  & 1.0213 & 0.0017\tabularnewline
 &  &  & (0.0620) &  &  & (0.0352) & \tabularnewline
\hline 
\end{tabular}

\footnotesize

Notes: The table reports results for 1000 simulations. Panels A, B
and C use different sets of basis functions to parameterize $h(a,x)$
and $g(a,x)$. Column (1) shows the true parameter values in the model.
Columns (2) and (4) report the empirical mean and standard deviations
for the estimated parameters, based on a sample of 1000 and 3000 markets,
respectively. The mean squared errors are reported in columns (3)
and (5). All results are based on the estimation method without correction
function.

\end{threeparttable}
\end{table}

\section{Conclusions\label{sec:Conclusions}}

We propose two new estimators for DDC models which overcome previous
computational and statistical limitations by combining traditional
CCP estimation approaches with the idea of TD learning from the RL
literature. The first approach, linear semi-gradient, makes use of
simple matrix inversion techniques, is computationally very cheap
and therefore fast. The second approach, Approximate Value Iteration,
can be easily combined with any machine learning method devised for
prediction. Unlike previous estimation methods, our methods are able
to easily handle continuous and/or high-dimensional state spaces in
settings where a finite dependence property does not hold. This is
of particular importance for the estimation of dynamic discrete games.
We also propose a locally robust estimator to account for the non-parametric
estimation in the first stage. We prove the statistical properties
of our estimator and show that it is consistent and converges at parametric
rates. A range of Monte Carlo simulations using a dynamic firm entry
problem, a dynamic firm entry game and a version of the famous \citet{Rust1987}
engine replacement problem show that the proposed algorithms work
well in practice. 

\bibliographystyle{IEEEtranSN}
\bibliography{References_1}

\newpage{}

\appendix

\section{Proofs of main results\label{sec:Appendix:A}}

In what follows we drop the functional argument $(a,x)$ when the
context is clear and denote $f^{\prime}\equiv f(a^{\prime},x^{\prime})$
for different functions $f$. 

\begin{lem} \label{Lemma 1}There exists a unique fixed point to
the operator $P_{\phi}\Gamma_{z}$. If Assumption 1(i) holds, this
fixed point is given by $\phi^{\intercal}\omega^{*}$, where $\omega^{*}$
is such that 
\begin{equation}
\mathbb{E}\left[\phi\left(z+\beta\phi^{\prime\intercal}\omega^{*}-\phi^{\intercal}\omega^{*}\right)\right]=0.\label{eq:Lemma 1 - 1}
\end{equation}
\end{lem} 
\begin{proof}
First, note that $\Gamma_{z}$, and therefore $P_{\phi}\Gamma_{z}$,
are both contraction maps with the contraction factor $\beta$. This
implies $P_{\phi}\Gamma_{z}$ has a unique fixed point. Clearly, this
fixed point must lie in the space $\mathcal{L}_{\phi}$. Let us denote
this as $\phi^{\intercal}\omega^{*}$. Now, for any function $f\in\mathcal{L}_{\phi}$,
\begin{align*}
P_{\phi}\Gamma_{z}[f]-f & =\phi^{\intercal}\mathbb{E}[\phi\phi^{\intercal}]^{-1}\mathbb{E}\left[\phi\left(z+\beta f^{\prime}\right)\right]-\phi^{\intercal}\mathbb{E}[\phi\phi^{\intercal}]^{-1}\mathbb{E}[\phi f]\\
 & =\phi^{\intercal}\mathbb{E}[\phi\phi^{\intercal}]^{-1}\mathbb{E}\left[\phi\left(z+\beta f^{\prime}-f\right)\right].
\end{align*}
Since $\phi^{\intercal}\omega^{*}$ is the fixed point, we must have
\[
\phi^{\intercal}\mathbb{E}[\phi\phi^{\intercal}]^{-1}\mathbb{E}\left[\phi\left(z+\beta\phi^{\prime\intercal}\omega^{*}-\phi^{\intercal}\omega^{*}\right)\right]=0.
\]
But $\phi$ is linearly independent and $E[\phi\phi^{\intercal}]^{-1}$
is non-singular, by Assumption 1(i). Hence it must be the case that
$\mathbb{E}\left[\phi\left(z+\beta\phi^{\prime\intercal}\omega^{*}-\phi^{\intercal}\omega^{*}\right)\right]=0.$ 
\end{proof}
For the next lemma, we use the following definition: a square, possibly
asymmetric, matrix $A$ is said to be negative definite with the coefficient
$\bar{\lambda}(A)$ if 
\[
\sup_{\vert w\vert=1}w^{\intercal}Aw\le\bar{\lambda}(A)<0.
\]
For a symmetric negative-definite matrix, $\bar{\lambda}(A)=\textrm{maxeig}(A)$,
where $\textrm{maxeig}(\cdot)$ is the maximal eigenvalue. We can
similarly define a positive-definite matrix with the coefficient $\underline{\lambda}(A)$.
If the latter is symmetric, then $\underline{\lambda}(A)=\textrm{mineig}(A)$.
Note that under our definition, if $A$ is negative definite, it is
also invertible. This holds even if the matrix is asymmetric, see
e.g \citet{johnson1970positive}.

\begin{lem} \label{Lemma 2}Under Assumption 1(i), the matrix $A:=\mathbb{E}\left[\phi\left(\beta\phi^{\prime}-\phi\right)^{\intercal}\right]$
is negative definite with $\bar{\lambda}(A)\le-(1-\beta)\underline{\lambda}(\mathbb{E}[\phi\phi^{\intercal}])$,
and is therefore invertible. \end{lem}
\begin{proof}
The proof is adapted from Tsitsiklis and van Roy (1997). Recall the
definition of $\phi^{\intercal}\omega^{*}$ as the fixed point to
$P_{\phi}\Gamma_{z}[\cdot]$ from Lemma 1. We now show that 
\[
(\omega-\omega^{*})^{\intercal}A(\omega-\omega^{*})\le-(1-\beta)\underline{\lambda}(\mathbb{E}[\phi\phi^{\intercal}])\left|\omega-\omega^{*}\right|^{2}\ \forall\ \omega\in\mathbb{R}^{k}.
\]
Observe that 
\begin{align*}
A(\omega-\omega^{*}) & =\mathbb{E}\left[\phi\left(z+\beta\phi^{\prime\intercal}\omega-\phi^{\intercal}\omega\right)\right]-\mathbb{E}\left[\phi\left(z+\beta\phi^{\prime\intercal}\omega^{*}-\phi^{\intercal}\omega^{*}\right)\right]\\
 & =\mathbb{E}\left[\phi\left(z+\beta\phi^{\prime\intercal}\omega-\phi^{\intercal}\omega\right)\right],
\end{align*}
since the second expression in the first equation is $0$. Now, 
\begin{align*}
\mathbb{E}\left[\phi\left(z+\beta\phi^{\prime\intercal}\omega-\phi^{\intercal}\omega\right)\right] & =\mathbb{E}\left[\phi(a,x)\left(z(a,x)+\beta\mathbb{E}\left[\phi(a^{\prime},x^{\prime})^{\intercal}\omega\vert a,x\right]-\phi(a,x)^{\intercal}\omega\right)\right]\\
 & =\mathbb{E}\left[\phi\left(\Gamma_{z}[\phi^{\intercal}\omega]-\phi^{\intercal}\omega\right)\right]\\
 & =\mathbb{E}\left[\phi\left(P_{\phi}\Gamma_{z}[\phi^{\intercal}\omega]-\phi^{\intercal}\omega\right)\right],
\end{align*}
where the last equality holds since $\mathbb{E}\left[\phi(I-P_{\phi})[f]\right]=0$
for all $f$. We thus have
\begin{align*}
(\omega-\omega^{*})^{\intercal}A(\omega-\omega^{*}) & =\mathbb{E}\left[\left(\omega^{\intercal}\phi-\omega^{*\intercal}\phi\right)\left(P_{\phi}\Gamma_{z}[\phi^{\intercal}\omega]-\phi^{\intercal}\omega\right)\right]\\
 & =\mathbb{E}\left[\left(\omega^{\intercal}\phi-\omega^{*\intercal}\phi\right)\left(P_{\phi}\Gamma_{z}[\phi^{\intercal}\omega]-\phi^{\intercal}\omega^{*}\right)\right]-\left\Vert \phi^{\intercal}\omega-\phi^{\intercal}\omega^{*}\right\Vert _{2}^{2}\\
 & \le\left\Vert \phi^{\intercal}\omega-\phi^{\intercal}\omega^{*}\right\Vert _{2}\cdot\left\Vert P_{\phi}\Gamma_{z}[\phi^{\intercal}\omega]-\phi^{\intercal}\omega^{*}\right\Vert _{2}-\left\Vert \phi^{\intercal}\omega-\phi^{\intercal}\omega^{*}\right\Vert _{2}^{2}.
\end{align*}
Since $P_{\phi}\Gamma_{z}[\cdot]$ is a contraction mapping with contraction
factor $\beta$, it follows
\[
\left\Vert P_{\phi}\Gamma_{z}[\phi^{\intercal}\omega]-\phi^{\intercal}\omega^{*}\right\Vert _{2}=\left\Vert P_{\phi}\Gamma_{z}[\phi^{\intercal}\omega]-P_{\phi}\Gamma_{z}\left[\phi^{\intercal}\omega^{*}\right]\right\Vert _{2}\le\beta\left\Vert \phi^{\intercal}\omega-\phi^{\intercal}\omega^{*}\right\Vert _{2}.
\]
In view of the above,
\begin{align*}
(\omega-\omega^{*})^{\intercal}A(\omega-\omega^{*}) & \le-(1-\beta)\left\Vert \phi^{\intercal}\omega-\phi^{\intercal}\omega^{*}\right\Vert _{2}^{2}\\
 & =-(1-\beta)(\omega-\omega^{*})^{\intercal}\mathbb{E}[\phi\phi^{\intercal}](\omega-\omega^{*})\\
 & \le-(1-\beta)\underline{\lambda}(\mathbb{E}[\phi\phi^{\intercal}])\left|\omega-\omega^{*}\right|^{2}.
\end{align*}
This completes the proof of the lemma.
\end{proof}
\begin{lem} \label{Lemma 3}Suppose that Assumption 1(i) holds. Then,
\[
\left\Vert h-\phi^{\intercal}\omega^{*}\right\Vert _{2}\le(1-\beta)^{-1}\left\Vert h-P_{\phi}[h]\right\Vert _{2}.
\]
\end{lem}
\begin{proof}
Recall that $h(\cdot,\cdot)$ is the unique fixed point of $\Gamma_{z}$,
and $\phi^{\intercal}\omega^{*}$ is the unique fixed point of $P_{\phi}\Gamma_{z}$.
The operator $\Gamma_{z}$ is a contraction mapping with contraction
factor $\beta$. Furthermore, the projection operator $P_{\phi}$
is linear, and $\left\Vert P_{\phi}[f]\right\Vert _{2}\le\left\Vert f\right\Vert _{2}$
for any function $f$. Thus 
\begin{align*}
\left\Vert h-\phi^{\intercal}\omega^{*}\right\Vert _{2} & \le\left\Vert h-P_{\phi}[h]\right\Vert _{2}+\left\Vert P_{\phi}[h]-P_{\phi}\Gamma_{z}[\phi^{\intercal}\omega^{*}]\right\Vert _{2}\\
 & \le\left\Vert h-P_{\phi}[h]\right\Vert _{2}+\left\Vert h-\Gamma_{z}[\phi^{\intercal}\omega^{*}]\right\Vert _{2}\\
 & =\left\Vert h-P_{\phi}[h]\right\Vert _{2}+\left\Vert \Gamma_{z}[h]-\Gamma_{z}[\phi^{\intercal}\omega^{*}]\right\Vert _{2}\\
 & \le\left\Vert h-P_{\phi}[h]\right\Vert _{2}+\beta\left\Vert h-\phi^{\intercal}\omega^{*}\right\Vert _{2}.
\end{align*}
Rearranging the above expression proves the desired claim.
\end{proof}
For the proofs of Theorems \ref{Theorem 1}-\ref{Theorem 2}, we work
within a more general setting than in the main text, by letting the
distribution of $(a_{it},x_{it})$ be time-varying. Let $P_{t}$ denote
the population distribution of $(a,x)$ at time $t$. Also, let $P$
denote the probability distribution of the process $\{(a_{1},x_{1}),\dots,(a_{T},x_{T})\}$.
Note that $P\equiv P_{1}\times\dots\times P_{T}$. We will denote
$E[\cdot]$ as the expectation over $P$. Furthermore, we use the
$o_{p}(\cdot)$ and $O_{p}(\cdot)$ notations to denote convergence
in probability, and bounded in probability, respectively, under the
probability distribution $P$.

We also need to extend the definitions of $\mathbb{P}$ and $\mathbb{E}[\cdot]$:
Let $\mathbb{P}$ denote the relative frequency of occurrence of $(a,x,a^{\prime},x^{\prime})$
in the data as $n\to\infty$. Let $\mathbb{E}[\cdot]$ denote the
corresponding expectation over $\mathbb{P}$. Note that $P$ is different
from $\mathbb{P}$ since the latter is the distribution of $(a,x,a,x^{\prime})$
after dropping the time index. However, the two are related since
for any function $f$, we have $\mathbb{E}[f(a,x,a^{\prime},x^{\prime})]=(T-1)^{-1}\sum_{t=1}^{T-1}E[f(a_{it},x_{it},a_{it+1},x_{it+1})]$.
These updated definitions of $\mathbb{P}$ and $\mathbb{E}[\cdot]$
are applicable wherever these notations are used in the main text. 

Note that due to the Markov process assumption, the conditional distribution
$P(a_{t+1},x_{t+1}\vert a_{t},x_{t})$ is always independent of $t$
(indeed, one could always consider $t$ as also a part of $x$). Hence,
$\mathbb{P}(a^{\prime},x^{\prime}\vert a,x)\equiv P(a_{t+1},x_{t+1}\vert a_{t},x_{t})$
and $\mathbb{E}[f(a^{\prime},x^{\prime})\vert a,x]\equiv E[f(a_{t+1},x_{t+1})\vert a_{t},x_{t}]$
for all $t$. Also note that time stationarity of $(a_{it},x_{it})$,
if it holds, implies $P_{t}\equiv\mathbb{P}$ and $E_{t}[\cdot]\equiv\mathbb{E}[\cdot]$
for all $t$.

\subsection{Proof of Theorem \ref{Theorem 1}\label{subsec:Proof-of-Theorem1}}

Lemma \ref{Lemma 1} implies $\omega^{*}$ exists. To prove that $\hat{\omega}$
exists, it suffices to show that $\hat{A}:=\mathbb{E}_{n}\left[\phi\left(\beta\phi^{\prime}-\phi\right)^{\intercal}\right]$
is invertible with probability approaching one. Recall that using
our notation, $\hat{A}=(n(T-1))^{-1}\sum_{i}\sum_{t=1}^{T-1}\phi_{it}(\beta\phi_{it+1}-\phi_{it})^{\intercal}$,
while $A=(T-1)^{-1}\sum_{t=1}^{T-1}E[\phi_{it}(\beta\phi_{it+1}-\phi_{it})^{\intercal}$.
We can thus write $\left|\hat{A}-A\right|\le(T-1)^{-1}\sum_{t=1}^{T-1}\left|\hat{A}_{t}-A_{t}\right|$,
where $\hat{A}_{t}:=n^{-1}\sum_{i}\phi_{it}(\beta\phi_{it+1}-\phi_{it})^{\intercal}$
and $A_{t}:=E[\phi_{it}(\beta\phi_{it+1}-\phi_{it})^{\intercal}]$.
By Assumption 1(ii), $\vert\phi(a,x)\vert_{\infty}\le M$ independent
of $k_{\phi}$, so
\begin{align*}
E\left|\hat{A}_{t}-A_{t}\right|^{2} & =E\left|\frac{1}{n}\sum_{i}\phi_{it}\left(\beta\phi_{it+1}-\phi_{it}\right)^{\intercal}-E\left[\phi_{it}\left(\beta\phi_{it+1}-\phi_{it}\right)^{\intercal}\right]\right|^{2}\\
 & \le\frac{1}{n}\sum_{i}E\left|\phi_{it}\left(\beta\phi_{it+1}-\phi_{it}\right)^{\intercal}\right|^{2}\le\frac{k_{\phi}^{2}M^{4}}{n}.
\end{align*}
This proves $\left|\hat{A}_{t}-A_{t}\right|=O_{p}(k_{\phi}/\sqrt{n})$.
But $T$ is fixed, which implies that $\left|\hat{A}-A\right|=O_{p}(k_{\phi}/\sqrt{n})$
as well. We thus obtain $\bar{\lambda}(\hat{A})\le\bar{\lambda}(A)+\left|\hat{A}-A\right|\le\bar{\lambda}(A)+o_{p}(1)$.
Since $\bar{\lambda}(A)<0$, this proves that $\bar{\lambda}(\hat{A})<0$
with probability approaching one, and subsequently, that $\hat{A}$
is invertible. This completes the proof of the first claim.

The second claim follows directly from Lemma \ref{Lemma 3} and Assumption
1(iii).

To prove the last claim, we first show that with probability approaching
one, 
\begin{equation}
\left|\hat{\omega}-\omega^{*}\right|\le C(1-\beta)^{-1}\sqrt{\frac{k_{\phi}}{n}}\label{eq:estimation of omega}
\end{equation}
for some $C<\infty$. Define $b=\mathbb{E}[\phi z]$ and $\hat{b}=\mathbb{E}_{n}[\phi z]$.
We then have $A\omega^{*}=b$ and $\hat{A}\hat{\omega}=\hat{b}$.
We can combine the two equations to get 
\[
\hat{A}(\hat{\omega}-\omega^{*})=(\hat{b}-b)+(A-\hat{A})\omega^{*}.
\]
The above implies 
\begin{equation}
(\hat{\omega}-\omega^{*})^{\intercal}(-\hat{A})(\hat{\omega}-\omega^{*})=(\hat{\omega}-\omega^{*})^{\intercal}(b-\hat{b})+(\hat{\omega}-\omega^{*})^{\intercal}(\hat{A}-A)\omega^{*}.\label{eq:pf:Thm1_1}
\end{equation}
We earlier showed $\left|\hat{A}-A\right|=O_{p}(k_{\phi}/\sqrt{n})$.
Hence, $\underline{\lambda}(-\hat{A})\ge\underline{\lambda}(-A)+o_{p}(1)$,
so
\begin{equation}
(\hat{\omega}-\omega^{*})^{\intercal}(-\hat{A})(\hat{\omega}-\omega^{*})\ge c(1-\beta)\underline{\lambda}(\mathbb{E}[\phi\phi^{\intercal}])\left|\hat{\omega}-\omega^{*}\right|^{2},\label{eq:pf:Thm1_2}
\end{equation}
with probability approaching one, for any constant $c\in(0,1)$. Given
(\ref{eq:pf:Thm1_1}) and (\ref{eq:pf:Thm1_2}), 
\[
\left|\hat{\omega}-\omega^{*}\right|\le\frac{1}{c(1-\beta)\underline{\lambda}(\mathbb{E}[\phi\phi^{\intercal}])}\left(\left|\hat{b}-b\right|+\left|\hat{A}\omega^{*}-A\omega^{*}\right|\right),
\]
with probability approaching one.

It remains to bound $\left|\hat{b}-b\right|$ and $\left|\hat{A}\omega^{*}-A\omega^{*}\right|$.
As before, we can define $\hat{b}_{t}=n^{-1}\sum_{i}\phi_{it}z_{it}$
and $b_{t}=E[\phi_{it}z_{it}]$ to obtain 
\begin{align*}
E\left|\hat{b}_{t}-b_{t}\right|^{2} & =E\left|\frac{1}{n}\sum_{i}\left\{ \phi_{it}z_{it}-E\left[\phi_{it}z_{it}\right]\right\} \right|^{2}\le\frac{1}{n}E\left|\phi_{it}z_{it}\right|^{2}.
\end{align*}
This proves 
\[
E\left|\hat{b}-b\right|^{2}\le\frac{1}{T-1}\sum_{t=1}^{T-1}E\left|\hat{b}_{t}-b_{t}\right|^{2}\le\frac{1}{n}\mathbb{E}\left[\left|\phi z\right|^{2}\right]\le\frac{k_{\phi}L^{2}M^{2}}{n}=O_{p}(k_{\phi}/n).
\]
In a similar vein, 
\begin{align*}
E\left|\hat{A}\omega^{*}-A\omega^{*}\right|^{2} & =E\left|\frac{1}{n(T-1)}\sum_{t=1}^{T-1}\sum_{i}\left\{ \phi_{it}\left(\beta\phi_{it+1}-\phi_{it}\right)^{\intercal}\omega^{*}-E\left[\phi_{it}\left(\beta\phi_{it+1}-\phi_{it}\right)^{\intercal}\omega^{*}\right]\right\} \right|^{2}\\
 & =O_{p}\left(k_{\phi}/n\right),
\end{align*}
as long as $\mathbb{E}\left[\left|\phi\left(\beta\phi-\phi\right)^{\intercal}\omega^{*}\right|^{2}\right]=O(k_{\phi}).$
But the latter is true under Assumptions 1(ii)-(iv) since 
\begin{align*}
\mathbb{E}\left[\left|\phi\left(\beta\phi^{\intercal}\omega^{*}-\phi^{\intercal}\omega^{*}\right)\right|^{2}\right] & \le k_{\phi}M^{2}(2+2\beta^{2})\mathbb{E}\left[\left|\phi^{\intercal}\omega^{*}\right|^{2}\right],
\end{align*}
\[
\mathbb{E}\left[\left|\phi^{\intercal}\omega^{*}\right|^{2}\right]^{1/2}\le\left\Vert \phi^{\intercal}\omega^{*}-h\right\Vert _{2}+\left\Vert h\right\Vert _{2}\le O(k_{\phi}^{-\alpha})+(1-\beta)^{-1}L<\infty,
\]
where the second inequality uses $\left\Vert \phi^{\intercal}\omega^{*}-h\right\Vert _{2}=O(k_{\phi}^{-\alpha})$
(as shown above), and $\vert h(\cdot,\cdot)\vert_{\infty}\le(1-\beta)^{-1}\vert z(\cdot,\cdot)\vert_{\infty}<(1-\beta)^{-1}L$
(which can be easily verified using (\ref{eq:Recursion}) and Assumption
1(iv)). Combining the above, there exists $C<\infty$ such that $\left|\hat{\omega}-\omega^{*}\right|\le C\sqrt{k_{\phi}/n},$
with probability approaching one. We have thus shown (\ref{eq:estimation of omega}). 

Now observe that 
\begin{align*}
\left\Vert \phi^{\intercal}\hat{\omega}-h\right\Vert _{2}^{2} & \le2\left\Vert \phi^{\intercal}\hat{\omega}-\phi^{\intercal}\omega^{*}\right\Vert _{2}^{2}+2\left\Vert \phi^{\intercal}\omega^{*}-h\right\Vert _{2}^{2}\\
 & =2(\hat{\omega}-\omega^{*})^{\intercal}\mathbb{E}[\phi\phi^{\intercal}](\hat{\omega}-\omega^{*})+2\left\Vert \phi^{\intercal}\omega^{*}-h\right\Vert _{2}^{2}\\
 & \le2\bar{\lambda}(\mathbb{E}[\phi\phi^{\intercal}])O_{p}\left(\frac{k_{\phi}}{n}\right)+O_{p}(k_{\phi}^{-2\alpha}),
\end{align*}
where the final inequality follows from the second claim of this theorem
and (\ref{eq:estimation of omega}). The last claim then follows from
the above along with the fact that, by Assumption 1(iv),
\[
\bar{\lambda}(\mathbb{E}[\phi\phi^{\intercal}])\le\left\Vert \phi\right\Vert _{2}^{2}\le M^{2}k_{\phi},
\]

\subsection{Proof of Theorem \ref{Theorem 2}\label{subsec:Proof-of-Theorem2}}

The first two claims follow from steps analogous to those in Theorem
\ref{Theorem 1}. We thus need to show that with probability approaching
one,
\begin{equation}
\left|\hat{\xi}-\xi^{*}\right|\le C(1-\beta)^{-1}\sqrt{\frac{k_{r}}{n}},\label{eq:estimation of xi}
\end{equation}
for some $C<\infty$. The third claim is a straightforward consequence
of this. 

Recall that we use a cross-fitting procedure to estimate $\xi^{*}$.
Let $n_{1},n_{2}$ denote the sample sizes, and $\hat{\eta}_{1},\hat{\xi}_{1}$
and $\hat{\eta}_{2},\hat{\xi}_{2}$ the estimates of $\eta$ and $\xi^{*}$
in the two folds. We shall show that $\vert\hat{\xi}_{1}-\xi\vert=O_{p}(\sqrt{k_{r}/n})$
(and similarly $\vert\hat{\xi}_{2}-\xi\vert=O_{p}(\sqrt{k_{r}/n})$),
and therefore $\vert\hat{\xi}-\xi\vert=O_{p}(\sqrt{k_{r}/n})$. To
this end, let $A_{r}:=\mathbb{E}[rr^{\intercal}]$, $b_{r}:=\mathbb{E}[r(a,x)e(a^{\prime},x^{\prime})]$,
$\hat{A}_{r}^{(1)}:=\mathbb{E}_{n}^{(1)}[rr^{\intercal}]$ and $\hat{b}_{r}^{(1)}:=\mathbb{E}_{n}^{(1)}[r(a,x)e(a^{\prime},x^{\prime};\hat{\eta}_{2})]$,
where $\mathbb{E}_{n}^{(1)}[\cdot]$ denotes the empirical expectation
using only the first block. We also employ the notation $\psi(a,x,a^{\prime},x^{\prime};\eta):=r(a,x)e(a^{\prime},x^{\prime};\eta)$
and $\psi_{it}(\eta):=r(a_{it},x_{it})e(a_{it+1},x_{it+1};\eta)$.

Based on the above definitions, we have $\hat{A}_{r}^{(1)}\hat{\xi}_{1}=\hat{b}_{r}^{(1)}$,
and $A_{r}\xi^{*}=b_{r}$. Comparing with the proof of Theorem \ref{Theorem 1},
the only difference is in the treatment of $\vert\hat{b}_{r}^{(1)}-b_{r}\vert$.
As before, define $\hat{b}_{rt}^{(1)}:=n^{-1}\sum_{i}\psi_{it}(\hat{\eta}_{2})$
and $b_{rt}:=E[\psi_{it}(\eta)]$. We then have $\vert\hat{b}_{r}^{(1)}-b_{r}\vert=(T-1)^{-1}\sum_{t=1}^{T-1}\vert\hat{b}_{rt}^{(1)}-b_{rt}\vert$.
Since $T$ is finite, it suffices to bound $\vert\hat{b}_{rt}^{(1)}-b_{rt}\vert$
for some arbitrary $t$. Now, by similar arguments as in the proof
of Theorem \ref{Theorem 1}, we have
\[
\frac{1}{n_{1}}\sum_{i=1}^{n_{1}}\left\{ \psi_{it}(\eta)-E\left[\psi_{it}(\eta)\right]\right\} =O_{p}\left(\sqrt{k_{r}/n}\right).
\]
Hence (\ref{eq:estimation of xi}) follows once we show
\begin{equation}
\hat{b}_{rt}^{(1)}-b_{rt}=\frac{1}{n_{1}}\sum_{i=1}^{n_{1}}\left\{ \psi_{it}(\eta)-E\left[\psi_{it}(\eta)\right]\right\} +o_{p}\left(\sqrt{k_{r}/n}\right).\label{eq:pf:Thm2:1}
\end{equation}

We now prove (\ref{eq:pf:Thm2:1}). Denoting $\mathcal{N}_{2}$ the
set of observations in the second fold: 
\begin{align*}
 & \hat{b}_{rt}^{(1)}-b_{rt}-\frac{1}{n_{1}}\sum_{i=1}^{n_{1}}\left\{ \psi_{it}(\eta)-E\left[\psi_{it}(\eta)\right]\right\} \\
 & =\frac{1}{n_{1}}\sum_{i=1}^{n_{1}}\left\{ \left(\psi_{it}(\hat{\eta}_{2})-\psi_{it}(\eta)\right)-\left(E\left[\psi_{it}(\hat{\eta}_{2})\vert\mathcal{N}_{2}\right]-E\left[\psi_{it}(\eta)\right]\right)\right\} +\left\{ E\left[\psi_{it}(\hat{\eta}_{2})\vert\mathcal{N}_{2}\right]-E\left[\psi_{it}(\eta)\right]\right\} .
\end{align*}

Define the above as $R_{1nt}+R_{2nt}$. First consider the term $R_{1nt}$.
Define 
\[
\delta_{it}:=\left(\psi_{it}(\hat{\eta}_{2})-\psi_{it}(\eta)\right)-\left(E\left[\psi_{it}(\hat{\eta}_{2})\vert\mathcal{N}_{2}\right]-E\left[\psi_{it}(\eta)\right]\right).
\]
Clearly, $E[\delta_{it}\vert\mathcal{N}_{2}]=0$. We then have 
\begin{equation}
E\left[\left.\vert R_{1nt}\vert^{2}\right|\mathcal{\mathcal{N}}_{2}\right]=\frac{1}{n_{1}}E\left[\left.\vert\delta_{it}\vert^{2}\right|\mathcal{\mathcal{N}}_{2}\right]=\frac{1}{n_{1}}E\left[\left.\left|\psi_{it}(\hat{\eta}_{2})-\psi_{it}(\eta)\right|^{2}\right|\mathcal{\mathcal{N}}_{2}\right].\label{eq:pf:Thm2:2}
\end{equation}
Now for any $(a,x,a^{\prime},x^{\prime})$, we can note from the definition
of $\psi(\cdot)$ that with probability approaching one, 
\begin{align}
\left|\psi(a,x,a^{\prime},x^{\prime};\hat{\eta}_{2})-\psi(a,x,a^{\prime},x^{\prime};\eta)\right| & \le\vert r(a,x)\vert\left\{ \vert\ln\hat{\eta}_{2}-\ln\eta\vert+\vert\hat{\eta}_{2}-\eta\vert\right\} \nonumber \\
 & \le M\sqrt{k_{r}}\left\{ \vert\ln\hat{\eta}_{2}-\ln\eta\vert+\vert\hat{\eta}_{2}-\eta\vert\right\} \nonumber \\
 & \le M\sqrt{k_{r}}(2\delta^{-1}+1)\vert\hat{\eta}_{2}-\eta\vert,\label{eq:pf:Thm2_3}
\end{align}
where the second inequality follows from Assumption 2(iii), and the
third follows from 2(v).\footnote{In particular, we have used the fact $\hat{\eta}_{2}>\delta+o_{p}(1)$
which follows from $\eta>\delta$ and $\left|\hat{\eta}_{2}-\eta\right|=o_{p}(1)$.} In view of (\ref{eq:pf:Thm2:2}) and (\ref{eq:pf:Thm2_3}), there
exists $C<\infty$ such that
\begin{align*}
E\left[\left.\vert R_{1nt}\vert^{2}\right|\mathcal{\mathcal{N}}_{2}\right] & \le\frac{Ck_{r}}{n_{1}}E\left[\left.\left|\hat{\eta}_{2}(a_{it+1},x_{it+1})-\eta(a_{it+1},x_{it+1})\right|^{2}\right|\mathcal{\mathcal{N}}_{2}\right]\\
 & \le\frac{Ck_{r}T}{n_{1}}\left\Vert \hat{\eta}_{2}-\eta\right\Vert _{2}^{2}=o_{p}(k_{r}/n),
\end{align*}
where the last equality follows by Assumption 2(v). This proves
\begin{equation}
\vert R_{1nt}\vert=o_{p}(\sqrt{k_{r}/n}).\label{eq:pf:Thm2_4}
\end{equation}

Next consider the term $R_{2nt}$. We note that $E[\psi_{it}(\eta)]$
is twice Fr\'echet differentiable. In the main text we have shown
that $\partial_{\eta}E[\psi_{it}(\eta)]=0$ (c.f equation (\ref{eq:orthogonality for =00005Ceta})).
Furthermore, following some straightforward algebra it is possible
to show $\vert\partial_{\eta}^{2}E[\psi_{it}(\eta)]\vert\le C_{1}\sqrt{k_{r}}$,
for some $C_{1}<\infty$, as long as $\eta$ is bounded away from
$0$ (as assured by Assumption 2(v)). Hence 
\begin{align}
E\left[\left.\vert R_{2nt}\vert^{2}\right|\mathcal{\mathcal{N}}_{2}\right] & \le C_{1}\sqrt{k_{r}}E\left[\left.\left|\hat{\eta}_{2}(a_{it+1},x_{it+1})-\eta(a_{it+1},x_{it+1})\right|^{2}\right|\mathcal{\mathcal{N}}_{2}\right]\nonumber \\
 & \le C_{1}T\sqrt{k_{r}}\left\Vert \hat{\eta}_{2}-\eta\right\Vert _{2}^{2}=o_{p}(k_{r}/n).\label{eq:pf:Thm2_5}
\end{align}

(\ref{eq:pf:Thm2_4}) and (\ref{eq:pf:Thm2_5}) imply (\ref{eq:pf:Thm2:1}),
which concludes the proof of the theorem. 

\newpage{}

\section{Online Appendix\label{sec:Appendix:Supplemental-discussion}}

\subsection{Discrete states\label{subsec:Discrete-states}}

Following discretization, CCP methods (see e.g \citealp{AguirregabiriaMira2010}),
estimate $h(a,x)$ by solving the recursive equations
\begin{equation}
\breve{h}(a,x)=z(a,x)+\beta\sum_{x^{\prime}}\hat{K}(x^{\prime}\vert a,x)\sum_{a'}\hat{P}(a^{\prime}\vert x^{\prime})\breve{h}(a^{\prime},x^{\prime}),\label{eq:old-CCP}
\end{equation}
where $\hat{K},\hat{P}$ are estimates of $K,P$ obtained as cell
estimates. Now, by \citet{TsitsiklisRoy1997}, when the functional
approximation saturates all the states, the TD estimate from (\ref{eq:empirical TD fixed point})
satisfies
\[
z(a,x)+\beta\mathbb{E}_{n}[\hat{h}(a^{\prime},x^{\prime})\vert a,x]=\hat{h}(a,x),
\]
where $\mathbb{E}_{n}[\hat{h}(a^{\prime},x^{\prime})\vert a,x]$ denotes
the conditional expectation of $\hat{h}(a^{\prime},x^{\prime})$ given
$a$ and $x$ under the empirical distribution $\mathbb{P}_{n}$ (the
conditional distribution exists because of the discrete number of
states). But for discrete data, $\mathbb{E}_{n}[\hat{h}(a^{\prime},x^{\prime})\vert a,x]$
is simply 
\[
\mathbb{E}_{n}[\hat{h}(a^{\prime},x^{\prime})\vert a,x]=\sum_{x^{\prime}}\hat{K}(x^{\prime}\vert a,x)\sum_{a'}\hat{P}(a^{\prime}\vert x^{\prime})\hat{h}(a^{\prime},x^{\prime}),
\]
and the values of $\hat{h}(a,x)$ and $\breve{h}(a,x)$ coincide exactly.
Thus, the two algorithms give identical results (a similar property
also holds for $g(a,x)$). Since our estimates $\hat{h}(a,x)$ coincide
with those from the standard CCP estimators, the resulting estimate
$\hat{\theta}$ is also exactly the same. 

When the states are discrete, \citet{AguirregabiriaMira2002} show
that the estimation of $\eta$ is orthogonal to the estimation of
$\theta^{*}$. It is important to note, however, that the estimation
of the Markov kernel $K(x^{\prime}\vert a,x)$ is not orthogonal to
the estimation of $\theta^{*}$. Now, with discrete states, any estimate,
$\hat{K}(x^{\prime}\vert a,x)$, of $K(x^{\prime}\vert a,x)$ converges
at parametric rates, so $\sqrt{n}$-consistent estimation of $\theta$
is still possible. However, as we show in Section \ref{subsec:Continuous-states-and},
this creates issues with continuous states.

\subsection{Smoothness properties of dynamic-programming operators and fixed
points\label{subsec:Smoothness-properties-of}}

The following result provides a sufficient condition for Assumption
3(iii). 

\begin{lem}\label{Lem - smoothness} Suppose that $\max_{\vert p\vert\le\gamma}\sup_{a,x}\left|\partial_{x}^{p}z(a,x)\right|<C$
and\\
 $\max_{\vert p\vert\le\gamma}\sup_{a,x}\int\left|\partial_{x}^{p}K(x^{\prime}\vert a,x)\right|dx^{\prime}<C$
for some $C<\infty$. Then, there exists $M_{0}<\infty$ such that
$\left|\Gamma_{z}[f]\right|_{\infty}\le M_{0}$ for all $\left\{ f:\vert f\vert_{\infty}\le M_{0}\right\} $.
Furthermore, for each $M_{0}$, there exists $M<\infty$ such that
$\Gamma_{z}[f](\cdot,a)\in\mathcal{W}^{\gamma,\infty}$ for all $\vert f\vert_{\infty}\le M_{0}$
and $a\in\mathcal{A}$. \end{lem}
\begin{proof}
For the first claim, note that
\begin{align*}
\left|\Gamma_{z}[f](a,x)\right| & \le\left|z(a,x)\right|+(1-\beta)\sup_{a,x}\left|\int f(a^{\prime},x^{\prime})P(a^{\prime}\vert x^{\prime})K(x^{\prime}\vert a,x)da^{\prime}dx^{\prime}\right|\\
 & \le C+(1-\beta)M_{0}.
\end{align*}
Hence, the claim is satisfied for $M_{0}\ge C/\beta$. 

We now turn to the second claim. For any $\vert p\vert\le\gamma$,
we have 
\begin{align*}
\left|\partial_{x}^{p}\Gamma_{z}[f](a,x)\right| & \le\left|\partial_{x}^{p}z(a,x)\right|+(1-\beta)\sup_{a,x}\left|\int f(a^{\prime},x^{\prime})P(a^{\prime}\vert x^{\prime})\partial_{x}^{p}K(x^{\prime}\vert a,x)da^{\prime}dx^{\prime}\right|\\
 & <C+(1-\beta)M_{0}C.
\end{align*}
Hence, $\max_{\vert p\vert\le\gamma}\sup_{x\in\mathcal{X}}\left|D^{p}\Gamma_{z}[f]\right|<M:=\max\left\{ M_{0},C+(1-\beta)M_{0}C\right\} $. 
\end{proof}
We now provide sufficient conditions for $h(\cdot)$ to be $\gamma$-H\"{o}lder
continuous. The result provides a sufficient condition for Assumption
3(i). It can also be used as a justification for Assumption 1(iii);
see the discussion following Assumption 1. 

\begin{lem} \label{Lem - smoothness of fixed point}Under the assumptions
of Lemma \ref{Lem - smoothness}, there exist $M_{0},M<\infty$ such
that $\vert h\vert_{\infty}\le M_{0}$ and $h(a,\cdot)\in\mathcal{W}^{\gamma,\infty}$
for all $a\in\mathcal{A}$. \end{lem}
\begin{proof}
We start by showing $h(a,x)$ is uniformly bounded. Define $M_{0}$
to be any positive real number such that $\vert z(a,x)\vert_{\infty}<\beta M_{0}$
(such a number exists under the stated assumptions). Now, for any
$f$ such that $\vert f\vert_{\infty}\le M_{0}$, 
\[
\vert\Gamma_{z}[f](a,x)\vert_{\infty}\le\vert z(a,x)\vert_{\infty}+(1-\beta)M_{0}\le M_{0}.
\]
In other words, $\Gamma_{z}[\cdot]$ maps the space $\mathcal{S}_{0}\equiv\{f:\vert f\vert_{\infty}\le M_{0}\}$
onto itself. Hence, by the properties of contraction mappings, the
fixed point of $\Gamma_{z}[\cdot]$ must lie in $\mathcal{S}_{0}$,
i.e $\vert h(a,x)\vert_{\infty}\le M_{0}$.

We now show $h(a,\cdot)\in\mathcal{W}^{\gamma,\infty}$ for all $a$.
For any $f\in\mathcal{S}_{0}$, and $\vert p\vert\le\gamma$, 
\begin{align*}
\left|\partial_{x}^{p}\Gamma_{z}[f](a,x)\right|_{\infty} & \le\left|\partial_{x}^{p}z(a,x)\right|_{\infty}+(1-\beta)\sup_{a,x}\left|\int f(a^{\prime},x^{\prime})P(a^{\prime}\vert x^{\prime})\partial_{x}^{p}K(x^{\prime}\vert a,x)da^{\prime}dx^{\prime}\right|\\
 & \le C+(1-\beta)M_{0}C\le M,
\end{align*}
where $M:=\max\{M_{0},C+(1-\beta)M_{0}C\}$. 

Defining $\mathcal{S}_{1}\equiv\left\{ f:\max_{0<\vert p\vert\le\gamma}\vert\partial_{x}^{p}f\vert_{\infty}\le M\right\} $,
we have thus shown $\Gamma_{z}[\mathcal{S}_{0}]\subseteq\mathcal{S}_{1}$.
But the first part of this proof implies $\Gamma_{z}[\mathcal{S}_{0}]\subseteq\mathcal{S}_{0}$.
Hence, $\Gamma_{z}[\mathcal{S}_{0}]\subseteq\mathcal{S}_{0}\cap\mathcal{S}_{1}$.
Consequently, by the properties of contraction mappings, $h(a,\cdot)\in\mathcal{S}_{0}\cap\mathcal{S}_{1}$. 
\end{proof}
It is straightforward to write down analogous versions of Lemmas \ref{Lem - smoothness}
and \ref{Lem - smoothness of fixed point} for $\Gamma_{g}[\cdot]$
and $g(a,\cdot)$. We omit the details. 

\subsection{Verification of the locally robust moment function\label{subsec:Verification-of-the}}

The locally robust correction needs to satisfy two properties. First,
it must be mean zero. This is easily verified. Second, it must satisfy
the zero G\^ateaux derivative requirement, i.e., after adding the
correction term, the moment condition should satisfy 
\begin{equation}
\left.\partial_{\tau}\mathbb{E}[\zeta(\tilde{{\bf a}},\tilde{{\bf x}};\theta,h+\tau\gamma,g)]\right|_{\tau=0}=0,\ \left.\partial_{\tau}\mathbb{E}[\zeta(\tilde{{\bf a}},\tilde{{\bf x}};\theta,h,g+\tau\gamma)]\right|_{\tau=0}=0\label{eq:Gateaux condition requirement}
\end{equation}
for all $\gamma\in\mathcal{H}$, where $\mathcal{H}$ is the set of
all square integrable functions over the domain $\mathcal{A}\times\mathcal{X}$.
We now verify this below:

Let $\psi_{h}(\cdot)$ denote the solution of the functional equation
\begin{equation}
\left.\partial_{\tau}\mathbb{E}\left[m(a,x;\theta,h+\tau\gamma,g)\right]\right|_{\tau=0}=\mathbb{E}\left[\psi_{h}(a,x;\theta,h,g)\cdot\gamma(a,x)\right]\ \forall\ \gamma\in\mathcal{H},\label{eq:definition of derivative wrt h-1}
\end{equation}
and similarly, $\psi_{g}(\cdot)$, the solution of the functional
equation 
\[
\left.\partial_{\tau}\mathbb{E}\left[m(a,x;\theta,h,g+\tau\gamma)\right]\right|_{\tau=0}=\mathbb{E}\left[\psi_{g}(a,x;\theta,h,g)\cdot\gamma(a,x)\right]\ \forall\ \gamma\in\mathcal{H}.
\]
Following some tedious but straightforward algebra, we can show $\psi_{g}(\cdot)=\psi(\cdot)$
and $\psi_{h}(\cdot)=\theta\psi(\cdot)$, where $\psi(\cdot)$ is
defined in (\ref{eq:derivate wrt h}) in the main text. 

Now, the doubly robust moment condition in (\ref{eq:doubly robust moment - non parametric})
can be expanded as 
\begin{align}
\zeta(\tilde{{\bf a}},\tilde{{\bf x}};\theta,h,g) & =m(a,x;\theta,h,g)-\theta\lambda(a,x;\theta)\left\{ z(a,x)+\beta h(a^{\prime},x^{\prime})-h(a,x)\right\} \nonumber \\
 & \quad-\lambda(a,x;\theta)\left\{ \beta e(a^{\prime},x^{\prime};\eta)+\beta g(a^{\prime},x^{\prime})-g(a,x)\right\} \nonumber \\
 & =m(a,x;\theta,h,g)-\lambda_{h}(a,x;\theta)\left\{ z(a,x)+\beta h(a^{\prime},x^{\prime})-h(a,x)\right\} \nonumber \\
 & \quad-\lambda_{g}(a,x;\theta)\left\{ \beta e(a^{\prime},x^{\prime};\eta)+\beta g(a^{\prime},x^{\prime})-g(a,x)\right\} ,\label{eq:doubly robust moment - non parametric- alternative form}
\end{align}
where we have defined $\lambda_{h}(a,x;\theta)=\theta\lambda(a,x;\theta)$
and $\lambda_{g}(a,x,;\theta)=\lambda(a,x;\theta)$ for the second
equality. Using the definitions of $\psi(\cdot),\psi_{h}(\cdot),\psi_{g}(\cdot)$
and $\lambda(\cdot)$, we observe that $\lambda_{h}(a,x;\theta)$
is the fixed point of the `backward' dynamic programming operator
\begin{equation}
\Gamma_{h,\theta}^{\dagger}[f](a,x):=-\psi_{h}(a,x;\theta,h,g)+\beta\mathbb{E}\left[f(a^{-\prime},x^{-\prime})\vert a,x\right],\label{eq:adjustment term 1-1-1}
\end{equation}
and $\lambda_{g}(a,x;\theta)$ is the fixed point of 
\begin{equation}
\Gamma_{g,\theta}^{\dagger}[f](a,x):=-\psi_{g}(a,x;\theta,h,g)+\beta\mathbb{E}\left[f(a^{-\prime},x^{-\prime})\vert a,x\right].\label{eq:adjustment term 2-1}
\end{equation}

We now verify (\ref{eq:Gateaux condition requirement}). To this end,
observe that for any square integrable $\gamma$,
\begin{align}
 & \left.\partial_{\tau}\mathbb{E}[\zeta(\tilde{{\bf a}},\tilde{{\bf x}};\theta,h+\tau\gamma,g)]\right|_{\tau=0}\nonumber \\
 & =\left.\partial_{\tau}\mathbb{E}[m(a,x;\theta,h+\tau\gamma,g)]\right|_{\tau=0}-\mathbb{E}[\beta\lambda_{h}(a,x;\theta)\gamma(a^{\prime},x^{\prime})]+\mathbb{E}[\lambda_{h}(a,x;\theta)\gamma(a,x)]\nonumber \\
 & =\mathbb{E}[\psi_{h}(a,x;\theta,h,g)\gamma(a,x)]-\mathbb{E}[\beta\lambda_{h}(a,x;\theta)\gamma(a^{\prime},x^{\prime})]+\mathbb{E}[\lambda_{h}(a,x;\theta)\gamma(a,x)],\label{eq:Gauteax derivative}
\end{align}
where the second equality follows from the definition of $\psi_{h}(\cdot)$
in (\ref{eq:definition of derivative wrt h-1}). Since $\lambda_{h}(\cdot)$
is the fixed point of $\Gamma_{h,\theta}^{\dagger}[\cdot]$, we can
expand the third term in (\ref{eq:Gauteax derivative}) as 
\begin{align*}
\mathbb{E}[\lambda_{h}(a,x;\theta)\gamma(a,x)] & =\mathbb{E}\left[\left\{ -\psi_{h}(a,x;\theta,h,g)+\beta\lambda_{h}(a^{-\prime},x^{-\prime};\theta)\right\} \gamma(a,x)\right]\\
 & =-\mathbb{E}\left[\psi_{h}(a,x;\theta,h,g)\gamma(a,x)\right]+\mathbb{E}[\beta\lambda_{h}(a,x;\theta)\gamma(a^{\prime},x^{\prime})],
\end{align*}
where the second equality uses the fact that $\mathbb{E}[\cdot]$
corresponds to a stationary distribution. We thus conclude $\left.\partial_{\tau}\mathbb{E}[\zeta(\tilde{{\bf a}},\tilde{{\bf x}};\theta,h+\tau\gamma,g)]\right|_{\tau=0}=0$
for all $\gamma$, or $\partial_{h}\mathbb{E}[\zeta(\tilde{{\bf a}},\tilde{{\bf x}};\theta,h,g)]=0$,
as required. In fact, by similar arguments, we can also show the stronger
statement that $\partial_{h}\mathbb{E}[\zeta(\tilde{{\bf a}},\tilde{{\bf x}};\theta,h,g)]=0$
and $\partial_{g}\mathbb{E}[\zeta(\tilde{{\bf a}},\tilde{{\bf x}};\theta,h,g)]=0$
in a Fr\'echet sense. The argument for showing $\left.\partial_{\tau}\mathbb{E}[\zeta(\tilde{{\bf a}},\tilde{{\bf x}};\theta,h,g+\tau\gamma)]\right|_{\tau=0}=0$
is similar. 

\subsection{The locally robust estimator for linear functional classes\label{subsec:Construction-of-the}}

Suppose $h(x,a)$ and $g(x,a)$ were truly finite-dimensional, i.e
$h(x,a)\equiv\phi(x,a)^{\intercal}\omega^{*}$ and $g(x,a)\equiv r(x,a)^{\intercal}\xi^{*}$.
Denote $(\tilde{{\bf a}},\tilde{{\bf x}}):=(a,x,a^{\prime},x^{\prime})$,
$\mathbf{v}:=(\omega,\xi)$, $\mathbf{v}^{*}:=(\omega^{*},\xi^{*})$
and 
\[
\pi(a,x;\theta,\mathbf{v}):=\frac{\exp\left\{ \left(\phi(a,x){}^{\intercal}\omega\right)\theta+r(a,x){}^{\intercal}\xi\right\} }{\sum_{\breve{a}}\exp\left\{ \left(\phi(\breve{a},x)^{\intercal}\omega\right)\theta+r(\breve{a},x)^{\intercal}\xi)\right\} }.
\]
The true value $\theta^{*}$ solves
\begin{equation}
\mathbb{E}[m(a,x;\theta^{*},{\bf v}^{*})]=0;\quad m(a,x;\theta,\mathbf{v}):=\partial_{\theta}\ln\pi(a,x;\theta,\mathbf{v}).\label{eq:moment criterion}
\end{equation}
Now, (\ref{eq:inversion for TD fixed point}) and (\ref{eq:xi-estimate-method1})
imply $\omega^{*}$ and $\xi^{*}$ are identified by the auxiliary
moments 
\begin{align}
\mathbb{E}[\varphi_{h}(\tilde{{\bf a}},\tilde{{\bf x}},\omega^{*})] & =0,\ \textrm{and}\ \mathbb{E}[\varphi_{g}(\tilde{{\bf a}},\tilde{{\bf x}},\xi^{*})]=0,\ \textrm{where}\label{eq:auxiliary moments}\\
\varphi_{h}(\tilde{{\bf a}},\tilde{{\bf x}},\omega) & :=\phi(a,x)\left\{ z(a,x)+\beta\phi(a^{\prime},x^{\prime})^{\intercal}\omega-\phi(a,x)^{\intercal}\omega\right\} ,\ \textrm{and}\nonumber \\
\varphi_{g}(\tilde{{\bf a}},\tilde{{\bf x}},\xi) & :=r(a,x)\left\{ \beta e(a^{\prime},x^{\prime};\hat{\eta})+\beta r(a^{\prime},x^{\prime})^{\intercal}\xi-r(a,x)^{\intercal}\xi\right\} .\nonumber 
\end{align}

We make use of (\ref{eq:moment criterion}) and (\ref{eq:auxiliary moments})
to construct a locally robust moment for $\theta^{*}$. Following
\citet{Newey1994} and \citet{CEINR2018}, this is given by
\begin{equation}
\mathbb{E}[\zeta(\tilde{{\bf a}},\tilde{{\bf x}};\theta^{*},{\bf v}^{*})]=0,\label{eq:locally robust moment}
\end{equation}
where
\begin{align*}
\zeta(\tilde{{\bf a}},\tilde{{\bf x}};\theta,{\bf v}) & :=m(a,x;\theta,{\bf v})-\mathbb{E}[\partial_{\omega}m(a,x;\theta,{\bf v})]\mathbb{E}[\partial_{\omega}\varphi_{h}(\tilde{{\bf a}},\tilde{{\bf x}},\omega)]^{-1}\varphi_{h}(\tilde{{\bf a}},\tilde{{\bf x}},\omega)\\
 & \hfill-\mathbb{E}[\partial_{\xi}m(a,x;\theta,{\bf v})]\mathbb{E}[\partial_{\xi}\varphi_{g}(\tilde{{\bf a}},\tilde{{\bf x}},\xi)]^{-1}\varphi_{g}(\tilde{{\bf a}},\tilde{{\bf x}},\xi).
\end{align*}

We can now construct a locally robust estimator for $\theta^{*}$
based on (\ref{eq:locally robust moment}). Following \citet{CEINR2018},
we employ a cross-fitting procedure by randomly splitting the data
into two samples $\mathcal{N}_{1}$ and $\mathcal{N}_{2}$. We compute
$\hat{\omega}$ and $\hat{\xi}$ using one of the samples, say $\mathcal{N}_{2}$.
We also compute, $\tilde{\theta}$, a preliminary consistent estimator
of $\theta^{*}$ by applying the plug-in estimator (\ref{eq:estimate of theta})
on observations in $\mathcal{N}_{2}$. Denote by $\mathbb{E}_{n}^{(1)}[\cdot]$
the empirical expectation using only the observations in the first
sample, $\mathcal{N}_{1}$. We then obtain $\hat{\theta}$ as the
solution to the moment equation 
\begin{equation}
\mathbb{E}_{n}^{(1)}\left[\zeta_{n}(\tilde{{\bf a}},\tilde{{\bf x}};\theta,\hat{\omega},\hat{\xi})\right]=0,\label{eq:locally robust sample moment}
\end{equation}
where 
\begin{align*}
\zeta_{n}(\tilde{{\bf a}},\tilde{{\bf x}};\theta,{\bf v}) & :=m(a,x;\theta,{\bf v})-\mathbb{E}_{n}^{(1)}[\partial_{\omega}m(a,x;\tilde{\theta},{\bf v})]\mathbb{E}_{n}^{(1)}[\partial_{\omega}\varphi_{h}(\tilde{{\bf a}},\tilde{{\bf x}},\omega)]^{-1}\varphi_{h}(\tilde{{\bf a}},\tilde{{\bf x}},\omega)\\
 & \hfill-\mathbb{E}_{n}^{(1)}[\partial_{\xi}m(a,x;\tilde{\theta},{\bf v})]\mathbb{E}_{n}^{(1)}[\partial_{\xi}\varphi_{g}(\tilde{{\bf a}},\tilde{{\bf x}},\xi)]^{-1}\varphi_{g}(\tilde{{\bf a}},\tilde{{\bf x}},\xi).
\end{align*}
The use of cross-fitting is critical. If we had used the entire sample
to estimate all of $\theta^{*},\omega^{*}$ and $\xi^{*}$, we would
have $\mathbb{E}_{n}[\varphi_{h}(\tilde{{\bf a}},\tilde{{\bf x}},\hat{\omega})]=0$
and $\mathbb{E}_{n}[\varphi_{g}(\tilde{{\bf a}},\tilde{{\bf x}},\hat{\xi})]=0$,
which implies $\mathbb{E}_{n}\left[\zeta_{n}(\tilde{{\bf a}},\tilde{{\bf x}},\theta,\hat{\omega},\hat{\xi})\right]=\mathbb{E}_{n}\left[m(a,x,\theta,\hat{\omega},\hat{\xi})\right]$.
As noted by \citet{CEINR2018}, cross-fitting gets rid of the `own
observation bias' that is the source of the degeneracy here. The solution,
$\hat{\theta}$, of (\ref{eq:locally robust sample moment}) is the
locally robust pseduo-MLE estimator of $\theta^{*}$. 

Even though the estimator in (\ref{eq:locally robust sample moment})
is predicated on $h(x,a)$ and $g(x,a)$ being truly finite dimensional,
the work of \citet{CEINR2018} suggests that $\zeta_{n}(\tilde{{\bf a}},\tilde{{\bf x}};\theta,{\bf v})$
should have the same form as $\zeta_{n}(\tilde{{\bf a}},\tilde{{\bf x}};\theta,h,g)$
from (\ref{eq:feasible locally robust moment}). Suppose that the
same vector of basis functions is used for both $h$ and $g$, i.e.,
$\phi(a,x)\equiv r(a,x)$. Then, substituting the expressions for
$\varphi_{h}(\tilde{{\bf a}},\tilde{{\bf x}},\omega),\varphi_{g}(\tilde{{\bf a}},\tilde{{\bf x}},\xi)$,
we obtain (after some algebra)
\begin{align*}
\zeta_{n}(\tilde{{\bf a}},\tilde{{\bf x}};\theta,{\bf v}) & =m(a,x;\theta,{\bf v})-\hat{\lambda}(a,x;\theta)\left\{ z(a,x)\theta+\beta e(a^{\prime},x^{\prime};\hat{\eta})\right.\\
 & \hfill+\left.\beta V(a^{\prime},x^{\prime};\tilde{\theta},\hat{\omega},\hat{\xi})-V(a,x;\tilde{\theta},\hat{\omega},\hat{\xi})\right\} ,
\end{align*}
where \footnotesize
\begin{align*}
\hat{V}(a,x;\theta,\hat{\omega},\hat{\xi}) & :=\theta\phi(a,x)^{\intercal}\hat{\omega}+r(a,x)^{\intercal}\hat{\xi},\\
\hat{\lambda}(a,x;\tilde{\theta}) & :=-\phi(a,x)^{\intercal}\mathbb{E}_{n}^{(1)}\left[\left(\beta\phi(a^{\prime},x^{\prime})-\phi(a,x)\right)\phi(a,x)^{\intercal}\right]^{-1}\mathbb{E}_{n}^{(1)}\left[\phi(a,x)\psi(a,x;\tilde{\theta},\hat{h},\hat{g})\right].
\end{align*}
\normalsize Comparing with (\ref{eq:doubly robust moment - non parametric}),
we observe that $\hat{\lambda}(a,x;\theta)$ is simply the linear
semi-gradient estimator of $\lambda(a,x;\theta)$. The benefit of
(\ref{eq:locally robust sample moment}) over (\ref{eq:feasible locally robust moment})
is that we no longer need to estimate $\lambda(\cdot)$ separately.

\subsection{Locally robust estimators for dynamic games\label{sec:Locally-robust-estimators-dynamic-games}}

The locally robust estimator for dynamic games is similar to that
for single-agent models. To describe this, we recast the PMLE criterion
function in the form $Q(a,x;\theta,\{h_{i}\},\{g_{i}\})\linebreak=\sum_{i}\ln Q_{i}(a_{i},x;\theta,h_{i},g_{i}),$
where 
\begin{align*}
Q_{i}(a_{i},x;\theta,h_{i},g_{i}) & :=\ln\pi_{i}(a_{i},x;\theta,h_{i},g_{i});\\
\pi_{i}(a_{i},x;\theta,h_{i},g_{i}) & :=\frac{\exp\left\{ h_{i}(a_{i},x)\theta+g(a_{i},x)\right\} }{\sum_{a}\exp\left\{ h_{i}(a,x)\theta+g(a,x)\right\} }.
\end{align*}
Denote $m_{i}(a_{i},x;\theta,h_{i},g_{i}):=\partial_{\theta}Q_{i}(a_{i},x;\theta,h_{i},g_{i})$,
and 
\[
V_{i}(a_{i},x;\theta,h_{i},g_{i})=h_{i}(a_{i},x)\theta+g_{i}(a_{i},x).
\]
The locally robust moment for $\theta$ is 
\[
\zeta(\tilde{{\bf a}},\tilde{{\bf x}};\theta,\{h_{i}\},\{g_{i}\}):=\sum_{i}\zeta^{(i)}(\tilde{{\bf a}}_{i},\tilde{{\bf x}};\theta,h_{i},g_{i}),
\]
where 
\begin{align*}
\zeta^{(i)}(\tilde{{\bf a}}_{i},\tilde{{\bf x}};\theta,h_{i},g_{i}) & :=m_{i}(a_{i},x;\theta,h_{i},g_{i})-\lambda_{i}(a_{i},x;\theta)\left\{ z_{i}(a_{i},a_{-i},x)\theta+\beta e(a_{i}^{\prime},x^{\prime};\eta_{i})\right.\\
 & \hfill+\left.\beta V_{i}(a_{i}^{\prime},x^{\prime};\theta,h_{i},g_{i})-V_{i}(a_{i},x;\theta,h_{i},g_{i})\right\} .
\end{align*}
Here, $\lambda_{i}(a_{i},x;\theta)$ is the fixed point to 
\begin{align*}
\Gamma_{i}^{\dagger}[f](a_{i},x) & :=-\psi_{i}(a_{i},x;\theta,h_{i},g_{i})+\beta\mathbb{E}\left[f(a_{i}^{-\prime},x^{-\prime})\vert a_{i},x\right],
\end{align*}
where 
\[
\psi_{i}(a_{i},x;\theta,h_{i},g_{i})=-\sum_{\tilde{a}\in\mathcal{A}_{i}}\left\{ \mathbb{I}(a_{i}=\tilde{a})-\pi_{i}(\tilde{a},x;\theta,h_{i},g_{i})\right\} h_{i}(\tilde{a},x).
\]

For computation, we employ cross-fitting as in the single-agent setting
and randomly split the markets into two samples $\mathcal{N}_{1}$
and $\mathcal{N}_{2}$. We compute $\hat{h}_{i},\hat{g}_{i},\hat{\lambda}_{i}$
for all players using one of the samples, say $\mathcal{N}_{2}$.
Let $\tilde{\theta}$ denote the plug-in estimator of $\theta^{*}$
obtained using $\mathcal{N}_{2}$. Also, denote by $\mathbb{E}_{n}^{(1)}[\cdot]$
the empirical expectation as in (\ref{eq:game empirical expectation}),
but constructed only from observations in $\mathcal{N}_{1}$. The
feasible locally robust moment for $\theta$ is
\begin{equation}
\mathbb{E}_{n}^{(1)}\left[\sum_{i}\zeta_{n}^{(i)}\left(\tilde{{\bf a}}_{i},\tilde{{\bf x}};\theta,\hat{h}_{i},\hat{g}_{i}\right)\right]=0,\label{eq:locally robust sample moment-games}
\end{equation}
where 
\begin{align*}
\zeta^{(i)}(\tilde{{\bf a}}_{i},\tilde{{\bf x}};\theta,\hat{h}_{i},\hat{g}_{i}) & :=m_{i}(a_{i},x;\theta,\hat{h}_{i},\hat{g}_{i})-\hat{\lambda}_{i}(a_{i},x;\tilde{\theta})\left\{ z_{i}(a_{i},a_{-i},x)\tilde{\theta}+\beta e(a_{i}^{\prime},x^{\prime};\hat{\eta}_{i})\right.\\
 & \hfill+\left.\beta V_{i}(a_{i}^{\prime},x^{\prime};\tilde{\theta},\hat{h}_{i},\hat{g}_{i})-V_{i}(a_{i},x;\tilde{\theta},\hat{h}_{i},\hat{g}_{i})\right\} .
\end{align*}
Note that $\hat{\theta}$ can be computed player-by-player, via $\mathbb{E}_{n}^{(1)}\left[\zeta_{n}^{(i)}(\tilde{{\bf a}}_{i},\tilde{{\bf x}};\theta,\hat{h}_{i},\hat{g}_{i})\right]=0$,
if there were no common parameters $\theta$ across players, i.e if
we could partition $\theta\equiv(\theta_{1},\dots,\theta_{N})$. Alternatively,
if the players were symmetric, i.e, $z_{i}(\cdot)$ did not depend
on $i$, we can pool all the players and compute a common $\hat{h},\hat{g},\hat{\lambda}$.

The locally robust estimator (\ref{eq:locally robust sample moment-games})
has the same form as (\ref{eq:locally robust sample moment}), except
for there being separate correction terms for the estimates $\hat{h}_{i},\hat{g}_{i}$
of each player $i$. Its theoretical properties are thus equivalent
to, and can be derived in the same manner, as those for single-agent
models. 

\subsection{Incorporating permanent unobserved heterogeneity\label{sec:Incorporating-permanent-unobserv}}

We incorporate permanent unobserved heterogeneity into our models
by pairing the techniques from Section \ref{sec:Temporal-difference-learning-alg}
with the sequential Expectation-Maximization (EM) algorithm (\citealp{arcidiacono2003finite}).
The use of the sequential EM algorithm in CCP estimation under unobserved
heterogeneity was first advocated by \citet{ArcidiaconoMiller2011},
and we employ a similar approach. 

Suppose that in addition to the observed state $x$, and the choice-specific
shock $e$, individuals also base their choice decisions on a random
state variable $s$ which is known to the individual, but unobserved
to the econometrician. As is common in the literature, we assume a
finite set of unobserved states indexed by $\{1,2,\dots,k,...K\}$.
The number of states is also assumed to be known a priori. Let $\pi_{k}$
denote the population probability $P(s=k)$. The value of $s$ for
an individual is assumed to be permanent and not change with time.
However, we do not place any restrictions on the transition density
$K(x^{\prime}\vert a,x,s)$, which is allowed to change with $s$. 

To simplify the exposition, we only describe the basic version of
the algorithm without local robustness corrections as in Section \ref{subsec:Continuous-states-and}.
It is straightforward to incorporate the correction term into the
algorithm, but it comes at the expense of higher computational times. 

Suppose that the per-period utility is given by $z(a,x,s)\theta$.
For each $k$, define $h_{k}(a,x)$ and $g_{k}(a,x)$ as the solutions
to
\begin{align*}
h_{k}(a,x) & =z(a,x,k)+\beta\mathbb{E}\left[h_{k}(a^{\prime},x^{\prime})\vert a,x,s=k\right],\\
g_{k}(a,x) & =\beta\mathbb{E}\left[e(a^{\prime},x^{\prime})+g_{k}(a^{\prime},x^{\prime})\vert a,x,s=k\right].
\end{align*}
To simplify notation, let $h_{itk}:=h_{k}(a_{it},x_{it})$ and $g_{itk}:=g_{k}(a_{it},x_{it})$.
If these quantities were known, one can estimate $(\theta,\pi)$ by
maximizing the integrated pseudo-likelihood
\begin{equation}
Q(\theta,\pi)=\sum_{i=1}^{N}\log\left[\sum_{k=1}^{K}\pi_{k}\prod_{t=1}^{T-1}\frac{\exp\left\{ h_{itk}\theta+g_{itk}\right\} }{\sum_{a}\exp\left\{ h_{k}(a,x_{it})\theta+g_{k}(a,x_{it})\right\} }\right].\label{eq:EM-integrated lik}
\end{equation}
As before, we have chosen to make $h(\cdot)$ uni-dimensional to simplify
the notation. We emphasize that the methods suggested here could be
employed even if there was no heterogeneity in individual utilities,
but the transition density were heterogenous across individuals. This
is because even just the latter would result in heterogenous values
for $h(a,x)$, as it is a function of transition densities. 

To make (\ref{eq:EM-integrated lik}) usable, $h_{k}(a,x),g_{k}(a,x)$
would have to be estimated. Similar to the motivation of TD methods
in Section \ref{sec:Setup}, the heuristic we employ is to approximate
$h_{k}$ (at each $k$) by minimizing the mean-squared TD error
\begin{align*}
\textrm{TDE}(\omega) & =\bar{\mathbb{E}}\left[\mathbb{I}(s=k)\left\Vert z(a,x,s)+\beta h_{k}(a^{\prime},x^{\prime};\omega)-h_{k}(a,x;\omega)\right\Vert ^{2}\right]\\
 & =\mathbb{E}\left[P(s=k\vert{\bf a},{\bf x})\left\Vert z(a,x,k)+\beta h_{k}(a^{\prime},x^{\prime};\omega)-h_{k}(a,x;\omega)\right\Vert ^{2}\right],
\end{align*}
where $\bar{\mathbb{E}}[\cdot]$ differs from $\mathbb{E}[\cdot]$
in also taking the expectation over the distribution of the unobserved
state $s$, and
\[
P(s=k\vert{\bf a},{\bf x}):=Pr(s=k\vert a_{1},x_{1},\dots,a_{T},x_{T}).
\]
If $P(s=k\vert{\bf a},{\bf x})$ were known, we could use them as
weights in the semi-gradient and AVI approaches. In particular, for
linear semi-gradient methods, the estimates for $h_{k},g_{k}$ would
be $\hat{h}_{k}(a,x)=\phi(a,x)^{\intercal}\hat{\omega}_{k}$ and $\hat{g}_{k}(a,x)=r(a,x)^{\intercal}\hat{\xi}_{k}$,
where 
\begin{align}
\hat{\omega}_{k} & =\left[\sum_{i=1}^{n}\sum_{t=1}^{T-1}p_{ik}\phi_{it}\left(\phi_{it}-\beta\phi_{it+1}\right)^{\intercal}\right]^{-1}\sum_{i=1}^{n}\sum_{t=1}^{T-1}p_{ik}\phi_{it}z_{itk},\label{eq: infeasible Value function estimate-1}\\
\hat{\xi}_{k} & =\left[\sum_{i=1}^{n}\sum_{t=1}^{T-1}p_{ik}r_{it}\left(r_{it}-\beta r_{it+1}\right)^{\intercal}\right]^{-1}\sum_{i=1}^{n}\sum_{t=1}^{T-1}\beta p_{ik}r_{it}\dot{e}_{it+1k},\nonumber 
\end{align}
and we have used the notation $p_{ik}=P(s=k\vert{\bf a}_{i},{\bf x}_{i})$,
$z_{itk}:=z(a_{it},x_{it},k)$, $\phi_{it}:=\phi(a_{it},x_{it})$,
$r_{it}:=r(a_{it},x_{it})$, and $\dot{e}_{it+1k}$ is the current
estimate (described below) of $e_{it+1k}:=\gamma-\ln P(a_{it}\vert x_{it},s_{i}=k)$.
Similarly, for AVI, we could obtain iterative approximations $\{\hat{h}_{k}^{(j)},j=1,\dots,J\}$
for $h_{k}(\cdot)$ using 
\begin{equation}
\hat{h}_{k}^{(j+1)}=\argmin_{f\in\mathcal{F}}\sum_{i=1}^{n}\sum_{t=1}^{T-1}p_{ik}\left\Vert z_{ik}+\beta\hat{h}_{k}^{(j)}(a_{it+1},x_{it+1})-f(a_{it},x_{it})\right\Vert ^{2}.\label{eq: infeasible Value function estimate-2}
\end{equation}
The above involves solving a sequence of weighted non-parametric estimation
problems, which most machine learning methods can easily handle. The
estimates $\hat{g}_{k}^{(j)}(\cdot)$ of $g_{k}(\cdot)$ can be obtained
in a similar fashion.

Estimation of $\theta^{*},h_{k},g_{k}$ using (\ref{eq:EM-integrated lik})
and (\ref{eq: infeasible Value function estimate-1}) or (\ref{eq: infeasible Value function estimate-2})
requires knowledge of the unknown quantities $\pi_{k}$ and $p_{ik}$
along with $\dot{e}_{it+1k}$. Furthermore, even if $\pi_{k}$ were
known, maximizing the integrated likelihood function (\ref{eq:EM-integrated lik})
is computationally very expensive. The sequential EM algorithm of
\citet{arcidiacono2003finite} solves both issues.\footnote{Maximizing (\ref{eq:EM-integrated lik}) is not equivalent to Full
Information Maximum Likelihood (FIML). As in \citet{arcidiacono2003finite},
the identification and asymptotic properties of $\theta,\pi$ are
in fact determined by constructing moment conditions that correspond
to the first order conditions from maximizing $Q(\theta,\pi)$, augmented
with additional moments identifying $h_{k},g_{k}$. Together, these
moment conditions, which motivate the sequential EM algorithm, can
in turn be related to the identification properties of FIML.} To describe the procedure, let 
\[
l_{it}(\theta,h_{k},g_{k}):=\frac{\exp\left\{ h_{k}(a_{it},x_{it})\theta+g_{k}(a_{it},x_{it})\right\} }{\sum_{a}\exp\left\{ h_{k}(a,x_{it})\theta+g_{k}(a,x_{it})\right\} }.
\]
Denote by $\hat{\pi}_{k}$ and $\hat{p}_{ik}$ the estimates for $\pi_{k}$
and $p_{ik}$. The algorithm consists of two steps: the M-step and
the E-step. We first describe the M-step. Here, we update the estimates
for $h_{k},g_{k}$ and $\theta^{*}$ based on the current estimates
for $\pi_{k},p_{ik}$ and $e_{it+1k}$. To this end, note that we
can estimate $h_{k},g_{k}$ using either (\ref{eq: infeasible Value function estimate-1})
or (\ref{eq: infeasible Value function estimate-2}). From these we
can in-turn update $\hat{\theta}$ as 
\begin{equation}
\hat{\theta}=\arg\max_{\theta}\left[\sum_{i=1}^{n}\sum_{t=1}^{T-1}\sum_{k}p_{ik}\ln l_{it}\left(\theta,\hat{h}_{k},\hat{g}_{k}\right)\right].\label{eq:M-step-theta}
\end{equation}
Next, given $(\hat{\theta},\hat{h}_{k},\hat{g}_{k})$, we update $\hat{\pi}_{k},\hat{p}_{ik}$
and $\dot{e}_{it+1k}$ for all $i,k$. This is the E-step of the EM
algorithm. This step consists of three parts. In the first part, we
use the current $\hat{\theta},\hat{h}_{k},\hat{g}_{k}$ and $\hat{\pi}_{k}$
to update $\hat{p}_{ik}$ for each $i,k$ using Bayes' rule:
\begin{equation}
\hat{p}_{ik}\longleftarrow\frac{\hat{\pi}_{k}\prod_{t=1}^{T-1}l_{it}(\hat{\theta},\hat{h}_{k},\hat{g}_{k})}{\sum_{\tilde{k}}\hat{\pi}_{\tilde{k}}\prod_{t=1}^{T-1}l_{it}(\hat{\theta},\hat{h}_{\tilde{k}},\hat{g}_{\tilde{k}})}.\label{eq:Expectation step}
\end{equation}
In the second part, we update $\hat{\pi}_{k}$, for each $k$, as
\begin{equation}
\hat{\pi}_{k}\longleftarrow\frac{1}{N}\sum_{i=1}^{N}\hat{p}_{ik}.\label{eq:Expectation step - 2}
\end{equation}
Finally, we also update $\dot{e}_{it+1k}$ for all $i,t,k$ as 
\begin{equation}
\dot{e}_{it+1k}\longleftarrow\gamma-\ln l_{it+1k}(\hat{\theta},\hat{h}_{k},\hat{g}_{k}).\label{eq:error update - EM}
\end{equation}
The E and M steps are iterated until convergence. 

It is also possible to extend our methods to allow for Markovian unobserved
heterogeneity, by employing a variant of the classical Baum-Welch
algorithm. The computational and statistical details of such a procedure
are however more involved and will be described elsewhere. 

\subsection{Additional simulation results\label{sec:Appendix:Additional-simulations}}

\subsubsection{Firm entry problem - finite sample distributions}

Figures \ref{fig:distributions _ LSG_non_locally_robust} and \ref{fig:distributions _ LSG_locally_robust}
plot the finite sample distribution of the estimates for linear semi-gradients
under non-locally robust and locally robust methods. The distributions
are very close to normal even without a locally robust correction.

\begin{figure}[t]
\begin{centering}
\includegraphics[width=12cm]{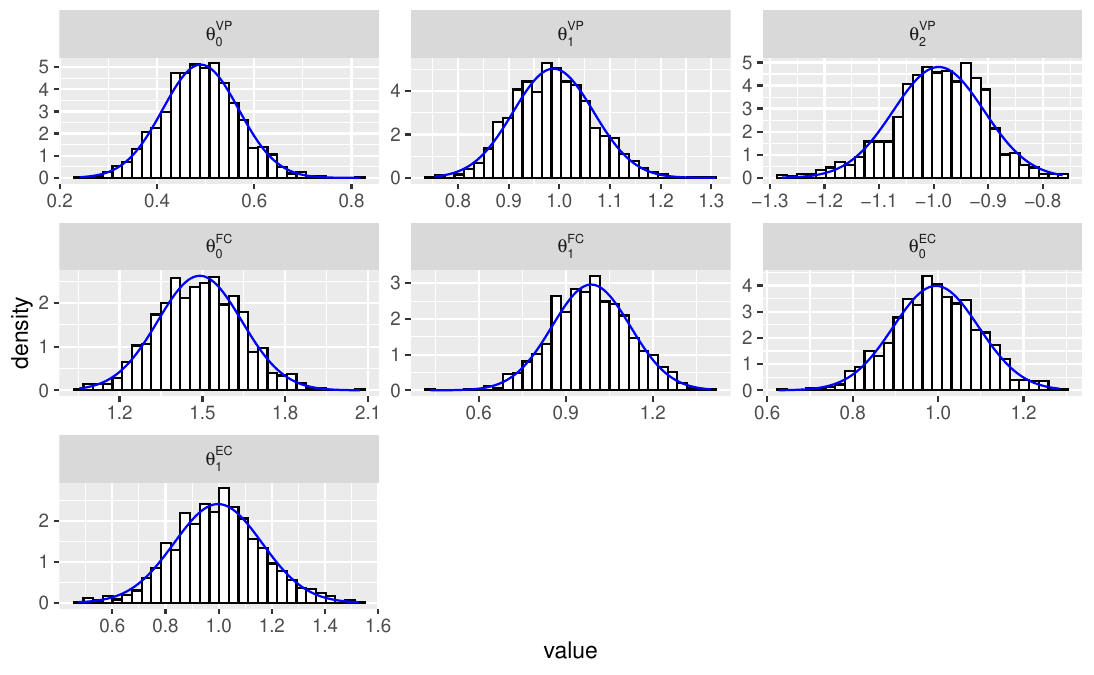}
\par\end{centering}
\begin{raggedright}
{\scriptsize{}Note: Histograms denote the finite sample distribution,
and blue line is normal density}{\scriptsize\par}
\par\end{raggedright}
\caption{Finite sample distributions under linear semi-gradients without locally
robust corrections\label{fig:distributions _ LSG_non_locally_robust}.}
\end{figure}

\begin{figure}[t]
\begin{centering}
\includegraphics[width=12cm]{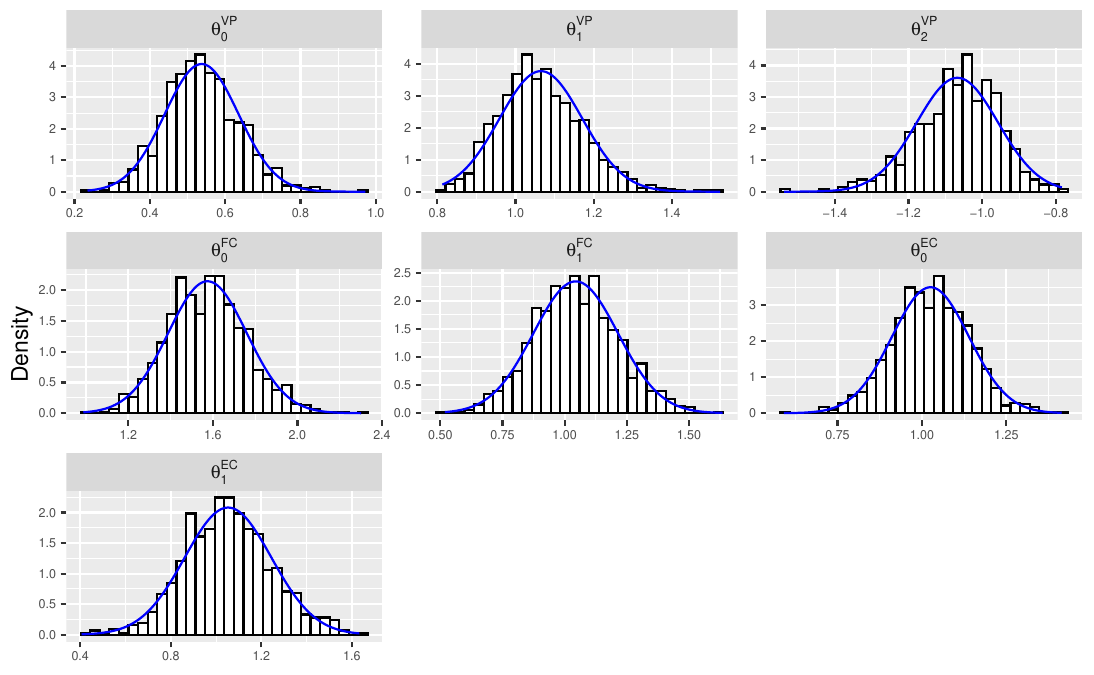}
\par\end{centering}
\begin{raggedright}
{\scriptsize{}Note: Histograms denote the finite sample distribution,
and blue line is normal density}{\scriptsize\par}
\par\end{raggedright}
\caption{Finite sample distributions under linear semi-gradients with locally
robust corrections\label{fig:distributions _ LSG_locally_robust}.}
\end{figure}

Figures \ref{fig:distributions _AVI_non_locally_robust} and \ref{fig:distributions _ AVI_locally_robust}
present similar results for AVI.

\begin{figure}[t]
\begin{centering}
\includegraphics[width=12cm]{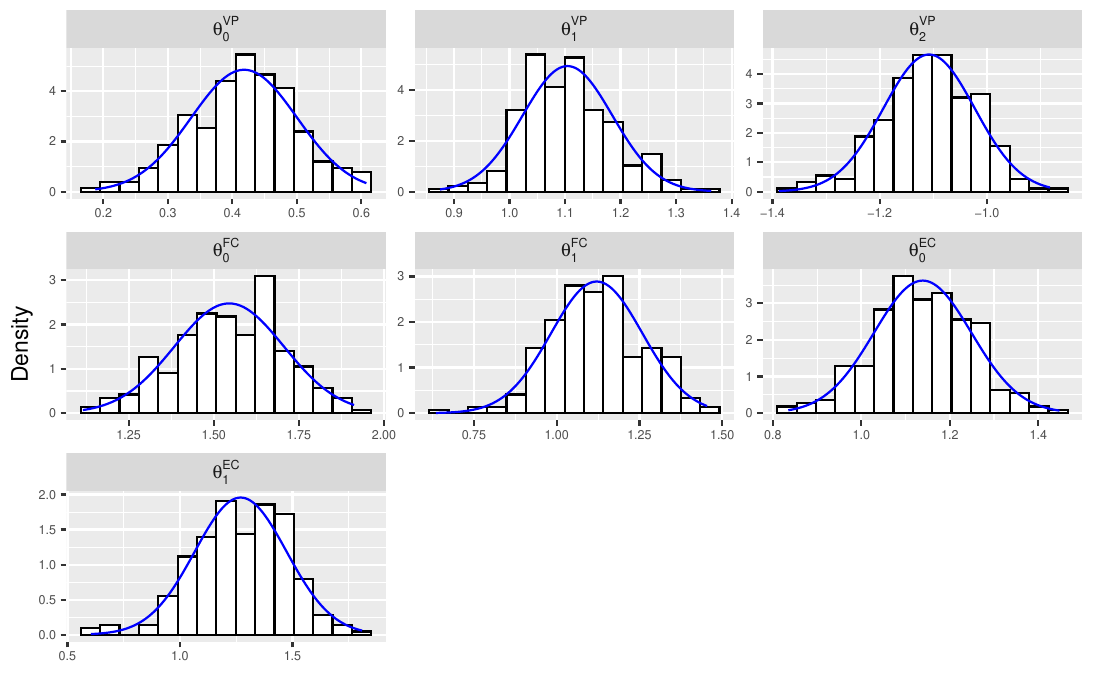}
\par\end{centering}
\begin{raggedright}
{\scriptsize{}Note: Histograms denote the finite sample distribution,
and blue line is normal density}{\scriptsize\par}
\par\end{raggedright}
\caption{Finite sample distributions under AVI without locally robust corrections\label{fig:distributions _AVI_non_locally_robust}.}
\end{figure}

\begin{figure}[t]
\begin{centering}
\includegraphics[width=12cm]{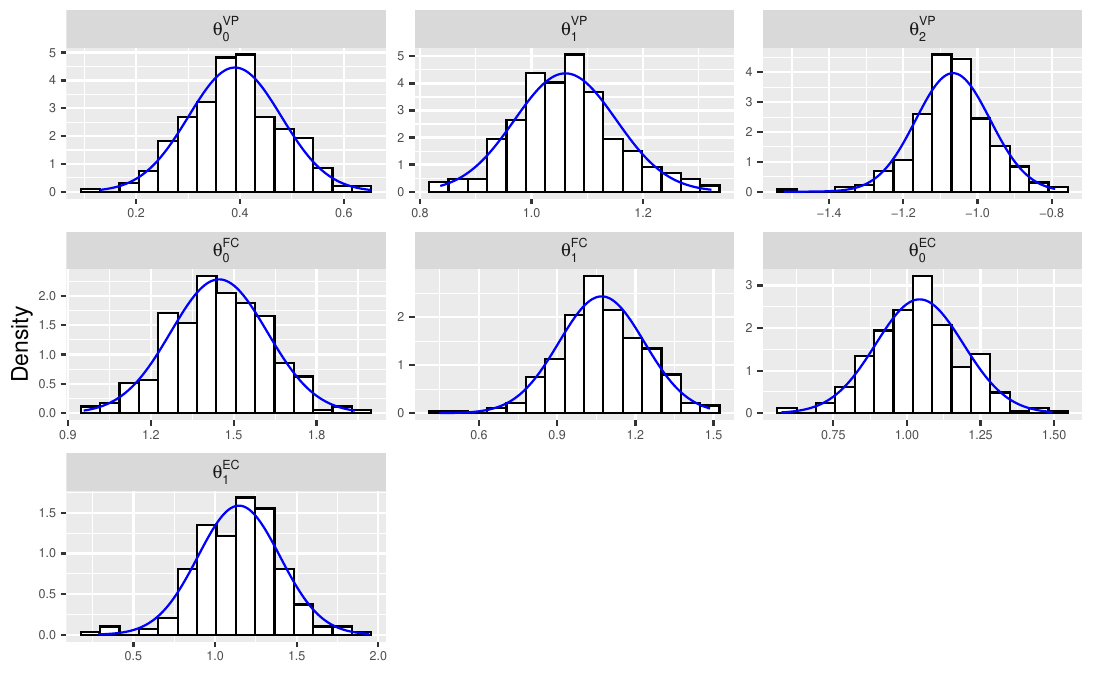}
\par\end{centering}
\begin{raggedright}
{\scriptsize{}Note: Histograms denote the finite sample distribution,
and blue line is normal density}{\scriptsize\par}
\par\end{raggedright}
\caption{Finite sample distributions under AVI with locally robust corrections\label{fig:distributions _ AVI_locally_robust}.}
\end{figure}

\subsubsection{Bus engine replacement problem\label{subsec:Bus-Engine-Replacement}}

$\setcounter{table}{0}
\global\long\def\thetable{B.\arabic{table}}%
$

Consider the following version of the \citet{Rust1987} bus engine
replacement problem which is adapted from \citet{ArcidiaconoMiller2011}.
Each period $t=1,...,T;T<\infty$, Harold Zurcher decides whether
to replace the engine of a bus ($a_{t}=0$), or keep it ($a_{t}=1$).
Denote his action by $j\in\{0,1\}$. Each bus is characterized by
a permanent type $s\in\{1,2\}$, and the mileage accumulated since
the last engine replacement $x_{t}\in\{1,2,...\}$. Harold Zurcher
observes both $s$ and $x_{t}$. The econometrician observes mileage
$x_{t}$, and we make different assumptions on the observability of
bus type $s$.

Mileage increases by one unit if the engine is kept in period $t$
and is set to zero if the engine is replaced. The current-period payoff
for keeping the engine is given by $\theta_{0}+\theta_{1}x_{t}+\theta_{2}s+e_{1t}$,
where $\theta^{*}\equiv\{\theta_{0},\theta_{1},\theta_{2}\}$ are
the structural parameters of interest, and $e_{jt}$ is a choice-specific
transitory shock that follows a Type 1 Extreme Value distribution.
The current-period payoff of replacing the engine is normalized to
$e_{0t}$, and the discount factor is set to 0.9 in this application. 

To carry out the simulations, we choose values for the structural
parameters ($\theta_{0}=2,\theta_{1}=-0.15,\theta_{2}=1$), and recursively
derive the value functions for each possible combination of $x$,
$s$ and $t$. We then use these to compute the conditional replacement
probabilities for the same set of combinations of variables. 

We provide results using the linear semi-gradient method to approximate
$h(a,x)$ and $g(a,x)$. Our first set of results treats the bus type
$s$ as known to the econometrician. We then provide findings for
a setting with permanent unobserved heterogeneity. All results are
based on $1000$ simulations with $1000$ buses and $T=30$ time periods
each. Each round of the simulations begins by generating data for
$1000$ buses and $2000$ time periods. The mileage of each bus is
set to zero in $t=0$. We then simulate the choices $a_{t}$ using
the conditional replacement probabilities. Finally, we restrict the
generated data to $30$ time periods between $t=1000$ and $t=1030$.
This is to ensure that our data is close to being drawn from a stationary
model. In all simulations, we parameterize $h(a,x)$ and $g(a,x)$
using a third order polynomial in $s$, $x_{t}$ and $a_{t}$.\footnote{In total, there are $k_{\phi}=k_{r}=16$ terms. These include a constant,
the binary variables $s$ and $a_{t}$, all $x_{t}$ terms up to a
third order, and pairwise and triple interactions between the $x_{t}$
terms and the binary variables.} The choice probabilities $\eta$ are estimated using a logit model
that is a function of the state variables $s$ and $x_{t}$, where
the same polynomial is used as before.\footnote{We run our simulations on a MacBook Pro with an M1 chip and 16 GB
of RAM. The approximate computation times for one estimation round
is 4 seconds without locally robust correction, and 14 seconds with
locally robust correction.} 

Table \ref{tab:table 1-1} shows the results treating bus type $s$
as known, with and without locally robust correction. It can be seen
that our estimator produces parameter estimates that are closely centered
around the true values. These results are comparable to those found
by \citet{ArcidiaconoMiller2011} in a similar version of the bus
engine replacement problem. However, in contrast to their CCP method,
our estimator does not exploit a finite dependence property. When
comparing the results from our locally robust estimator in column
(4) to the results from the suboptimal estimator in column $(2)$,
it can be seen that the absolute bias is smaller for all three parameter
estimates. However the variance of the locally robust estimator is
higher due to the sample splitting employed in the locally robust
procedure (we used two-fold cross-fitting). Overall, while in theory
the locally robust estimator is preferable to the non-robust version,
we find that in practice there is very little difference between the
two versions of the algorithm. 

\begin{table}

\caption{Simulations: Bus engine replacement problem}
\label{tab:table 1-1}

\begin{tabular}{cccccccc}
\hline 
 &  &  & \multicolumn{1}{c}{} &  &  & \multicolumn{1}{c}{} & \tabularnewline
 &  &  & \multicolumn{2}{c}{not locally robust} &  & \multicolumn{2}{c}{locally robust}\tabularnewline
\cline{4-5} \cline{5-5} \cline{7-8} \cline{8-8} 
 &  &  &  &  &  &  & \tabularnewline
 & DGP &  & TDL & MSE &  & TDL & MSE\tabularnewline
 & (1) &  & (2) & (3) &  & (4) & (5)\tabularnewline
\hline 
\emph{Linear semi-gradient} &  &  &  &  &  &  & \tabularnewline
 &  &  &  &  &  &  & \tabularnewline
$\theta_{0}$ (intercept) & 2.0 &  & 1.9788 & 0.0080 &  & 1.9778 & 0.0081\tabularnewline
 &  &  & (0.0868) &  &  & (0.0870) & \tabularnewline
$\theta_{1}$ (mileage) & -0.15 &  & -0.1492 & 1.2e-05 &  & -0.1489 & 1.3e-05\tabularnewline
 &  &  & (0.0033) &  &  & (0.0034) & \tabularnewline
$\theta_{2}$ (bus type) & 1.0 &  & 1.0044 & 0.0034 &  & 1.0032 & 0.0034\tabularnewline
 &  &  & (0.0583) &  &  & (0.0584) & \tabularnewline
\hline 
\end{tabular}

\footnotesize

\noindent\begin{minipage}[t]{1\columnwidth}%
Notes: The table reports results for 1000 simulations. Column (1)
shows the true parameter values in the model. Columns (2) and (4)
report the empirical mean and standard deviations for the estimated
parameters. Columns (2)-(3) are based on the estimation method without
correction function, columns (4)-(5) report results for the locally
robust estimator. The mean squared error are reported in columns (3)
and (5), respectively.%
\end{minipage}
\end{table}

In a final set of simulations, we introduce permanent unobserved heterogeneity
into our setting by assuming that the permanent bus type $s\in\{1,2\}$
is unknown to the researcher. To generate results for these simulations,
we follow the steps outlined in Section \ref{sec:Incorporating-permanent-unobserv}
where we pair our techniques with a sequential EM algorithm \citep{arcidiacono2003finite},
without using a locally robust correction.\footnote{We run our simulations on a MacBook Pro with an M1 chip and 16 GB
of RAM. The approximate computation times is 305 seconds for one estimation
round.} The results are shown in Table \ref{tab: table 2-1}. Once again,
our algorithm produces parameter estimates that are closely centered
around the true values. Compared to the results without permanent
unobserved heterogeneity, the standard deviation of our estimates
is slightly higher due to the uncertainty around the bus type $s$. 

\begin{table}
\caption{Simulations: Bus engine replacement problem with unobserved heterogeneity}

\label{tab: table 2-1}

\begin{tabular}{ccccc}
\hline 
 &  &  &  & \tabularnewline
 & DGP &  & TDL & MSE\tabularnewline
 & (1) &  & (2) & (3)\tabularnewline
\hline 
\emph{Linear semi-gradient} &  &  &  & \tabularnewline
 &  &  &  & \tabularnewline
$\theta_{0}$ (intercept) & 2.0 &  & 1.9750 & 0.0164\tabularnewline
 &  &  & (0.1255) & \tabularnewline
$\theta_{1}$ (mileage) & -0.15 &  & -0.1492 & 1.5e-05\tabularnewline
 &  &  & (0.0039) & \tabularnewline
$\theta_{2}$ (bus type) & 1.0 &  & 1.0035 & 0.0104\tabularnewline
 &  &  & (0.1018) & \tabularnewline
\hline 
\end{tabular}

\footnotesize

\noindent\begin{minipage}[t]{1\columnwidth}%
Notes: The table reports results for 1000 simulations. Column (1)
shows the true parameter values in the model. Column (2) reports the
empirical mean and standard deviations for the estimated parameters.
The mean squared errors are reported in column (3). The results are
based on the estimation method without correction function.%
\end{minipage}
\end{table}

\end{document}